\documentclass[12pt]{article}
\usepackage{amsmath,bm}
\usepackage{graphicx,psfrag,epsf}
\usepackage{enumerate}
\usepackage{natbib}
\usepackage{url} 
\usepackage{mathtools}

\newtagform{brackets}{[}{]}
\usetagform{brackets}

\usepackage{setspace}
 \DeclareGraphicsExtensions{.eps, .ps}
\usepackage{soul,color}
\usepackage{amsmath}
\usepackage{amsfonts}
\usepackage{amssymb}
\usepackage{newfloat}
\usepackage{multirow}
\usepackage{arydshln}
\usepackage{comment}
\usepackage{enumitem}
\usepackage{amsthm}
\usepackage[title]{appendix}
\usepackage{lineno}
\usepackage{booktabs}
\usepackage{tikz}
\usepackage{comment}
\usepackage{xr}
\usetikzlibrary{positioning}
\usepackage{subcaption}
\usepackage{authblk}
\usepackage{adjustbox}
\usepackage{threeparttable,threeparttablex}

\usepackage{xcolor}
\usepackage{hyperref}
\usepackage{algorithm}
\usepackage{algorithmic}

\newcommand{\blind}{1}

\usepackage[margin=1in]{geometry}

\newtheorem{theorem}{Theorem}
\newtheorem{proposition}{Proposition} 

\newtheorem{corollary}{Corollary}
\newtheorem{lemma}{Lemma}

\newtheorem{remark}{Remark}

\def\P{\mathbb{P}}

\begin{document}

\def\spacingset#1{\renewcommand{\baselinestretch}%
{#1}\small\normalsize} \spacingset{1}


\if1\blind
{
  \title{\bf A robust and powerful method for assessing replicability of high dimensional data} 
  \author[a]{Haochen Lei}
    \author[b]{Yan Li}
    \author[a]{Hongyuan Cao}
\affil[a]{Department of Statistics, Florida State University}
\affil[b]{School of Mathematics, Jilin University, Changchun, China}
  \maketitle
} \fi

\if0\blind
{
  \bigskip
  \bigskip
  \bigskip
  \begin{center}
    {\LARGE\bf A robust and powerful method for assessing replicability analysis in high dimensional data}
\end{center}
  \medskip
} \fi

\newcommand{\tval}{0} 

\newcommand{\Mysubsection}[1]{%
        \ifnum\tval=0
        \subsection{#1}%
        \else
        \subsection*{#1}%
        \fi
}

\date{}
\maketitle
\begin{abstract}

Identifying signals that replicate across multiple studies is essential for establishing robust scientific evidence. However, existing methods for high-dimensional replicability analysis either rely on restrictive modeling assumptions, are limited to two-study settings, or suffer from low statistical power. We propose a general empirical Bayes framework for multi-study replicability analysis that jointly models summary-level 
$p$-values while explicitly accounting for between-study heterogeneity. Within each study, we estimate the nonparametric density function of non-null $p$-values under monotonicity constraints, enabling flexible and tuning-parmeter-free inference. For the two study setting, we construct a local false discovery rate (Lfdr) statistic for the composite null hypothesis of non-replicability and obtain the maximum likelihood estimators. 
We establish cubic-rate convergence and minimax optimality of the nonparametric maximum likelihood estimator for the joint p-value density, along with identifiability and consistency of the parametric components. We further prove asymptotic false discovery rate (FDR) control under mild regularity conditions and demonstrate that Lfdr-based thresholding achieves optimal power.
Extending replicability analysis to $n$ studies generally requires estimating $2^n$ latent configurations, which is computationally infeasible. To overcome this challenge, we introduce a scalable pairwise rejection strategy that decomposes the exponentially large composite null into disjoint components, yielding computational complexity that scales quadratically with the number of studies. We establish valid asymptotic FDR control for this procedure. Extensive simulations demonstrate that our method maintains valid FDR control and outperforms state-of-the-art alternatives across a wide range of scenarios. Applying our framework to East Asian– and European–ancestry genome-wide association studies (GWASs) of type 2 diabetes uncovers replicable genetic associations that competing approaches fail to detect, illustrating the practical utility of our method in large-scale biomedical research.

\end{abstract}
\section{Introduction}

Improving statistical power while maintaining robustness and reproducibility is a central challenge in high-dimensional inference. Replicability, 
the consistency of findings across independent studies, populations, or experimental settings, is a cornerstone of scientific credibility. The ``replicability crisis" has sparked intense debate, with estimates suggesting that over 50\% of preclinical research is irreproducible \citep{ioannidis2005most}, costing approximately \$28 billion annually in the United States alone \citep{freedman2015economics}. While irreproducibility stems from diverse factors, we focus on the statistical analysis of high-throughput data. Specifically, we investigate conceptual replicability, where consistent signals are detected across distinct populations. Identifying such signals helps rule out spurious associations driven by unmeasured confounding, thereby reinforcing scientific conclusions. 



In large-scale genomic and biomedical studies, thousands or millions of genes or genetic variants are tested for association with a phenotpye or trait, and significant findings are prioritized for downstream investigation. 
Our goal is to identify genes or variants whose associations replicate across multiple studies. In a single-study analysis, the  Benjamini–Hochberg (BH) procedure \citep{benjamini1995controlling} is standard for false discovery rate (FDR) control. Extending valid FDR control to multi-study replicability analysis is, however, challenging because the replicability null hypothesis is composite: a feature is non-replicable if it is null in at least one study.

A common but \textit{ad hoc} strategy applies the BH procedure separately to each study and declares overlapping discoveries as replicable. This approach fails to control FDR under the composite null, and has  
low power because it neglects the shared information across studies \citep{bogomolov2022replicability}. The MaxP method \citep{benjamini2009selective}, which applies BH procedure to the maximum of $p$-values across studies, guarantees valid FDR control but is overly conservative. 
To improve power, \citet{lyu2023jump} incorporates composition-null enumeration and proportion estimation, 
but lacks theoretical guarantees.  Cross-screening procedure \citep{bogomolov2018assessing} uses prescreening to borrow information, yet it cannot be scaled to multiple studies. 
Furthermore, 
the irreproducible discovery rate \citep{li2011measuring} and 
ranking-based methods 
\citep{philtron2018maximum} either ignore between-study heterogeneity or assume a common signal configuration across different studies. While Bayesian and empirical Bayes alternatives have also been proposed \citep{heller2014replicability, zhao2020quantify}, 
they require tuning parameters and parametric modeling assumptions. Additional conceptual discussions and related developments appear in \cite{hung2020statistical, bogomolov2013discovering}. 
A recent review is provided by \cite{bogomolov2022replicability}.

In this paper, we propose a robust and powerful empirical Bayes framework for multi-study replicability analysis that overcomes these limitations. 
Our method jointly models summary-level $p$-values across studies while explicitly accounting for between-study heterogeneity. We use a latent variable to encode the configuration of each genetic variant  
across different studies. The composite null comprises all configurations in which the genetic variant is null in at least one study. To accommodate heterogeneity, we allow the non-null $p$-value densities to vary across studies and estimate them nonparametrically under a monotonicity constraint \citep{grenander1956theory}, yielding a robust and tuning-parameter-free procedure.

Our test statistic is the local false discovery rate (Lfdr), defined as the posterior probability that a genetic variant is non-replicable. Unknown parameters and density functions are estimated via an EM algorithm \citep{dempster1977maximum} combined with the pool-adjacent-violators algorithm (PAVA) \citep{robertson1988order}. Using empirical process theory \citep{van1993hellinger, van2000empirical}, we establish identifiability of all model components and prove consistency and cubic-rate convergence of the nonparametric maximum likelihood estimator, along with its minimax optimality. 
We further show that Lfdr-based thresholding is asymptotically power-optimal and that our procedure asymptotically controls the FDR. 

A direct extension of this framework to $n$ studies requires estimating $2^n$ parameters, which scales poorly. To overcome this scalability barrier, we introduce a pairwise rejection strategy. By decomposing the exponentially growing composite null space into disjoint subsets indexed by the number of null studies, we achieve computational complexity that scales quadratically in $n$. 
We develop a general algorithm for $n$-study replicability analysis and establish its asymptotic validity. 
Extensive simulations confirm the theoretical guarantees and demonstrate substantial power gains over existing methods.
Finally, we apply our method to genome-wide association studies (GWAS) of type 2 diabetes (T2D) across European and East Asian cohorts. Our approach identifies replicable SNPs that competing procedures fail to detect, illustrating its practical utility in large-scale genomic research.

\section{Two study case}
\label{sec:method2}
\subsection{Notation and problem setup}
We consider the simultaneous testing of $m$ hypotheses based on data from $n$ independent studies, where each study provides a $p$-value for every hypothesis. 
For ease of exposition, we focus on the case of two studies using single nucleotide polymorphisms (SNPs) as illustrative examples; the extension to 
$n>2$ studies is presented in Section \ref{sec:extend}. 

Let $\theta_{ij}, i = 1, \ldots, m, j=1,2$ denote the unobserved latent state of hypothesis $i$ in study $j$, where $\theta_{ij}=1$ indicates a non-null effect and $\theta_{ij}=0$ indicates null. Conditional 
on these latent states, the $p$-values are assumed to follow the two group mixture model \citep{efron2012large}:
\begin{equation}\label{mixture0}
\begin{aligned}&p_{i1}|\theta_{i1} \sim\left(1-\theta_{i1}\right) f_0+\theta_{i1} f_1, i=1, \ldots, m,\\&p_{i 2}|\theta_{i 2} \sim\left(1-\theta_{i 2}\right) f_0+\theta_{i 2} f_2, i=1, \ldots, m, \end{aligned}
\end{equation}
where $f_0$ denotes the null density of $p$-values, and $f_1$ and $f_2$ denote the non-null densities  
for studies 1 and 2, respectively. This formulation accommodates heterogeneity in both the latent configurations $(\theta_{i1},\theta_{i2})$ and the non-null distributions of $p$-values, $f_1$ and $f_2$, across studies. 

For $i = 1, \ldots, m,$ define the joint latent configuaration $\tau_i = (\theta_{i1}, \theta_{i2})$ with
$$P\left(\tau_i=(u, v)\right)=\xi_{u, v}, \quad  u,v\in\{0,1\},
$$
where
$$
\tau_i= \begin{cases}(0,0) & \text { if } \theta_{i1}=\theta_{i 2}=0, \\ (0,1) & \text { if } \theta_{i1}=0, \theta_{i 2}=1, \\ (1,0) & \text { if } \theta_{i1}=1, \theta_{i 2}=0, \\ (1,1) & \text { if } \theta_{i1}=\theta_{i2}=1, \quad i=1, \ldots, m.\end{cases}
$$
The replicability null hypothesis is defined as 
\begin{equation}H_{0i}: \theta_{i1}\theta_{i2}=0,i=1,\cdots, m,  \label{null}
\end{equation}
which is composite and consists of three configurations $(0,0), (0,1)$ and $(1,0).$ A SNP is declared replicable only if 
$\theta_{i1}=\theta_{i2}=1$. This defintion provides strong evidence, as a spurious signal appearing in only one 
study will not be deemed replicable under (\ref{null}).

We define the local false discovery rate (Lfdr) as the posterior probability that hypothesis $i$ is non-replicable given the observed data \citep{efron2012large}:
\begin{align}
\operatorname{Lfdr}(p_{i1}, p_{i2}) &:= P\left(\theta_{i1}\theta_{i 2}=0 \mid p_{i1}, p_{i 2}\right) \label{lfdr}
\\&=1-P\left(\theta_{i1}=\theta_{i 2}=1 \mid p_{i1}, p_{i 2}\right) \nonumber\\
& =\frac{\xi_{00} f_0\left(p_{i1}\right) f_0\left(p_{i 2}\right)+\xi_{01} f_0\left(p_{i1}\right) f_2\left(p_{i 2}\right)+\xi_{10} f_1\left(p_{i1}\right) f_0\left(p_{i 2}\right)}{\xi_{00} f_0\left(p_{i1}\right) f_0\left(p_{i 2}\right)+\xi_{01} f_0\left(p_{i1}\right) f_2\left(p_{i 2}\right)+\xi_{10} f_1\left(p_{i1}\right) f_0\left(p_{i 2}\right)+\xi_{11} f_1\left(p_{i1}\right) f_2\left(p_{i 2}\right)}.\nonumber
\end{align}
In (\ref{lfdr}), the numerator enumerates the three components of the composite null, while the denominator corresponds to the full four-component mixture. The null density $f_0$ is typically taken to be the standard uniform distribution. To enhance robustness, we estimate $f_1$ and $f_2$ nonparametrically under the monotonicity constraint \citep{sun2007oracle, cao2013optimal}:
\begin{eqnarray}
    \label{monotone}
    f_1/f_0 \mbox{ and } f_2/f_0 \mbox{ are non-increasing}.
\end{eqnarray}
Under (\ref{monotone}), $\mbox{Lfdr}$ is coordinate-wise non-decreasing; that is, if $x_1 \le x_2$ and $y_1 \le y_2,$ then $\mbox{Lfdr}(x_1,y_1)\le\mbox{Lfdr}(x_2,y_2).$

The $\operatorname{Lfdr}$ in (\ref{lfdr}) serves as the test statistic for the replicability null hypothesis in (\ref{null}). For hypothesis $i$, 
we reject $H_{0i}$ accoreding to 
\begin{equation}\label{rejection}
1\left\{\operatorname{Lfdr}(p_{i1},p_{i2})\le \lambda \right\},
\end{equation}
where $\lambda$ is a threshold to be specified. 
The total number of rejections is $$R_{m}(\lambda)=\sum\limits_{i=1}^m 1\{\operatorname{Lfdr}(p_{i1},p_{i2})\le \lambda\},$$
and the number of false rejections is
$$V_m(\lambda)=\sum\limits_{i=1}^m 1\{\operatorname{Lfdr}(p_{i1},p_{i2})\le \lambda\} (1-\theta_{i1}\theta_{i2}).$$
The FDR at threshold $\lambda$ is defined as
\begin{equation}
\label{fdr}
\mbox{FDR}(\lambda)=E\left\{\frac{V_m(\lambda)}{1\vee R_{m}(\lambda)}\right\} = E\left\{ \mbox{FDP}_m(\lambda) \right\},
\end{equation}
where the false discovery proportion is 
$$\mbox{FDP}_m(\lambda)=\frac{\sum_{i=1}^m 1\left\{\operatorname{Lfdr}(p_{i1},p_{i 2}) \leq \lambda\right\}(1-\theta_{i1}\theta_{i2})}{\sum_{i=1}^m 1\left\{\operatorname{Lfdr}(p_{i1},p_{i 2}) \leq \lambda\right\}}.$$
Our objective is to select $\lambda$ such that  $\mathrm{FDR}(\lambda) \le \alpha$ for a prespecified 
nominal level $\alpha$.


\subsection{Parameter estimation and FDR control algorithmt}
Assume the null density is standard uniform, i.e., $f_0(x)=1\{ 0\le x \le 1\}.$ 
Let the unknown parameters and functions be 
$$w=(\xi_{00},\xi_{10},\xi_{01},\xi_{11},f_1,f_2),$$ 
which belong to the semiparametric space $\Omega$ defined by 
\begin{equation}
    \label{par_space}
    \Omega=\Delta\times \mathcal G\times \mathcal G,
\end{equation}
where $$\Delta=\{(\xi_{00},\xi_{10},\xi_{01},\xi_{11})\in R^4_+:\xi_{00}+\xi_{10}+\xi_{01}+\xi_{11}=1\},$$
and 
\begin{equation}
    \label{G1}
    \mathcal G=\{f: \text{ non-increasing density function supported on } (0,1)\}.
\end{equation}
Under this model, the joint density of $(p_1, p_2)$ values is
$$p_w(p_{1},p_2)=\xi_{00}+\xi_{10}f_1(p_1)+\xi_{01}f_2(p_2)+\xi_{11}f_1(p_1)f_2(p_2).$$
The maximum likelihood estimator (MLE) of $w$ is defined as 
\begin{equation}\label{mle}
\hat w=\operatornamewithlimits{argmax}_{w\in \Omega} \frac{1}{m}\sum\limits_{i=1}^m \log p_{w} (p_{i1},p_{i2}). 
\end{equation}
The optimization problem in (\ref{mle}) is solved via the EM algorithm \citep{dempster1977maximum} with the pool-adjacent-violator algorithm (PAVA)\citep{JSSv032i05} incorporated to enforce the monotonicity constraint on $f_1$ and $f_2$. Implementation details are provided in the Supplementary Material \ref{sup:EM}. After obtaining $\hat w$, we apply Algorithm 1 to control the FDR at level $\alpha$.

\begin{algorithm}
\caption{Replicability analysis of two studies}\label{alg1}
\begin{algorithmic}[1]
\REQUIRE For $m$ SNPs and $n=2$ studies, a $p$-value matrix $P\in (0,1)^{m\times 2}$, where $p_{ij}$ indicates the $p$-value of SNP $i$ in study $j,$ and a pre-specified FDR level $\alpha.$
\ENSURE Rejection index $\mathcal I$.

\STATE Compute the MLE 
$\hat w$ in (\ref{mle}) using the EM algorithm combined with PAVA. 

\STATE For each SNP $i=1,\ldots m$, compute the estimated Lfdr 
$$
\widehat{\mbox{Lfdr}}_i=\frac{\hat \xi_{00} +\hat \xi_{10} \hat f_1\left(p_{i1}\right)+\hat \xi_{01} \hat f_2\left(p_{i 2}\right) }{\hat \xi_{00} +\hat \xi_{10} \hat f_1\left(p_{i1}\right)+\hat \xi_{01}\hat f_2\left(p_{i 2}\right)+\hat \xi_{11} \hat f_1\left(p_{i1}\right) \hat f_2\left(p_{i 2}\right)}.
$$

\STATE Apply a step-up procedure to determine 
$$
\hat\lambda_m=\mbox{sup}\left\{\lambda:\frac{\sum\limits_{i=1}^m\widehat{\mbox{Lfdr}}_i1\{\widehat{\mbox{Lfdr}}_i\le \lambda\}}{\sum\limits_{i=1}^m1\{\widehat{\mbox{Lfdr}}_i\le \lambda\}}\le \alpha\right\}.
$$
\STATE 
Reject
\begin{equation}
    \mathcal I=\{i:\widehat{\mbox{Lfdr}}_{i} \le \hat\lambda_m\}.
\end{equation}
\end{algorithmic}
\end{algorithm}

\subsection{Identifiability}
Let $f_1$ denote the 
probability density function of $p_{i1}$ conditional on $\theta_{i1}= 1$ and let $f_2$ denote the 
probability density function of $p_{i2}$ conditional on $\theta_{i2} =1,$ for $i =1, \ldots, m.$ 
Define a bivariate functional space
\begin{equation}\label{eq:mixture}
\mathcal{F}=\left\{p_w(x,y)=\xi_{00}+\xi_{10} f_1(x) +\xi_{01} f_2(y)+\xi_{11} f_1(x) f_2(y): w\in \Omega, x,y\in (0,1) \right\}.
\end{equation}
Let $w=(\xi_{00},\xi_{10},\xi_{01},\xi_{11},f_1,f_2)$ and $w^*=(\xi^*_{00},\xi^*_{10},\xi^*_{01},\xi^*_{11},f_1^*,f_2^*)$.
Define a metric $\tau$ on $\Omega$ by 
\begin{equation}
\label{def:tau}
\tau(w,w^*)=|\xi_{00}-\xi_{00}^*|+|\xi_{10}-\xi_{10}^*|+|\xi_{01}-\xi_{01}^*|+|\xi_{11}-\xi_{11}^*|+d(f_1,f_1^*)+d(f_2,f_2^*),
\end{equation}
where $d(f,g)=\int _0^1|f(x)-g(x)|dx$.
The Hellinger distance between $p_{w}$ and $p_{w^*}$ is defined as  
\begin{equation}\label{h-dist}
h(p_{w},p_{w^*})=\left [\frac{1}{2}\int_0^1\int_0^1 \left\{p_{w}(x,y)^{1/2}-p_{w^*}(x,y)^{1/2}\right\}^2 dx dy\right ]^{1/2}.
\end{equation}
The following proposition establishes identifiability of the model: 
if $h(p_{w},p_{w^*})=0$, then $\tau(w,w^*)=0$. Consequently, the representation
\begin{equation*}
p_w(x,y)=\xi_{00}+\xi_{10} f_1\left(x\right)+\xi_{01} f_2\left(y\right) +\xi_{11} f_1\left(x\right) f_2\left(y\right)
\end{equation*}
is unique.

\begin{proposition} (Identifiability)
\label{identify}
Let $w=(\xi_{00},\xi_{10},\xi_{01},\xi_{11},f_1,f_2)$, and $w^*=(\xi_{00}^*,\xi_{10}^*,\xi_{01}^*,\xi_{11}^*, f_1^*,f_2^*)$.
Assume $$\lim\limits_{x\to 1}f_1(x)=\lim\limits_{y\to 1}f_2(y)=0,\quad \mbox{and} \quad \lim\limits_{x\to 1}f^*_1(x)=\lim\limits_{y\to 1}f^*_2(y)=0.$$ 
If $p_w(x,y)=p_{w^*}(x,y)$ almost everywhere (a.e.) on $(0, 1)^2$, then 
$$\xi_{00}=\xi_{00}^*,\quad\xi_{10}=\xi_{10}^*,\quad\xi_{01}=\xi_{01}^*,\quad\xi_{11}=\xi_{11}^*,$$ and $$f_1(x)=f_1^*(x) \quad a.e.,\quad f_2(x)=f_2^*(x) \quad a.e.$$ 
\end{proposition}

The proof of Proposition \ref{identify} is provided in Supplementary material \ref{sec:id_proof}. We summarize the identifiability requirement as the following condition: 

(C1) For all $w=(\xi_{00},\xi_{10},\xi_{01},\xi_{11},f_1,f_2)\in \Omega$, 
$\lim\limits_{x\to 1}f_1(x)=\lim\limits_{y\to 1}f_2(y)=0$.

\subsection{Consistency} 
We apply empirical process theory \citep{van2000empirical,wellner1996weak} to study the asymptotic behavior of the MLE defined in (\ref{mle}).
Recall that $\mathcal F$ is defined in (\ref{eq:mixture}). 
Let $a\ge 1$ and $f\in\mathcal F$, define
$$\|f\|_a = \left(\int_0^1 |f(x)|^a dx\right)^{1/a}.$$
Let $\{(f_k^L,f^U_k)\}_{k=1}^N$ be pairs of functions such that 
$\|f_k^U-f_k^L\|_a\le \epsilon,\forall k$ and for every $f\in \mathcal F$ there exists $k$ satisfying $f^L_k\le f\le f^U_k.$ 
Then $[f^L_1,f^U_1],\ldots,[f^L_N,f^U_N]$ form an $\epsilon$ bracketing of $\mathcal F$. The minimal number of such 
brackets is the bracketing number, denoted by $N_{[]}(\epsilon,\mathcal F,\|\cdot\|_a),$ and its logarithm 
is called the bracketing entropy. 

Let $w_0 \in \Omega$ denote the true parameter. 
Define
$$\mathcal{\bar F}^{1/2}=\{\sqrt{\frac{p_w+p_{w_0}}{2}},w\in \Omega\}.$$
To derive the convergence rate of the MLE under the Hellinger distance defined in (\ref{h-dist}), we study the bracketing entropy of $\mathcal {\bar F}^{1/2}$ under the $L_2$ metric $\|\cdot\|_2.$
To bound $\log N_{[]}(\epsilon,\mathcal{\bar F}^{1/2}, \|\cdot\|_2),$ we impose the following assumption. 

(C2) There exists an envelope 
function $F(x)>0, x\in(0,1)$ such that $F\in L_{2+\epsilon'}$ \footnote{The condition $F\in L_{2+\epsilon'}$ can be weakened using the Lorentz norm. See Supplementary Material \ref{sec:bracket_proof_s1} for details.} for some $\epsilon'>0$. 
Moreover, for all $w=(\xi_{00},\xi_{10},\xi_{01},\xi_{11},f_1,f_2)\in \Omega$, we have $f_1(x)<F(x),f_2(x)<F(x)$ and $\xi_{00}\ge c_\xi,$ where $c_\xi>0$ is a constant.



The following proposition establishes the bracketing entropy of $\bar{\mathcal{F}}^{1/2}$ under 
\begin{proposition}
    \label{bracketing_2}
    Under (C1) and (C2), there exist $C_1>0$ and $C_2>0$ such that $$C_2\epsilon^{-1}\le\log N_{[]}(\epsilon,\mathcal{\bar F}^{1/2},\|\cdot\|_2)\le C_1\epsilon^{-1}.$$ 
\end{proposition}

\begin{remark} 
For the space $\mathcal G$ defined in (\ref{G1}), 
the envelope function may diverge near zero. Without additional assumptions, the bracketing number of $\mathcal G$ may be 
infinite 
\citep{gao2009rate}. Condition (C2) ensures that the envelope has a finite $L_{2+\epsilon'}$ norm, which guarantees finite bracketing entropy even when the envelope is unbounded. 
\end{remark}

Proposition \ref{bracketing_2} characterizes 
the 
bracketing entropy and thus enables us 
to derive the convergence rate of the Hellinger distance between $p_{\hat w_n}$ and $p_{w_0}.$ The proof is deferred to 
Supplementary Material \ref{sec:bracket_proof_s1}.

We now state the main result.
\begin{theorem}
\label{main_consistency}
Under (C1) and (C2), suppose $\mathcal F$ is the space defined in (\ref{eq:mixture}). Let $p_{w_0}\in \mathcal F$ denote the true density, $\hat w$ the MLE defined in (\ref{mle}), and $p_{\hat w}$ the corresponding density. 
Then there 
exist constants $C_3,C_4,C_5>0$ such that 
$$
P\left(h\left(p_{\hat{w}}, p_{w_0}\right)>C_3m^{-1/3} \right)\leq C_4 \exp \left[-C_5m^{1/3}\right].
$$
Consequently,
$$
h(p_{\hat {w}},p_{w_0})=O_p(m^{-1/3}).
$$
\end{theorem}
The proof is provided in Supplementary Material \ref{sec:bracket_proof}.

In \cite{gao2007entropy}, 
the bracketing entropy under $\|\cdot\|_2$ was derived for uniformly bounded monotone functions supported on $[0,1]^2.$ 
The bracketing entropy has order $\epsilon^{-2}\{\log (1/\epsilon)\}^{2},$ resulting in an $m^{-\frac{1}{4}}\log m$ rate of convergence for the MLE under the Hellinger distance. In contrast, we do not impose a boundedness assumption on the monotone non-increasing function class defined in (\ref{G1}). 
Furthermore, for log-concave density functions, 
the bracketing entropy has order $\epsilon^{-1} \{\log(1 / \epsilon)\}^{3/2},$ which implies an $m^{-1/ 3} (\log m)^{1/2}$ rate of convergence for the MLE under the Hellinger distance \citep{kim2016global}. Our function class differs fundamentally from the log-concave densities. For example, the $f(x)=\frac{4}{\pi}\frac{1}{1+x^2},0\le x\le 1$ is a monotone non-increasing density function on $[0,1],$ but it is not log-concave. 


The following corollary establishes consistency of $\hat{w}$ defined in (\ref{mle}). The proof is given in the Supplementary Material \ref{sec: consitency proof}.
\begin{corollary}
\label{parameter consistency us}
    Under (C1) and (C2),
    we have $\tau(\hat w,w_0)\overset{p}{\to} 0,$
    where $\tau$ is the matrix on $\Omega$ defined in (\ref{def:tau}), $\hat w$ is the MLE, $w_0$ is the true parameter, and $\overset{p}{\to}$ means convergence in probability.
\end{corollary}
\subsection{Optimal minimax rate}
Let $\mathcal F$ denote the functional space defined in (\ref{eq:mixture}), and let $p_{w_0}\in \mathcal F$ be the true density. 
Assume the observed $p$-values $p_1,\ldots,p_m$ are independent and identically distributed with density 
$p_{w_0}$, where $p_i=(p_{i1},p_{i2}),i=1,\ldots, m$, representing the pair of $p$-values corresponding to the $i$-th hypothesis. 

An estimator of the density is a measurable mapping
\[
\tilde p_m : [0,1]^{2m} \to \mathcal{F},\qquad 
p_{\hat w}=\tilde p_m(p_1,\ldots,p_m).
\]
Let $$\bar{\mathcal F_m}=\{\tilde p_m:\tilde p_m\in R^{2m}\to \mathcal F\}$$ denote the class of all measurable estimators. The minimax 
risk under squared Hellinger loss is defined as 
\begin{equation}
\label{eq:minimax}
    \mathcal R=\inf\limits_{\tilde p_m\in \bar{\mathcal F}_m}\sup\limits_{p_{w_0}\in \mathcal F} E_{p_{w_0}}[h^2\{\tilde p_m(p_1,\ldots, p_m),p_{w_0}\}].
\end{equation}
The quantity $\mathcal R$ represents the optimal worst case risk over $\mathcal F$ with respect to squared Hellinger loss. An estimator is minimax-rate optimal if its risk matches the order of $\mathcal R.$


To determine the exact rate, we establish matching upper and lower bounds. The upper bound follows from Theorem 
\ref{main_consistency}. The lower bound is derived using Theorem 1 %
of \cite{yang1999information}. 
The resulting minimax rate is stated below, which is proved in Supplementary Material \ref{minimax_proof}. 

\begin{theorem}
    \label{Minimax}
    Under the conditions of Theorem \ref{main_consistency}, 
    $$
    \mathcal R 
    \asymp O(m^{-2/3})
    $$ and the MLE in (\ref{mle}) is minimax rate optimal.
\end{theorem}

This result is new in the context of high dimensional replicability analysis with nonparametric estimation of $p$-value densities. 
In \cite{biau2003risk}, it is shown that
for two-dimensional bounded monotone densities under $L_1$ loss, the minimax rate is $O(m^{-1/3})$. In contrast, we do not impose a boundedness assumption on the densities. However, tailored to our application, the function space defined in (\ref{par_space}) is more restrictive and structured. For two-dimensional log-concave densities \cite{kim2016global} established a mimimax rate of $O(m^{-2/3})$ under sqaured Helliger loss. The class of log-cancave densities is supported on $\mathbb{R}^2$, whereas in our setting the density functions are supported on the compact domain $[0,1]^2$. Moreover, a log-concave densities necessarily exhibit light tails with exponential decay, while monotone non-increasing densities do not impose such tail restrictions.  


\subsection{Theoretical guarantee of Algorithm \ref{alg1}}
Given the MLE $\hat w=(\hat\xi_{00},\hat\xi_{01},\hat \xi_{10},\hat \xi_{11},\hat f_{1},\hat f_{2}),$ the estimated Lfdr 
is 
 \begin{equation}\label{hatlfdr}
    \widehat{\mbox{Lfdr}}(p_{i1},p_{i2})=\frac{\hat \xi_{00} +\hat\xi_{01} \hat f_{2}\left(p_{i 2}\right)+\hat \xi_{10} \hat f_{1}\left(p_{i1}\right)}{\hat \xi_{00} +\hat\xi_{01} \hat f_{2}\left(p_{i 2}\right)+\hat \xi_{10} \hat f_{1}\left(p_{i1}\right)+\hat \xi_{11} \hat f_{1}\left(p_{i1}\right) \hat f_{2}\left(p_{i 2}\right)}.
    \end{equation}
At threshold $\lambda,$ the estimated number of rejections is
$$
\hat R_m(\lambda)= \sum_{i=1}^m 1\left\{\widehat{\operatorname{Lfdr}}(p_{i1},p_{i2}) \leq \lambda\right\}. 
$$
Note that
\begin{equation}\label{trick}
 E\left[\mbox{Lfdr} (p_{i1}, p_{i2}) 1 \{ \mbox{Lfdr}(p_{i1}, p_{i2}) \le \lambda\} \right]=E[1\{\mbox{Lfdr}(p_{i1},p_{i2})\le \lambda\}(1-\theta_{i1}\theta_{i2})].
\end{equation}
Motivated by (\ref{trick}), 
define the estimated number of false rejections as 
$$
\hat V_{m}(\lambda)= \sum_{i=1}^m \widehat{\operatorname{Lfdr}}(p_{i1},p_{i2}) 1\left\{\widehat{\operatorname{Lfdr}}(p_{i1},p_{i2}) \leq \lambda\right\},
$$
and the estimated false discovery proportion (FDP) as 
$$
\widehat{\mbox{FDP}}_m(\lambda) = \frac{\hat V_{m}(\lambda)}{\hat R_m(\lambda)}.
$$

The data-driven threshold at target FDR level 
 $\alpha$ is
\begin{equation}\label{critical}
\hat\lambda_m=\sup\{\lambda\in[0,1]:\widehat{\mbox{FDP}}_m(\lambda)\le \alpha\}.
\end{equation}
We reject $H_{0i}$ when $\widehat{\mbox{Lfdr}}(p_{i1}, p_{i2})\le \hat{\lambda}_m.$ 
To establish asymptotic FDR control, we impose the following additional conditions. 

(C3) $f_1$ and $f_2$ are continuous, non-increasing density functions on $(0,1)$. 

(C4) $\forall \lambda\in (0,1)$,
    $$
    \begin{gathered}
    \frac{1}{m} \sum_{i=1}^m1\left\{\operatorname{Lfdr}(p_{i1},p_{i2}) \leq \lambda\right\}\overset{p}{\to} B_1(\lambda), \\
    \frac{1}{m} \sum_{i=1}^m \operatorname{Lfdr}(p_{i1},p_{i2}) 1\left\{\operatorname{Lfdr}(p_{i1},p_{i2}) \leq \lambda\right\}\overset{p}{\to} B_2(\lambda), \end{gathered}$$
    and
    \begin{equation}\label{numerator}
    \frac{1}{m} \sum\limits_{i=1}^m 1\{\operatorname{Lfdr}(p_{i1},p_{i2})\le \lambda\} (1-\theta_{i1}\theta_{i2}) \overset{p}{\to} B_2(\lambda),
    \end{equation}
where $B_1(\lambda)$ and $B_2(\lambda)$ are continuous 
on $(0,1)$. 

(C5) There exists $\lambda_0\in(0,1),$ such that $B_1(\lambda_0)>0$, and $\frac{B_2(\lambda_{0})}{B_1(\lambda_{0})} \le \alpha.$

Condition (C3) requires only monotonicity of the non-null $p$-value densities, which is substantially weaker than common parametric assumptions. 
Condition (C4) is analogous to assumptions in \cite{storey2004strong}. In view of (\ref{trick}), convergence in (\ref{numerator}) follows from the law of large numbers. Condition (C5) guarantees the existence of a nontrivial population-level threshold satisfying the FDR constraint. 
\begin{theorem}
\label{fdr_theorem}
Under (\ref{mixture0}) and conditions (C1)-(C5), 
$$\limsup\limits_{m\to\infty}\operatorname{FDR}(\hat \lambda_m)\le \alpha,$$
where $\hat\lambda_m$ is defined in (\ref{critical}) and ${\mbox{FDR}(\lambda)}$ is given in (\ref{fdr}). 
\end{theorem}
The proof is relegated in Supplementary Material \ref{sec: two_lfdr_proof}.

\begin{remark}
We estimate the unknown parameters and density functions via maximum likelihood, assuming independence across hypotheses. In the simulation studies, however, we demonstrate that the proposed  method remains robust under various forms of dependence. We conjecture that similar theoretical guarantees may extend to settings with weak dependence and leave a rigorous investigation of this extension for 
future work. 
\end{remark}
\subsection{Oracle power}
Let $\delta: [0,1]^2\to \{0,1\}$ denote a rejection rule, where $\delta(p_1, p_2) = 1$ indicates rejection when the paired $p$-values are 
$(p_1,p_2)$. We show that thresholding the Lfdr yields the most powerful rule under a marginal FDR constraint. 
\begin{proposition}
   \label{thm:power}
   Among all measurable decision rules $\delta$, the rule that maximizes $$E\{\delta(p_1,p_2)\mid\theta_1\theta_2=1\}$$
   subject to the constraint 
   $$E\{1-\theta_1\theta_2\mid \delta(p_1,p_2)=1\}\le \alpha$$ 
  is given by
  $$ \delta(p_1,p_2)=1\{\operatorname{Lfdr}(p_1,p_2)\le \lambda\}
   $$
    for some threshold $\lambda \in [0,1]$. 
\end{proposition}

The proof of Proposition \ref{thm:power} is relegated in Supplementary Material \ref{sec: oracle_pow_proof}

\section{Multi-study extension}
\label{sec:extend}
When more than two studies are available, a direct extension of the two-study procedure described in Section \ref{sec:method2} becomes computationally infeasible, as the number of mixture components and associated 
parameters grows exponentially with the number of studies. The setup is the same as in Section 2, except that we now consider $n$ studies. For the $i$ the hypothesis, the composite null hypothesis is 
\begin{equation}
\label{eq:composite_null}
   \mathcal H_{0i}=\{\theta_{i1}\theta_{i2}\ldots\theta_{in}=0\}. 
\end{equation}
That is, at least one study corresponds to a null effect. 

Our methodological development begins with an idealized benchmark, an omniscient oracle, to motivate the proposed procedure.



\subsection{Oracle case}\label{oracle}
In the oracle setting, we assume that, for each study, the null and non-null density functions as well as the prior probabilities of the latent configurations are known. For notational simplicity, we suppress the hypothesis index $i$ in this subsection and write $(p_1, \ldots, p_n)$ for the $p$-values of a representative hypothesis. 
Let $\delta\in\{0,1\}$ denote a rejection rule.
When $n=2$, the oracle rule rejects 
if the Lfdr falls below a threshold. For $n\ge 3$, we adopt a pairwise criterion: reject if all pairwise Lfdrs are below their respective thresholds, i.e.,
\begin{equation}
\label{eq:rej}
\delta(\lambda)=1\{\mbox{Lfdr}_{12}\le\lambda_{12},\ldots, \mbox{Lfdr}_{n-1,n}\le \lambda_{n-1,n}\},   
\end{equation}
where $\lambda= (\lambda_{jj'}: 1 \le j<j' \le n)$ 
and $\lambda_{jj'}$ is the rejection threshold for the pair $(j, j').$ 

As shown in \cite{sun2007oracle}, when the number of hypotheses $m\rightarrow \infty,$ the FDR is asymptotically equivalent to the marginal false discovery rate (mFDR), 
 \begin{equation}
     \label{eq:mfdr}\mbox{mFDR}=\frac{\mathbb{E}\!\big[\delta(\lambda)\cdot 1\{\mathcal H_0\}\big]}{\mathbb{E}(\delta(\lambda))}.
  \end{equation}

To obtain an upper bound for mFDR with complexity growing linearly in $n,$ we decompose the composite null space according to the number of null studies. 
Let $\theta=(\theta_1,\ldots, \theta_n)$ and define $h(\theta)=\#\{j:\theta_j=0\}.$ 
The null space $\Theta_0=\{\theta:\exists j, 
\theta_j=0\}$ can be partitioned as 
$$\Theta_0=\Theta_1\cup \Theta_2,$$
where
$$
\begin{aligned}
    &\Theta_1=\{\theta:h(\theta)= 1 \}\quad{\text{and}}\\
    &\Theta_2=\{\theta:h(\theta)\ge 2\}.
\end{aligned}
$$
The set $\Theta_1$ contains conofigurations with exactly one null study, while $\Theta_2$ contains those with at least two null studies. The following proposition provides bounds for each component. 
\begin{proposition}
\label{thm:mfdr_bound}
    Let $f_j$ denote the non-null density function 
    for study $j$, and define the conditional joint density $$f(p_1,\ldots,p_n\mid\theta)=\prod\limits_{j=1}^n f_j^{\theta_j}(p_j).$$ 
    Let $f_{jj'}(p_j,p_{j'})$ denote the joint 
    density of $(p_j,p_{j'}),$ 
    and define prior probabilities 
\begin{equation}
    \label{def:pi}
    \pi_{n+1}=P(h(\theta)\ge 2),\quad\text{ and }\quad \pi_j=P(\theta_j=0,\theta_{-j}=1),j=1,\ldots, n.
\end{equation}
Then 
\begin{align}
   E\{\delta(\lambda)\mid \theta_j=0,\theta_{-j}=1\}&\le \min\limits_{j',j'\ne j}E[\frac{1\{\operatorname{Lfdr}_{jj'}\le \lambda_{jj'} \}f_{j'}(p_{j'}) }{f_{jj'}(p_j,p_{j'})}]. \label{eq:original_tp1}\\
   \max\limits_{\theta\in \Theta_2} E\{\delta(\lambda)\mid \theta\}&\le \max\limits_{j,j',j\ne j'} E[\frac{1\{\operatorname{Lfdr}_{jj'}\le \lambda_{jj'} \} }{f_{jj'}(p_j,p_{j'})}]. \label{eq:original_tp2}
\end{align}
Consequently, 

\begin{align}
&\mbox{mFDR}
\nonumber \\
\le& \frac{1}{E(\delta(\lambda))}\{\pi_{n+1}\max\limits_{j,j',j\ne j'} E[\frac{1\{\operatorname{Lfdr}_{jj'}\le \lambda_{jj'} \} }{f_{jj'}(p_j,p_{j'})}]+\sum\limits_{j=1}^n\pi_j\min\limits_{j',j'\ne j}E[\frac{1\{\operatorname{Lfdr}_{jj'}\le \lambda_{jj'} \}f_{j'}(p_{j'}) }{f_{jj'}(p_j,p_{j'})}]\}\nonumber
\\:=&\mathcal R_1(\lambda) .\label{eq:original_t}
\end{align}
\end{proposition}
 In (\ref{eq:original_tp1}), we consider the specific configuration in which only study $j$ is null. Because the rejection rule requires that all pairwise criteria be satisfied simultaneously, the overall rejection probability is bounded above by the probability corresponding to any individual pair $(j,j')$. To obtain the sharpest upper bound, we therefore take the minimum over all $j' \ne j$. In contrast, in (\ref{eq:original_tp2}), we aim to derive a uniform bound over the entire parameter space $\Theta_2$. In this case, the rejection probability is controlled by the largest of the pairwise contributions, and hence the bound is given by the maximum of the corresponding pairwise terms. 

The proof 
is provided in Supplementary Material \ref{sec: thm:mfdr_bound_proof}. To control mFDR at level $\alpha,$ we seek $\lambda\in R^{n(n-1)/2}$ such that $\mathcal R_1(\lambda)\le \alpha$.

For computational simplicity, 
we impose a common pairwise nominal level 
$\alpha_p.$ For each pair $(j,j'),$ define  
$$\lambda_{jj'}(\alpha_p)=\sup\{\lambda:\frac{E[\mbox{Lfdr}_{jj'} 1\{\mbox{Lfdr}_{jj'}\le \lambda\}]}{E[1\{\mbox{Lfdr}_{jj'}\le\lambda\}]}\le\alpha_p\},$$ which can be implemented via the step-up procedure in Section \ref{sec:method2}. 
Let
$$\lambda(\alpha_p)=\{\lambda_{jj'}(\alpha_p): j,j'=1,\ldots, n, j\ne j'\}.$$ 
We then select 
$$\alpha_p=\sup\{t:\mathcal R_1\{\lambda(t)\}\le \alpha\}, $$ 
and reject according to 
$$\delta=1\{\mbox{Lfdr}_{jj'}\le \lambda_{jj'} (\alpha_p),\forall j\ne j'\}.$$

\subsection{Parameter estimation}
To implement the oracle procedure in Section \ref{oracle}, we need to estimate the prior probabilities 
$\pi_j, j =1, \ldots, n+1$, and the unknown functions appearing in $\mathcal R_1(\lambda),$ including the pairwise Lfdr, joint densities of the $p$-values, and marginal non-null densities. 
For unknown function estimations, 
The functional components can be estimated using the 
methods developed in Section \ref{sec:method2}. 
Here we focus on estimating $\pi_j, j =1, \ldots, n+1.$ 

Let 
$$\pi_0= P(\theta_1 = \dots = \theta_n =1),$$ 
so that $\sum_{j=0}^{n+1}\pi_j=1.$ There are $n+1$ unknown parameters $\pi_0, \pi_1, \ldots, \pi_n,$ and we construct $n+1$ equations to identify them. 

The key observation is that, when a signal replicates across studies, the corresponding $p$-values tend to be simultaneously small. For a small threshold $\tau>0$, the events
\begin{equation}
\Big\{\max_{j=1,\ldots,n} p_j \le \tau\Big\}
\quad\text{and}\quad
\Big\{\max_{j',j'\ne j} p_{j'} \le \tau\Big\},\quad j=1,\ldots,n, \label{eq:event}
\end{equation}
are therefore informative about configurations with at most one null study, i.e., 
$h(\theta)\le 1.$ 

Decomposing the probabilities of the $n+1$ events in (\ref{eq:event}) yields $n+1$ equations. 
Ignoring terms of order $O(\tau^2)$, we obtain the 
approximations
\begin{equation}
\label{eq:approx_all}
\begin{aligned}P(\max\limits_{j=1,\ldots,n}p_{j}\le \tau)&\approx\pi_0P(\max\limits_{j=1,\ldots, n}p_{j}\le \tau \mid\theta_1=\cdots=\theta_n=1)\\&+\tau \sum\limits_{j=1}^n \pi_jP(\max\limits_{j',j'\ne j}p_{j'}\le \tau  \mid \theta_{-j}=1),
\end{aligned}
\end{equation}
and, for $j=1, \ldots, n,$ ignoring terms of order $O(\tau)$, we have
\begin{equation}
    \label{eq:approx_mis1}
P(\max\limits_{j',j'\ne j}p_{j'}\le \tau)\approx(\pi_0+\pi_j)P(\max\limits_{j',j'\ne j}p_{j'}\le \tau\mid \theta_{-j}=1). 
\end{equation}
Solving (\ref{eq:approx_all}) and (\ref{eq:approx_mis1}), yields the estimators 
\begin{equation}
\label{eq:tilde_pi0}
       \tilde \pi_0=\frac{P(\max\limits_{j=1,\ldots, n}p_{j}\le \tau)-\tau\sum\limits_{j=1}^n P(\max\limits_{j',j'\ne j}p_{j'}\le \tau)}{P(\max\limits_{j=1,\ldots, n}p_{j}\le \tau\mid\theta_1=\cdots=\theta_n=1 )-\tau \sum\limits_{j=1}^n P(\max\limits_{j',j'\ne j}p_{j'}\le \tau\mid\theta_{-j}=1)},
\end{equation}
and
\begin{equation}
\label{eq:tilde_pij}
  \tilde\pi_j=\frac{P(\max\limits_{j',j'\ne j}p_{j'}\le \tau)}{P(\max\limits_{j',j'\ne j}p_{j'}\le \tau\mid \theta_{-j}=1)}-\tilde \pi_0, \quad j=1,\ldots, n.  
\end{equation}
We then set 
$\tilde\pi_{n+1}=1-\sum\limits_{j=1}^n \tilde \pi_j-\tilde\pi_0$. Detailed derivations 
are provided in Supplementary Material \ref{sup:proportion}. 

Let $F_j(x) = P(p_j \le x \mid \theta_j = 1)$ and define $$F_{\min}(x)=\min_{1\le j \le n}F_j(x).$$ We impose the following assumption.

(C6) $\tau\le F_{\min}(\tau)$ and $\lim\limits_{\tau\to 0}\frac{F_{\min}(\tau)}{\tau}=+\infty$.

This condition formalizes the requirement that small $p$-values are substantially more likely under the non-null than under the null. 
For example, if $X\sim N(\mu,1)$ and the one-sided $p$-value is $1-\Phi(X)$, where $\Phi(\cdot)$ is the cdf of the standard normal distribution, then (C6) holds; 
see Supplementary Material \ref{sup:C6}. 

For sufficiently small $\tau,$ both the numerator and denominator in (\ref{eq:tilde_pi0}) are strictly positive 
under (C6). 

Substituting $\tilde{\pi}_j, j =1, \ldots, n+1$ into $\mathcal R_1(\lambda)$ yields 
\begin{equation}
\label{eq:R2}
    \mathcal R_2(\lambda,\tau):= \frac{\tilde \pi_{n+1}\max\limits_{j,j',j\ne j'} E(\frac{1\{\mbox{Lfdr}_{jj'}\le \lambda_{jj'} \} }{f_{jj'}(p_j,p_{j'})})+\sum\limits_{j=1}^n\tilde \pi_j\min\limits_{j',j'\ne j}E(\frac{1\{\mbox{Lfdr}_{jj'}\le \lambda_{jj'} \}f_{j'}(p_{j'}) }{f_{jj'}(p_j,p_{j'})})}{E\{\delta(\lambda)\}}.
\end{equation}
For $j=1,\ldots,n$, define
    \begin{equation}
        \label{eq:ej}E_j(\lambda)=\min\limits_{j',j'\ne j}E(\frac{1\{\mbox{Lfdr}_{jj'}\le \lambda_{jj'} \}f_{j'}(p_{j'}) }{f_{jj'}(p_j,p_{j'})})=\min\limits_{j',j'\ne j} E(1\{\mbox{Lfdr}_{jj'}\le \lambda_{jj'} \}\mid\theta_j=0,\theta_{j'}=1),
    \end{equation}
    and
    \begin{equation}
    \label{eq:e0}E_0(\lambda)=\max\limits_{j,j',j\ne j'}E(\frac{1\{\mbox{Lfdr}_{jj'}\le \lambda_{jj'} \}}{f_{jj'}(p_j,p_{j'})})=\max\limits_{j,j',j\ne j'} E(1\{\mbox{Lfdr}_{jj'}\le \lambda_{jj'}\}\mid \theta_j=\theta_{j'}=0). 
    \end{equation}
Let
$$\Lambda:=\{\lambda: \min\limits_{j=1,\ldots, n}E_j(\lambda)\ge E_0(\lambda)\}.$$ 
Intuitively, $\Lambda$ contains thresholds under which rejection is more likely when exactly one study is non-null than when both are null; this set is nonempty 
(see Supplementary Material \ref{sup:lambda}). 

\begin{proposition}
    \label{thm:R_compare}
    Under (C6), for any $\lambda\in \Lambda$, there exists $\tau^{*}>0$ such that for all $\tau<\tau^{*}$, 
    $$
    \mathcal R_2(\lambda,\tau)\ge \mathcal R_1(\lambda).
    $$
\end{proposition}

The proof of Proposition \ref{thm:R_compare} is provided in Supplementary Material \ref{thm:R_compare_proof}.


Thus $\mathcal R_2(\lambda,\tau)$ provides a conservative upper bound for $\mathcal R_1(\lambda)$ when $\tau$ is sufficiently small.
For fixed $\lambda$, if the non-null densities are known or can be consistently estimated, the expectations in 
$\mathcal R_2(\lambda,\tau)$ can be approximated by sample averages via the law of large numbers, enabling practical implementation. 

\subsection{FDR control} 
Estimating $\tilde \pi_0$ in (\ref{eq:tilde_pi0}) and $\tilde \pi_{j}$ for $j =1, \ldots, n$ in (\ref{eq:tilde_pij}) requires evaluating the probabilities 
$$P(\max\limits_{j=1,\ldots,n} p_j \le \tau ),$$ $$
P(\max\limits_{j',j'\ne j} p_{j'} \le \tau ),$$ $$
P(\max\limits_{j=1,\ldots,n} p_j \le \tau \mid \theta_1=\cdots=\theta_n=1),$$ and
$$P(\max\limits_{j',j'\ne j} p_{j'} \le \tau \mid \theta_{-j}=1).$$  The first two probabilities can be consistently estimated using the law of large numbers. The latter two are obtained by plugging in the estimated cumulative distribution functions of the non-null $p$-values, which are constructed using the procedure described in Section \ref{sec:method2}. This yields estimates $\hat{\pi}_j, j =0, \ldots, n+1.$

Using these estimates, the plug-in version of $\mathcal R_2(\lambda, \tau)$ in (\ref{eq:R2}) is given by
\begin{equation}
\label{eq:R_3}
\mathcal R_3(\lambda,\tau)=\frac{\hat\pi_{n+1}\max\limits_{j,j',j\ne j'}\frac{1}{m}\sum\limits_{i=1}^m\frac{1\{\widehat{\operatorname{Lfdr}}_{i,jj'}\le \lambda_{jj'}\}}{\hat{f}_{jj'}(p_{i,j},p_{i,j'})}+\sum\limits_{j=1}^n \hat \pi_j\min\limits_{j',j'\ne j}\frac{1}{m}\sum\limits_{i=1}^m\frac{1\{\widehat{\operatorname{Lfdr}}_{i,jj'}\le \lambda_{jj'}\}\hat{f}_{j'}(p_{i,j'})}{\hat{f}_{jj'}(p_{i,j},p_{i,j'})} }{\frac{1}{m}\sum\limits_{i=1}^m1\{\widehat{\operatorname{Lfdr}}_{i,jj'}\le \lambda_{jj'},\forall j, j'=1,\ldots, n,j\ne j'\}}. 
\end{equation}

The complete procedure is summarized in Algorithm \ref{alg2}. 
\begin{algorithm}
\caption{Replicability analysis for $n$ studies}
\label{alg2}
\begin{algorithmic}[1]
\REQUIRE For $m$ SNPs and $n$ studies, a $p$-value matrix $P\in [0,1]^{m\times n}$, where $p_{ij}$ denotes the $p$-value for SNP $i$ in study $j$, and a pre-specified FDR level $\alpha$.
\ENSURE Rejection index $\mathcal I$

\STATE 
For each SNP $i$ and study pairs $(j,j')$ with $j \ne j'$, compute $\widehat{\mathrm{Lfdr}}_{i,jj'},$ the estimated pairwise Lfdr. Also estimate $\hat{f}_k$ for $k = 1, \dots, n$, where $\hat{f}_k$ denotes the density of the non-null $p$-values for study $k,$ as described in Section \ref{sec:method2}.  

\STATE 
Computer $\hat{\pi}_j, j =0, \ldots, n+1$ using (\ref{eq:tilde_pi0}) and (\ref{eq:tilde_pij}).

\STATE For each pair $(j,j')$ with $j\ne j'$, and for $t\in (0,1)$, define 
\[
\hat{\lambda}_{jj'}(t) = \sup_{\lambda} \left\{
\frac{\frac{1}{m} \sum_{i=1}^m \widehat{\mathrm{Lfdr}}_{i,jj'} \, 1\{ \widehat{\mathrm{Lfdr}}_{i,jj'} \le \lambda \}}
{\frac{1}{m} \sum_{i=1}^m 1\{ \widehat{\mathrm{Lfdr}}_{i,jj'} \le \lambda \}}
\le t
\right\}.
\]
\STATE Let 
\[
\hat\lambda(t)=(\hat\lambda_{12}(t),\ldots,\hat\lambda_{n-1,n}(t))\in R^{n(n-1)/2}.
\] and 
compute $\mathcal R_3(\hat\lambda(t),\tau)$ 
in 
(\ref{eq:R_3}), where $\tau>0$ is a small constant.

\STATE For $\hat\lambda(t)\in \Lambda,$ compute
\begin{equation}
\label{eq:alpha_hat}
    \hat{\alpha}_p = \sup_{t} \left\{ \mathcal{R}_3\big( \hat{\lambda}(t),\tau \big) \le \alpha \right\}.
\end{equation}

\STATE Return 
\begin{equation}
\label{eq:rejection_set}
    \mathcal I=\{i:\ \widehat{\mathrm{Lfdr}}_{i,jj'} \le \hat{\lambda}_{jj'}(\hat{\alpha}_p),\ \forall j\ne j',\ j,j'=1,\ldots, n\}.
\end{equation}
\end{algorithmic}
\end{algorithm}

  In step 5 of Algorithm \ref{alg2}, we require that $\hat{\lambda}(t) \in \Lambda.$ We next describe how to transform any threshold vector $\lambda$ 
   lying outside the admissible set $\Lambda$ back into $\Lambda.$ 

Let $(j_{1},j_{2})$ be a study pair that achieves
\[
E\!\left(1\{\mathrm{Lfdr}_{j_{1},j_{2}}\le \lambda_{j_{1},j_{2}}\}\mid 
\theta_{j_{1}}=1,\theta_{j_{2}}=0\right)
=
\min_{j=1,\ldots,n} E_j(\lambda).
\]
For indices $(j_3, j_4)$ 
satisfying 
\[
E\!\left(1\{\mathrm{Lfdr}_{j_{3},j_{4}}\le \lambda_{j_{3},j_{4}}\}\mid
\theta_{j_{3}}=\theta_{j_{4}}=0\right)
\;>\;
E\!\left(1\{\mathrm{Lfdr}_{j_{1},j_{2}}\le \lambda_{j_{1},j_{2}}\}\mid 
\theta_{j_{1}}=1,\theta_{j_{2}}=0\right),
\]
we have already showed that when $j_3 = j_1$ and $j_4 = j_2,$ the pair $(j_1, j_2)$ belongs to $\Lambda$ (see (\ref{eq:tmp_2comp}) in the Supplementary Material). Thereore, it suffices to consider the case $(j_{3},j_{4})\ne (j_1,j_2).$ In this case, 
we decrease $\lambda_{j_{3},j_{4}}$ to a new value $\lambda^{*}_{j_{3},j_{4}}$ such that
\[
E\!\left(1\{\mathrm{Lfdr}_{j_{3},j_{4}}\le \lambda_{j_{3},j_{4}}^{*}\}\mid
\theta_{j_{3}}=\theta_{j_{4}}=0\right)
\le
E\!\left(1\{\mathrm{Lfdr}_{j_{1},j_{2}}\le \lambda_{j_{1},j_{2}}\}\mid 
\theta_{j_{1}}=1,\theta_{j_{2}}=0\right).
\]

This adjustment ensures that the required ordering condition is satisfied. 
    

Algorithm \ref{alg2} therefore constructs fully data-driven thresholds $\widehat{\lambda}_{jj'}(\widehat{\alpha}_p)$ 
that balances the empirical estimate of the false discovery proportion with the target level~$\alpha$.
By combining pariwise Lfdr estimates with the mixture proportion estimates $(\widehat{\pi}_0,\widehat{\pi}_1,\ldots,\widehat{\pi}_{n+1})$,
the procedure provides a principled and scalable framework for multi-study replicability analysis, extending the two-study setting to an arbitrary number of studies.

The following theorem establishes the 
asymptotic validity of the method. Its proof if relegated in Supplementary Material \ref{thm:multi-fdr_proof}.

\begin{theorem}
\label{thm:multi-fdr}
Suppose we have $n$ studies and $m$ hypotheses. For every pair of studies, conditions (C1)–(C5) hold, with the parameters $\xi_{00}, \xi_{10}, \xi_{01}$ and $\xi_{11}$ in (C2) bounded below by a small positive constant, and assume that (C6) holds for all studies. 
Consider testing the composite null hypothesis
$$H_{0i}=\{\theta_{i1}\theta_{i2}\dots\theta_{in}=0\},$$
using Algorithm 2 
with the rejection set $\mathcal I$ defined in (\ref{eq:rejection_set}).
Then, for any prespecified level $\alpha,$ 
we have \[
    \limsup_{m\to\infty} \mbox{FDR}_m \le \alpha.
    \]
\end{theorem}

\section{Simulation studies}\label{sec:simulation}
\subsection{Independent case with two studies}\label{subsec1}
We conduct simulation studies to evaluate the finite sample performance of the proposed method in terms of FDR control and statistical power. The data are generated as follows. Given prior probabilities 
$(\xi_{00}, \xi_{01}, \xi_{10}, \xi_{11})$, we generate the joint hidden states from a multinomial distribution with 
$$ P(\theta_{i1}=k, \theta_{i2}=l)=\xi_{kl}$$ for $k,l\in\{0,1\}.$ 
For study $j\ (j=1,2),$ the test statistics are generated according to $X_{ij}\mid\theta_{ij} \sim \theta_{ij}N(0, \sigma_j^2) + (1-\theta_{ij})N(\mu_j,\sigma_j^2),\ i=1,\dots,m$, where $N(\mu,\sigma^2)$ denotes the normal distribution with mean $\mu$ and variance $\sigma^2$, and $m$ is the number of SNPs. The one-sided $p$-values are computed as $p_{ij} = 1-\Phi\left(\frac{X_{ij}}{\sigma_j}\right)$ for $i=1,\dots,m$ and $j=1,2$, where $\Phi$ is the cumulative distribution function of the standard normal distribution. 

We compare the proposed method with several competing
methods, including \textit{ad hoc} BH, MaxP \citep{benjamini2009selective}, IDR \citep{li2011measuring}, MaRR \citep{philtron2018maximum}, radjust \citep{bogomolov2018assessing}, AdaFilter \citep{wang2022detecting}, and JUMP \citep{lyu2023jump}. Implementation details are provided in Supplementary Material \ref{sec:comparison_methods}. 

\begin{figure}[!ht]
	\centering{\includegraphics[width=0.7\textwidth]{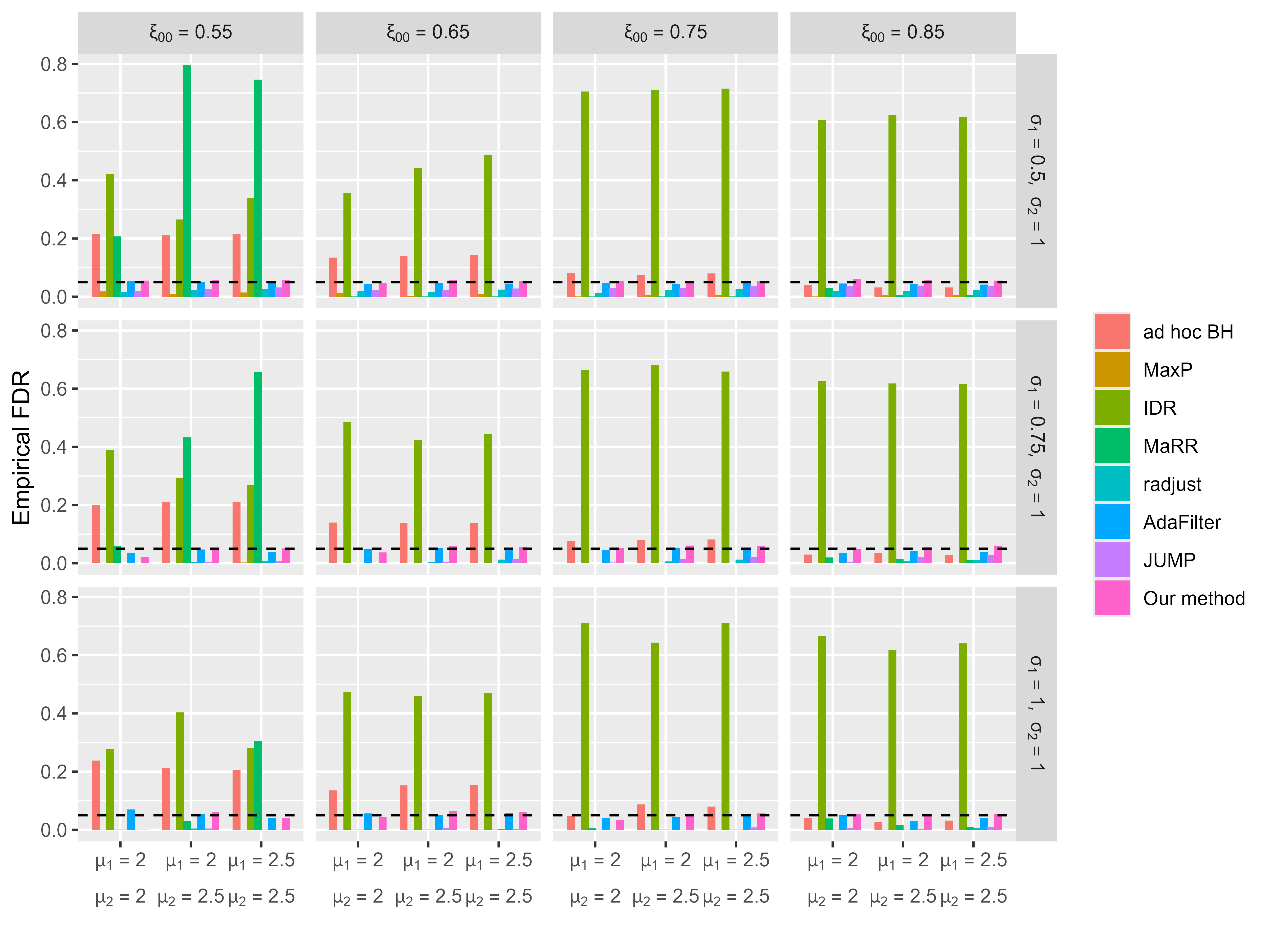}}
	\caption{FDR control of different methods for independent cases with two studies.} \label{fig:simu_fdr}
\end{figure}

\begin{figure}[!ht]
	\centering{\includegraphics[width=0.7\textwidth]{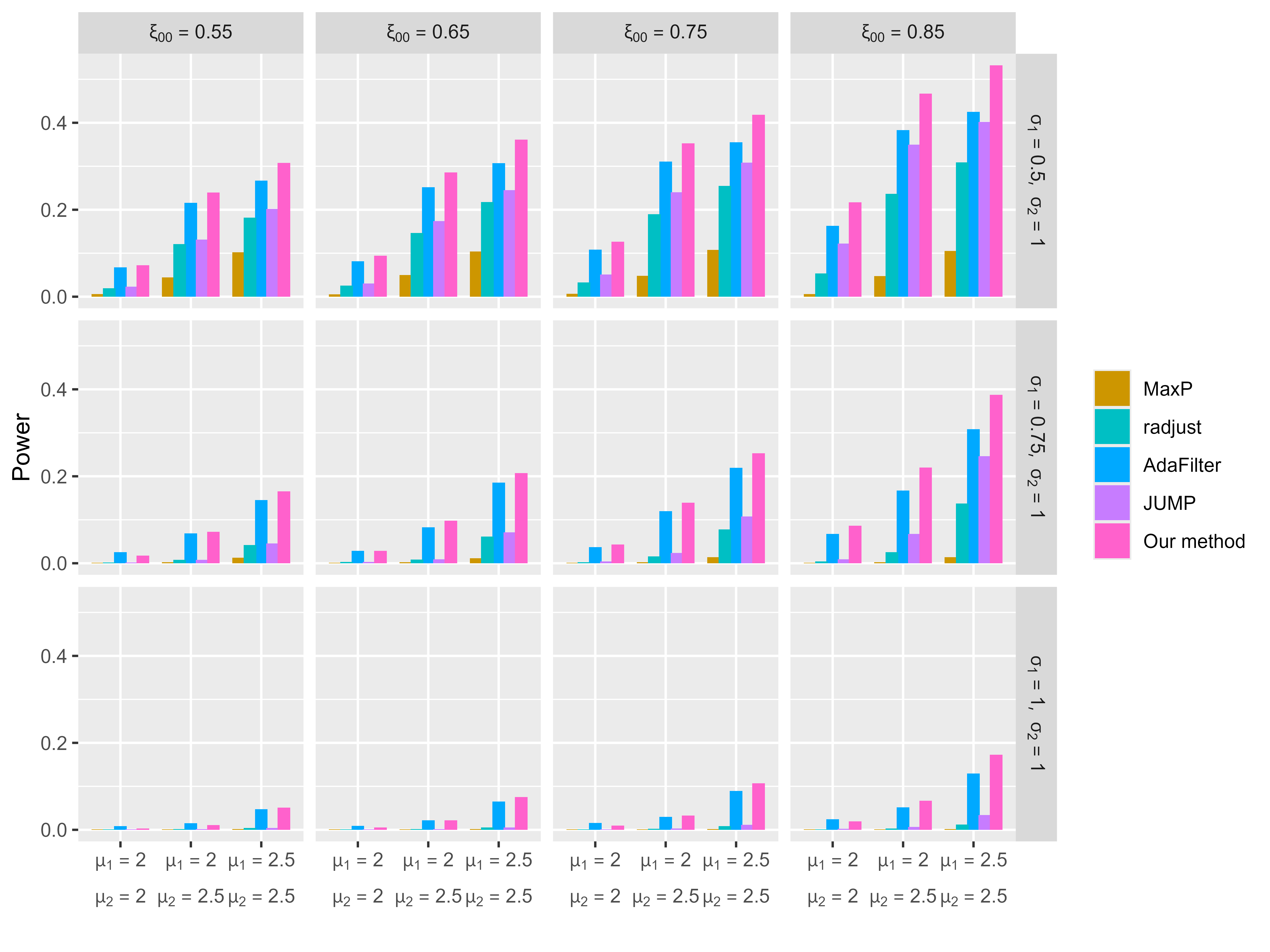}}
	\caption{Power comparison of different methods for independent cases with two studies.} \label{fig:simu_power}
\end{figure}

We set $m=10,000, \xi_{11} = 0.05$ and $\xi_{01} = \xi_{10}$. 
We vary $\mu_1,\mu_2,\sigma_1,\sigma_2$ and $\xi_{00}$ across simulation scenarios. 
For each configurate, empirical FDR and power are estimated 
based on $100$ simulated datasets. Results at the nominal FDR level $0.05$ are presented in Figures \ref{fig:simu_fdr} and \ref{fig:simu_power}, respectively. 

In Figure \ref{fig:simu_fdr}, 
the horizontal dashed line represents the target FDR level. We observe that \textit{ad hoc} BH, IDR, and MaRR fail to control the FDR, whereas the remaining methods achieve valid FDR control across the considered settings. 

Figure \ref{fig:simu_power} compares the power of the methods that successfully control the FDR. We observe that MaxP and radjust exhibit relatively low power, AdaFilter and JUMP achieve moderate power, and the proposed method attains the highest power.

Figure \ref{fig:simu_nom_emp} reports results under 
a sparser configuration with $m=10,000, \xi_{00} = 0.95, \xi_{11} = 0.02$, $\xi_{01} = \xi_{10} = 0.015, \mu_1 = \mu_2 = 2$ and $\sigma_1=\sigma_2=1.$ We evaluate the performance of different methods across nominal FDR levels ranging from $0.001$ to $0.2$. The left panel compares nominal and empirical FDR, where the black diagnoal line (slope 1) serves as a reference. 
Under this sparse setting, IDR and maRR exhibit inflated FDR. The empirical FDR of the proposed method closely tracks the nominal level, indicating accurate control. 
JUMP and AdaFilter also maintain valid FDR control whereas \textit{ad hoc} BH, MaxP, and radjust are overly conservative. In terms of power (right panel), 
\textit{ad hoc} BH, MaxP, and radjust have very low power. AdaFilter and JUMP demonstrate moderate power, while the proposed method achieves the highest power across all nominal levels. 
\begin{figure}[!ht]
\centering{\includegraphics[width=0.7\textwidth]{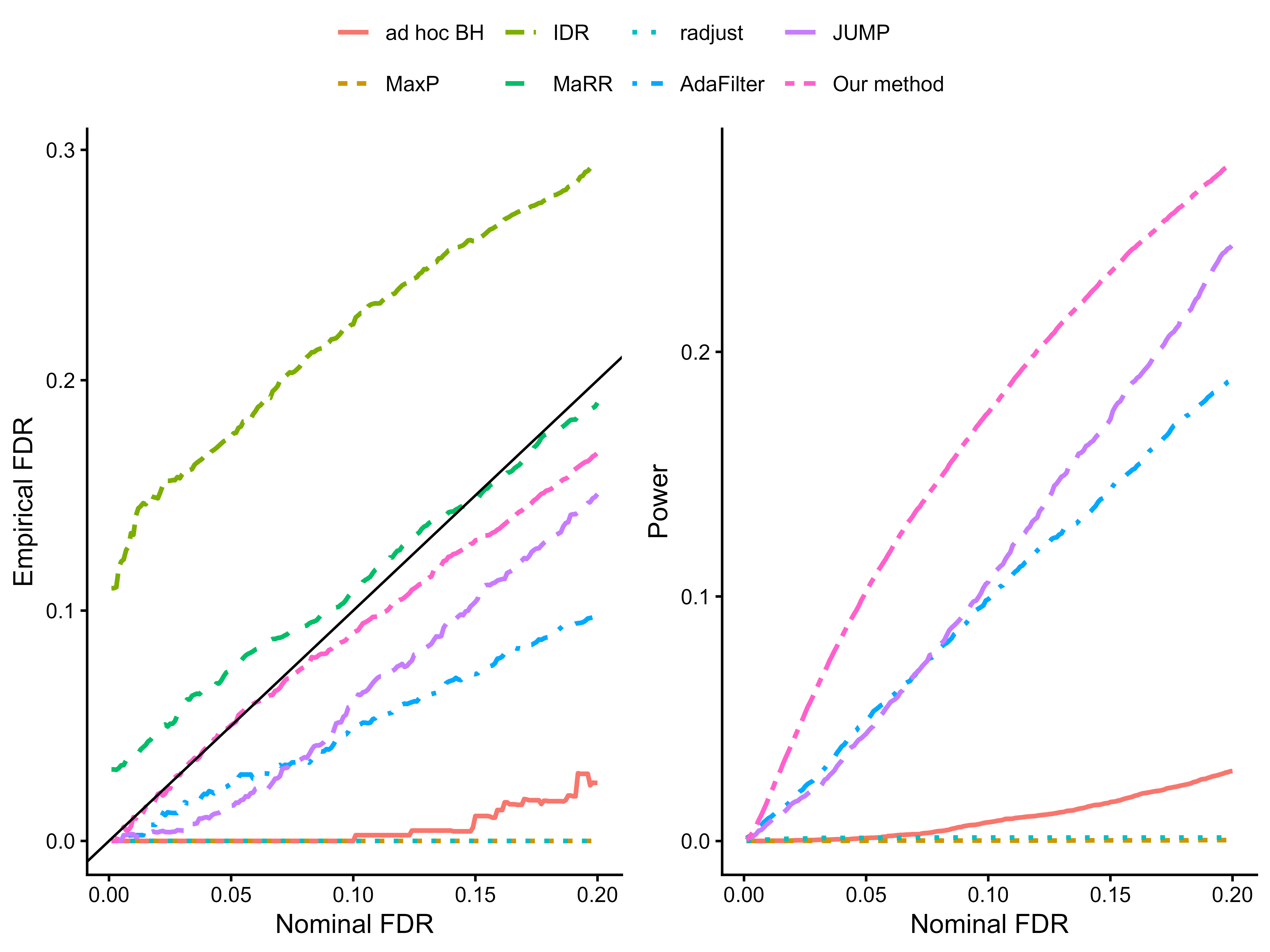}}
	\caption{Left: Empirical versus nominal FDR for independent cases with two studies. Right: Power comparison at different nominal FDR levels.} \label{fig:simu_nom_emp}
\end{figure}

\subsection{Dependent case with two studies}
To assess the robustness of the proposed procedure under local undependence, we 
generate data with correlated test statistics within each study. Specifically, we consider $m=10,000$ hypotheses per study. The hypotheses are divided into $100$ blocks, and each block is further partitiond into two equal-sized sub-blocks. Within each sub-block, test statistics share a constant correlation: one sub-block has correlation $0.2$ and another has correlation $-0.2.$ Different blocks are independent. 

For study $j$ ($j=1,2$), let $\theta_j^*=(\theta_{kj}^*)_{k=1}^{100}\in R^{100}$ denote the block-level latent states. For each block $k,$ the joint configuration $(\theta_{k1}, \theta_{k2})$ follows
$$P(\theta_{k1}=l_1,\theta_{k2}=l_2)=\xi_{l_1l_2} \quad \mbox{for} \quad l_1,l_2\in\{0,1\},$$ 
and configurations are independent across blocks. 

Within each study, the test statistics are generated from a multivariate normal distribution: 
$$ X_{j}\sim \mbox{MVN}(\mu, \Sigma),$$ where the mean vector is $\mu=\theta^*_{j}\cdot\mu_j,$ with $\mu_j$ specified as in Section \ref{subsec1}. The one-sided $p$-values are computed as
$$p_{ij} = 1-\Phi\left(\frac{X_{ij}}{\sigma_j}\right), \quad i=1,\dots,m, \quad j=1,2,$$ where $\Phi$ is the cumulative distribution function of the standard normal distribution. 

To illustrate the dependence structure, we use the following $4 \times 4$ covariance matrix with $\rho = 0.2$ and $-0.2:$  
$$
\Sigma=\begin{pmatrix}
    1    & 0.2  & -0.2 & -0.2\\
    0.2  & 1    & -0.2 & -0.2\\
    -0.2 & -0.2 & 1    & 0.2\\
    -0.2 & -0.2 & 0.2  & 1
\end{pmatrix}.
$$
Larger blocks follow the same correlation pattern.
\begin{figure}[!ht]
	\centering{\includegraphics[width=0.7\textwidth]{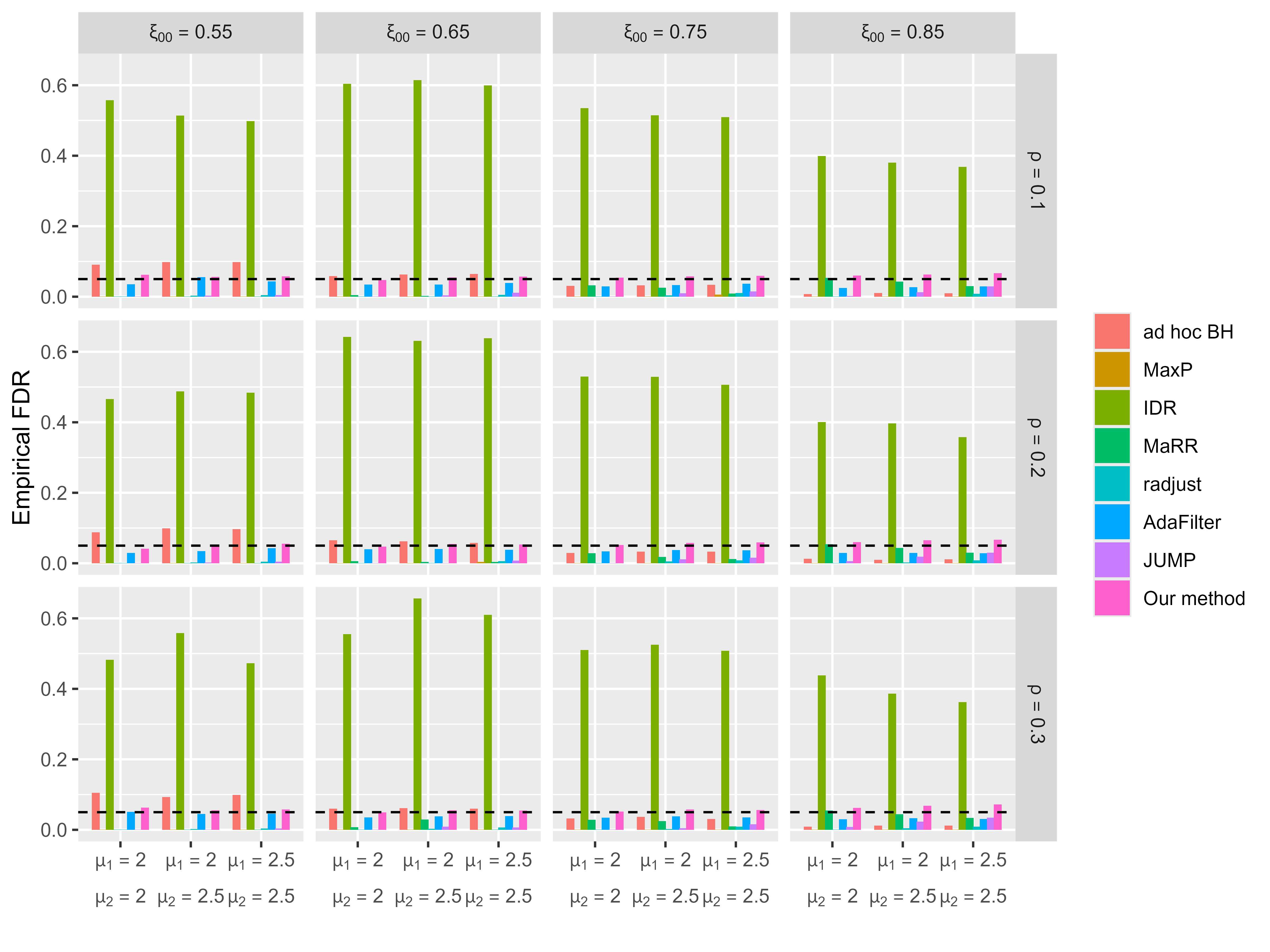}}
	\caption{FDR control of different methods for dependent cases with two studies.} \label{fig:simu_corr_fdr}
\end{figure}

\begin{figure}[!ht]
	\centering{\includegraphics[width=0.7\textwidth]{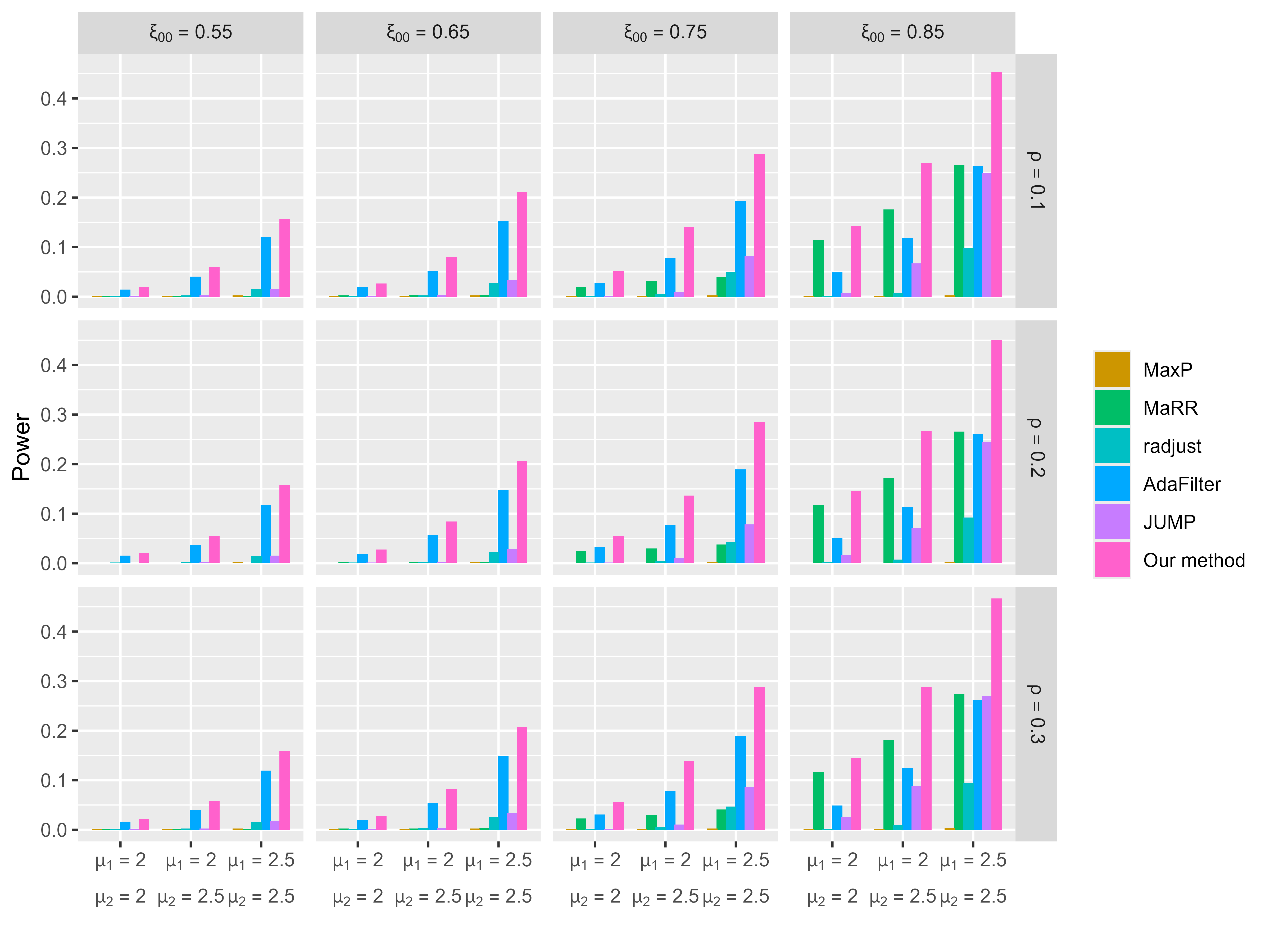}}
	\caption{Power comparison of different methods for dependent cases with two studies.} \label{fig:simu_corr_power}
\end{figure}

Figures \ref{fig:simu_corr_fdr} and \ref{fig:simu_corr_power} report FDR control and power 
at the nominal FDR level $0.05,$ respectively. 
In Figure \ref{fig:simu_corr_fdr}, the horizontal dashed line represents the target FDR level. We observe that \textit{ad hoc} BH and IDR fail to control the FDR under dependence. In contrast, MaxP, MaRR, radjust, AdaFilter, JUMP and the proposed method maintain valid FDR control.

Figure \ref{fig:simu_corr_power} compares power among the methods that control the FDR. 
MaxP and radjust exhibit relatively low
power. AdaFilter, MaRR and JUMP achieve moderate power, while our proposed method attains the highest power across all settings.

Figure \ref{fig:simu_nom_emp_dependent} further examines a sparser scenario with 
$m=10,000, \xi_{00} = 0.95, \xi_{11} = 0.02$, $\xi_{01} = \xi_{10} = 0.015, \mu_1 = \mu_2 = 2$ and $\rho=0.2$. The nominal FDR ranges from $0.001$ to $0.2$. The black diagonal line with slope 1 serves as a reference for exact FDR control. 

In this sparse setting, IDR and maRR exhibit inflated FDR. Our proposed method shows close agreement between nominal and empirical FDR across all levels. 
JUMP and AdaFilter also achieve valid FDR control, whereas \textit{ad hoc} BH, MaxP, and radjust are overly conservative. In terms of power, \textit {ad hoc}
BH, MaxP, and radjust perform poorly; AdaFilter and JUMP provide moderate power, and
our proposed method achieves the highest power throughout the range of nominal FDR levels.

\begin{figure}[!ht]
\centering{\includegraphics[width=0.7\textwidth]{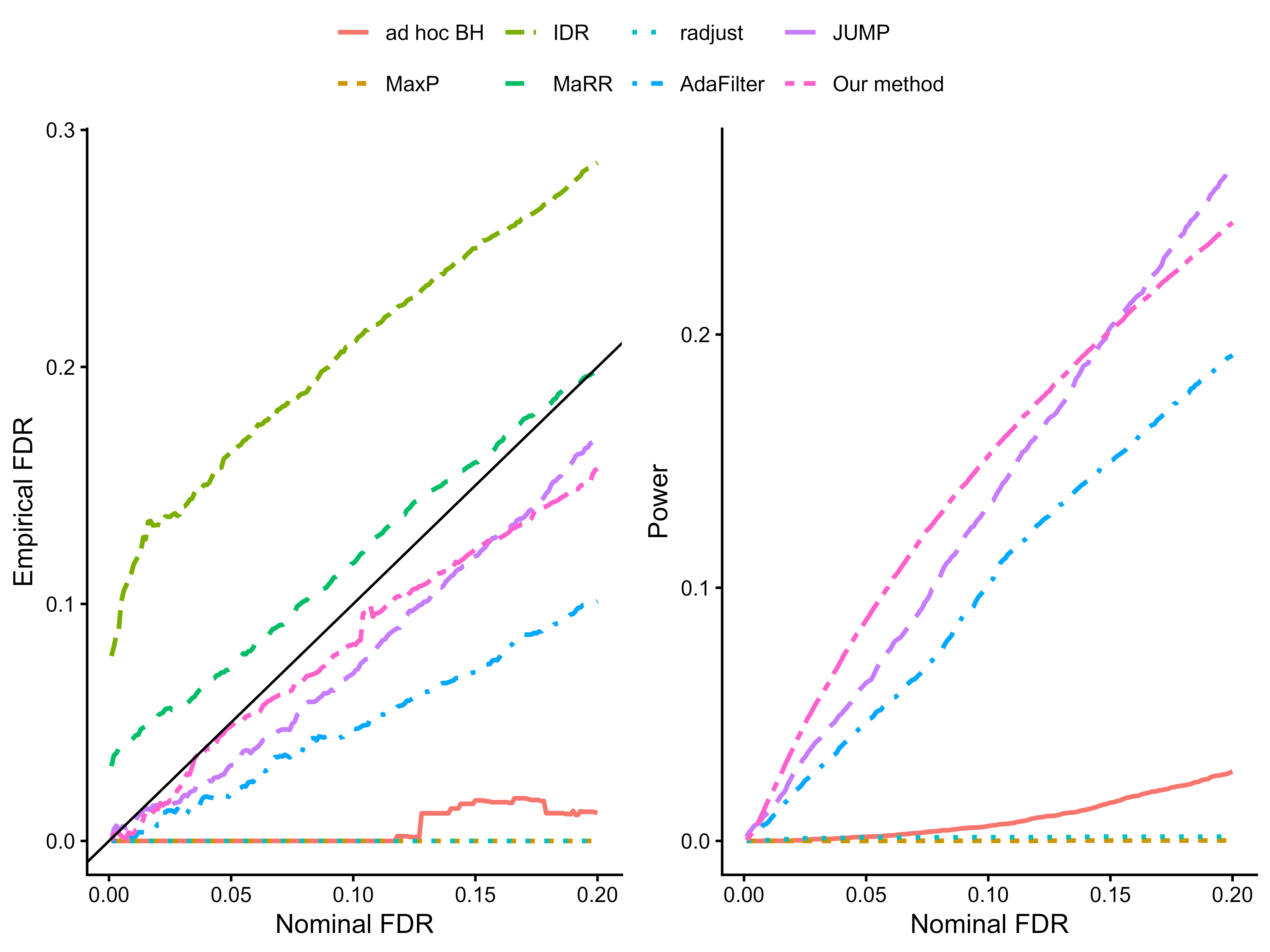}}
	\caption{Left panel: FDR control at different nominal levels for dependent cases with two studies. Right panel: Power comparison at different nominal FDR levels for dependent cases with two studies.} \label{fig:simu_nom_emp_dependent}
\end{figure}

\subsection{Three studies}
In this section, we evaluate FDR control and power for replicability analysis across three studies.

For each hypothesis $i,$ the joint latent states $(\theta_{i1}, \theta_{i2}, \theta_{i3})$ are generated from a multinomial distribution with 
$$P(\theta_{i1}=k, \theta_{i2} = l, \theta_{i3}=r)=\xi_{klr} \quad \mbox{for} \quad k,l,r\in\{0, 1\}.$$ 
For study $j$ ($j=1,2,3$), the test statistics are generated from 
$$X_{ij}\mid\theta_{ij}\sim \theta_{ij}N(0, 1) +(1-\theta_{ij})N(\mu_j, 1), i=1,\dots,m,$$
where $\mu_j>0$. The one-sided $p$-values are computed as
$$p_{ij}=1-\Phi(X_{ij}), \quad, i=1,\dots,m, \quad j=1,2,3.$$
Since IDR and radjust are designed for two studies only, they are excluded from this comparison. 
We extend JUMP to three-study setting. Implementation details are provided in Supplementary Material \ref{sec:comparison_methods}. 

We set $m=10,000, \mu_1 = \mu_2 = 2, \xi_{001} = \xi_{010} = \xi_{100}$ and $ \xi_{011} = \xi_{101} = \xi_{110} = 0.02$. The parameters $\mu_3, \xi_{000}$ and $\xi_{001}$ are varied across simulation settings.  
For each configuration, FDR and power 
are estimated based on 100 simulated datasets. 

Figures \ref{fig:simu_three_fdr} and \ref{fig:simu_three_pow} report FDR control and power at the nominal level $0.05,$
respectively. 
In Figure \ref{fig:simu_three_fdr}, the horizontal dashed line represents the target FDR level. All methods control the FDR except MaRR, which exhibits inflation in some settings. Figure \ref{fig:simu_three_pow} displays only the methods that achieve valid FDR control. MaxP and radjust show relatively low power, while JUMP and AdaFilter achieve moderate power. Our proposed method attains the highest power across all scenarios.

Figure \ref{fig:simu_nom_emp_three} presents a further comparison under the setting 
$$m=10,000, \xi_{000} = 0.28, \xi_{001}=\xi_{010}=\xi_{100} = 0.14, \xi_{011}=\xi_{101}=\xi_{110}=0.06, \xi_{111} = 0.12,$$ with $$\mu_1 = \mu_2 =2.5, \mu_3=3.5.$$ The nominal FDR ranges from $0.001$ to $0.2$.
A diagonal line with slope one serves as a reference for exact FDR control. 

In this setting, all methods control the FDR, although MaRR and MaxP are noticeably conservative. In terms of power, \textit{ad hoc} BH, AdaFilter and JUMP demonstrate moderate performance, whereas our proposed method achieves the highest power uniformly across the entire range of nominal FDR levels. 

\begin{figure}[!ht]
	\centering{\includegraphics[width=0.7\textwidth]{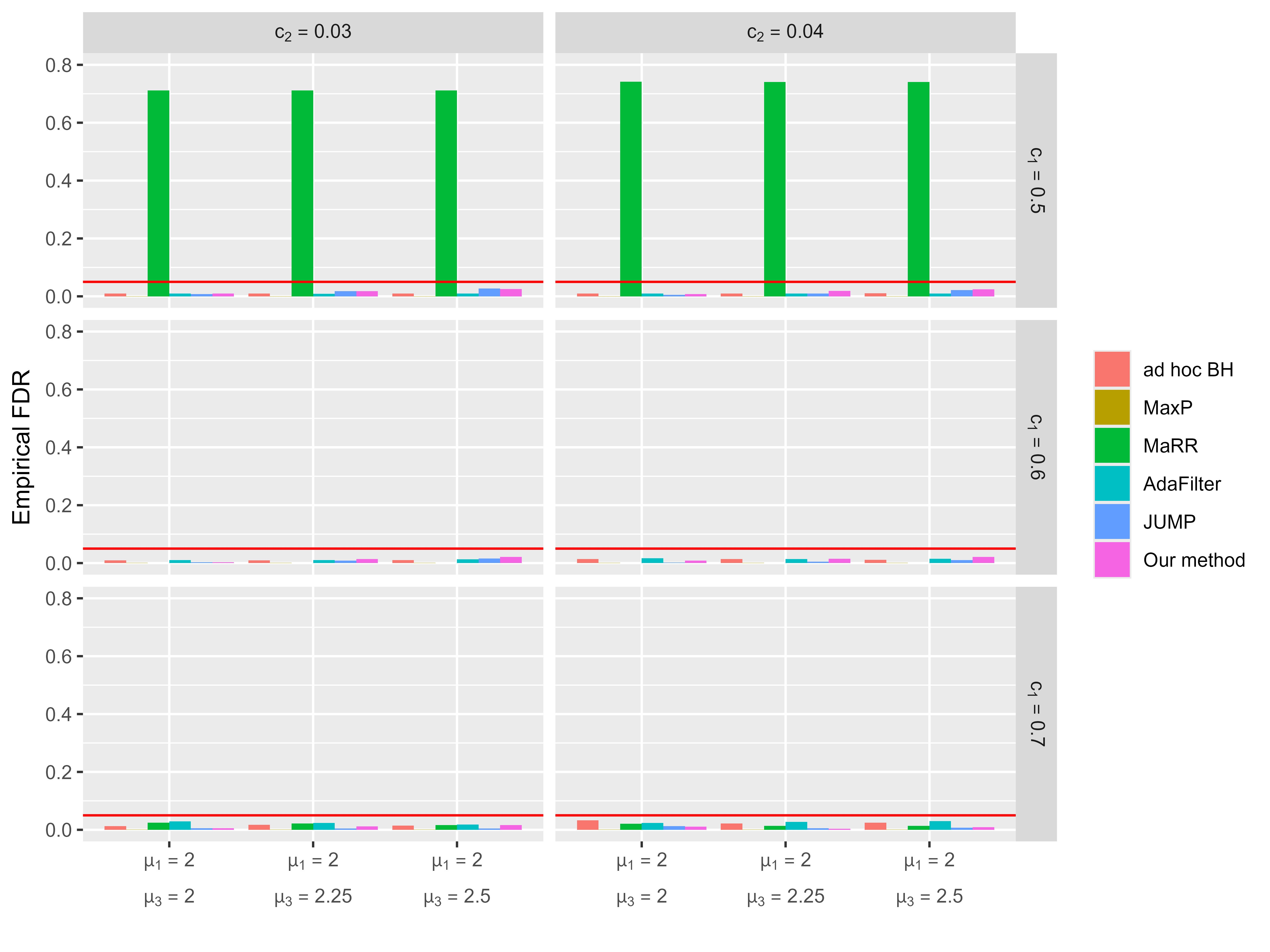}}
	\caption{FDR control of different methods for three studies. } \label{fig:simu_three_fdr}
\end{figure}

\begin{figure}[!ht]
    \centering{
    \includegraphics[width=0.7\linewidth]{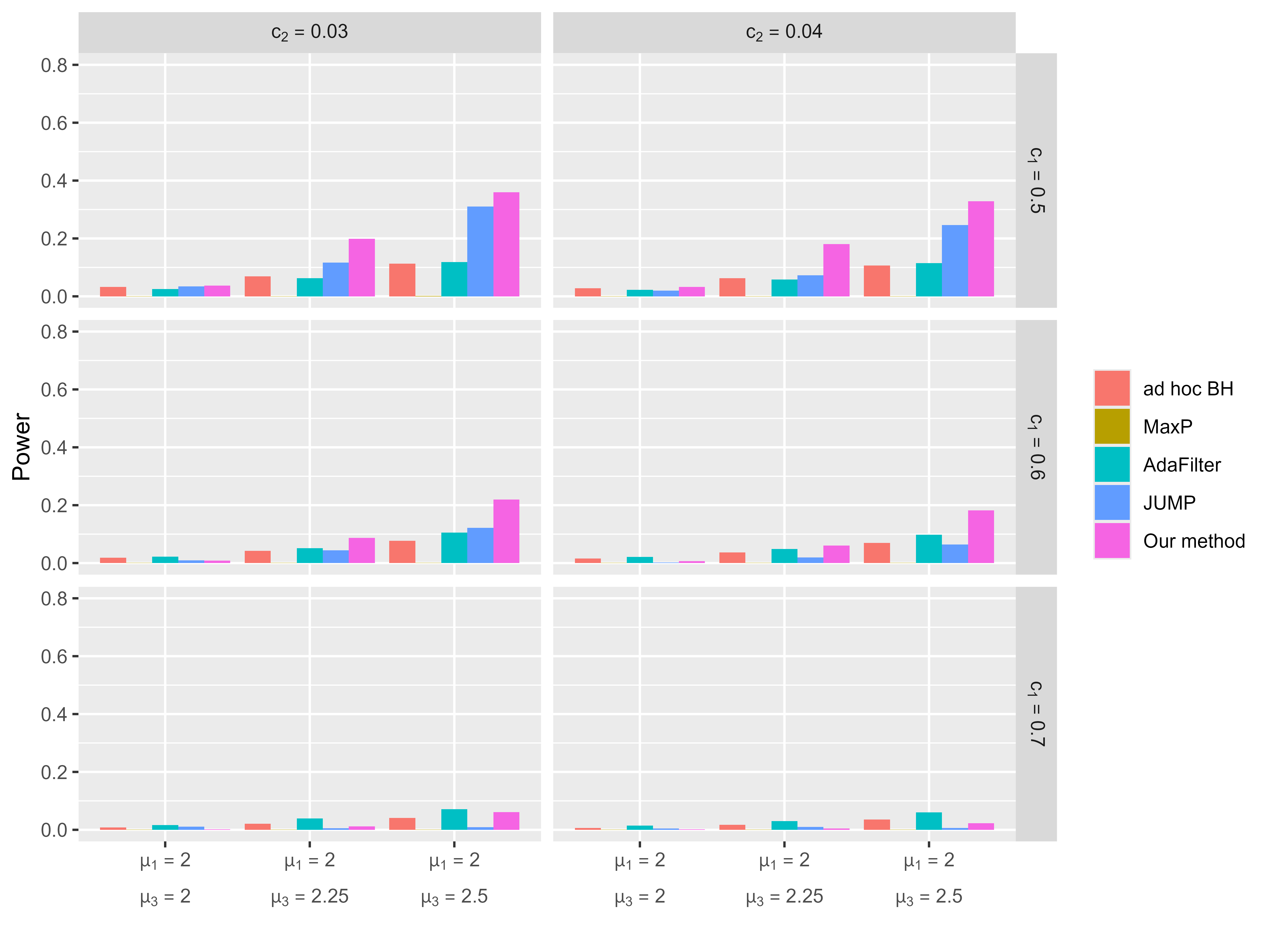}}
    \caption{Power comparison of different methods for three studies.}
    \label{fig:simu_three_pow}
\end{figure}

\begin{figure}[!ht]
\centering{\includegraphics[width=0.7\textwidth]{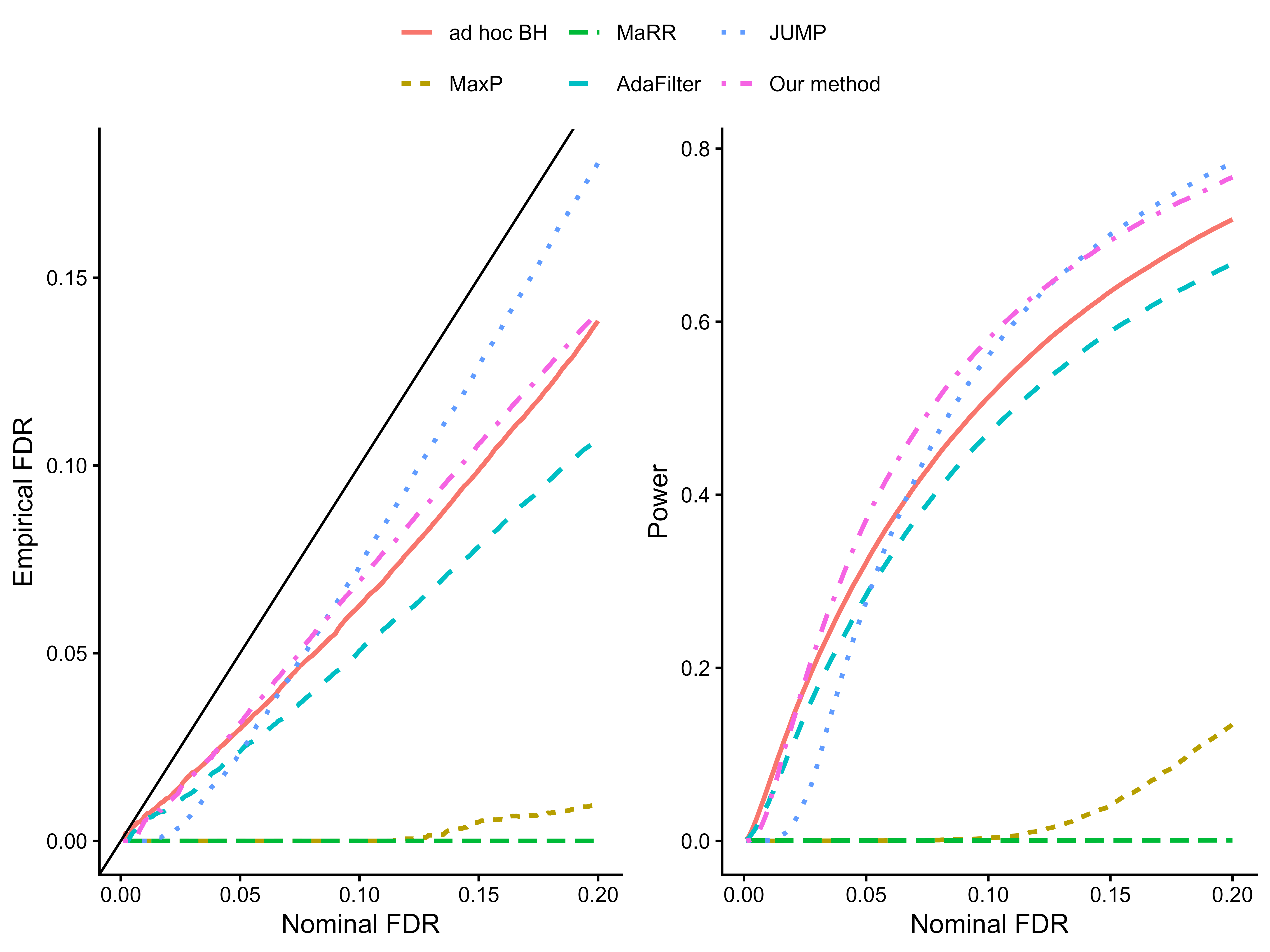}}
	\caption{Left panel: FDR control at different nominal levels for three studies. Right panel: Power comparison at different nominal FDR levels for three studies.} \label{fig:simu_nom_emp_three}
\end{figure}

\section{Data analysis}\label{sec:data}
We illustrate the proposed method using genome-wide association study (GWAS) data. \cite{mahajan2022multi} assembled an ancestrally diverse collection of 122 GWAS comprising $180,834$ type 2 diabetes (T2D) cases and $1,159,055$ controls across five ancestry groups. We conduct a replicability analysis of the East Asian-specific and European-specific GWAS to identify SNPs associated with T2D that replicate across the two populations. 

The datasets are publicly available from the DIAbetes Genetics Replication And Meta-analysis Consortium at \url{https://www.diagram-consortium.org/downloads.html}. The East Asian dataset contains summary statistics for $760,565$ SNPs on Chromosome 2, while the European dataset includes $879,506$ SNPs on the same chromosome. We focus on the $760,565$ SNPs shared by both studies. Let $p_{i1}, p_{i2}, i =1, \ldots, m$ denote the $p$-values for SNPs in the East Asian and European ancestry studies, respectively.

\begin{figure}[!ht]
	\centering{\includegraphics[width=\textwidth]{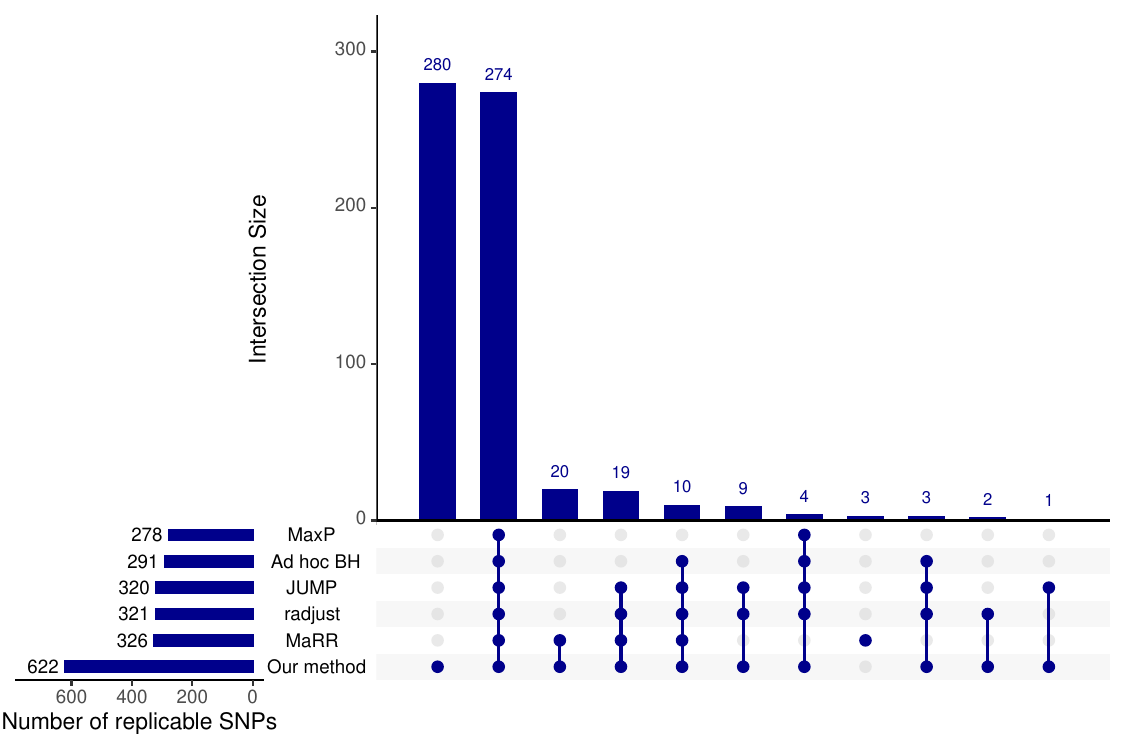}}
	\caption{Replicability analysis results of the T2D GWAS data.} \label{fig:analysis_t2d}
\end{figure}

We compare our approach with several 
competing replicability analysis methods. 
The IDR method is excluded because it is not computationally scalable for this dataset and did not complete within $10$ hours. The results at FDR level $0.01$ are summarized in Figure \ref{fig:analysis_t2d}. Among the compared method, MaxP is the most conservative, followed by \textit {ad hoc} BH, JUMP, radjust and MaRR. 
Our method identifies 622 replicable SNPs, including 280 SNPs that are not detected by any of the other methods. This finding is consistent with the simulation results, where the proposed method demonstrated the highest power while maintaining FDR control.

Given the substantially larger sample size in the European ancestry dataset, more significant signals are expected in the European population than in the East Asian population. This 
pattern is reflected in the estimated mixture proportions: 
$$\hat\xi_{00} = 0.8923, \quad \hat\xi_{01} = 0.0984, \quad \hat\xi_{10} = 0.0003, \quad \hat\xi_{11} = 0.0089. $$ 
The estimated non-null density functions for the two studies, $\hat f_1$ and $\hat f_2$, are displayed in Figure \ref{fig:t2d_nonnull}. 

To validate the findings, we refer to the NHGRI-EBI GWAS Catalog \citep{sollis2023nhgri}, which curates published SNP-trait associations. 
We declare a locus replicable if at least one identified SNPs maps to a T2D locus reported in the GWAS Catalog. 
If multiple significant SNPs fall within a locus, the SNP with the strongest association is designed as the lead SNP. The mapped gene is defined as the gene overlapping or closest to the lead SNP. 

Among the $280$ SNPs uniquely identified by our method, $7$ are directly reported in the GWAS Catalog. In addition, $233$ SNPs fall within $9$ known T2D loci tagged by SNPs detected by our method (see Table \ref{tab:t2d}). Of the reamaining $40$ SNPs, $37$ lie within three loci (2q23.1, 2q31.1, and 2q36.3) previously reported in the GWAS Catalog and tagged by other SNP markers. 

We furthe annotate the $280$ uniquely identified SNPs using SNPnexus, a web-based functional  annotation tool for human genetic variation \citep{oscanoa2020snpnexus}. 
These SNPs map to $41$ genes, many of which have established links to T2D. For instance,  \textit{GCKR} has been associated with T2D and kidney-related phenotypes \citep{ho2022association}. Variants in \textit{BCL11A} influence fetal hemoglobin levels and affect insulin response and glucagon secretion, thereby modulating T2D risk \citep{xu2013corepressor, cauchi2012european}. 
The \textit{GRB14} gene is known to inhibit insulin signaling and contribute to insulin resistance in T2D and obesity \citep{gondoin2017identification}. These biological findings provide additional support for the validity of the SNPs uniquely identified by our method.

\begin{figure}[!ht]
	\centering{\includegraphics[width=\textwidth]{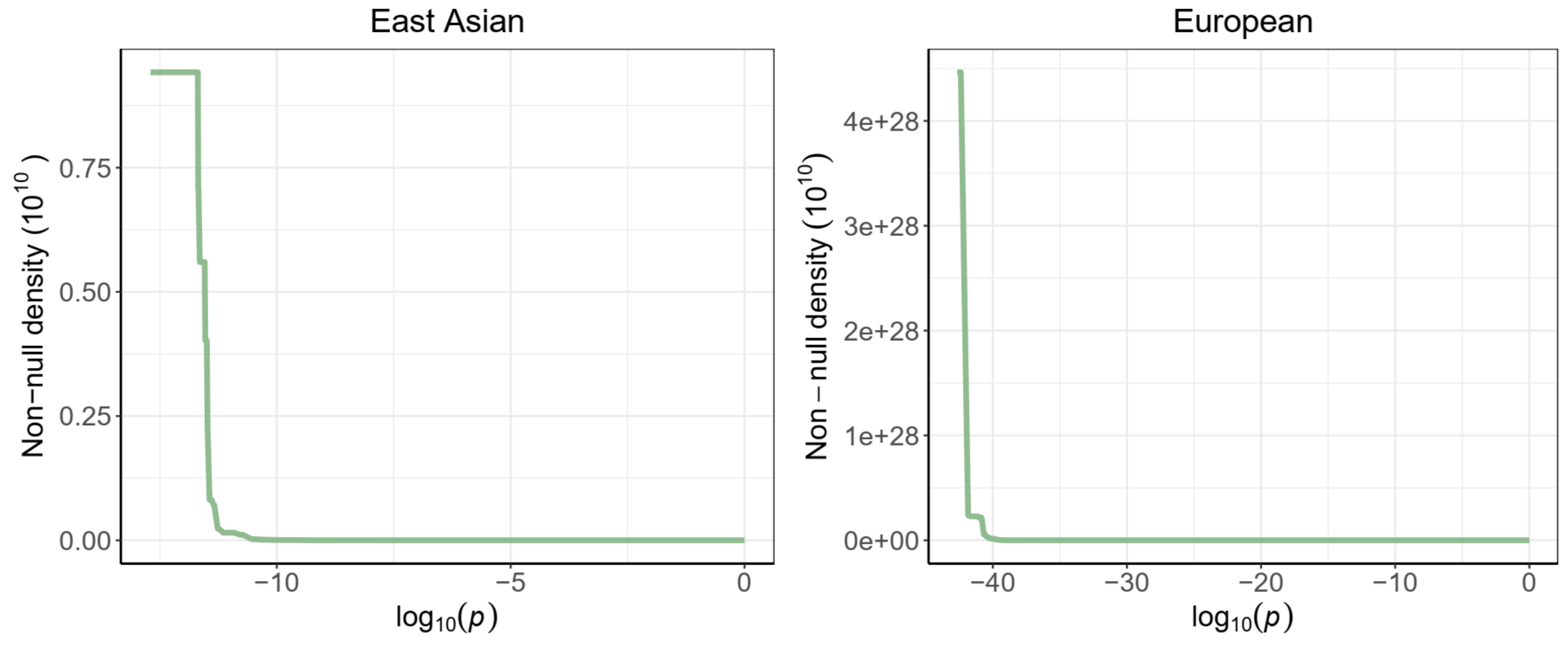}}
	\caption{Estimates of the non-null probability densities in two T2D studies} \label{fig:t2d_nonnull}
\end{figure}

\begin{table}[!ht]
    \centering
    \caption{Main characteristics of the 
    T2D loci tagged by SNPs uniquely identified by the proposed method.}
    \label{tab:t2d}
    \vspace{0.5cm}
    \
    \scalebox{0.9}{
    \renewcommand\arraystretch{1.5}
    \begin{threeparttable}
        \begin{tabular}{ llll }
    \hline\hline
        \textbf{Locus} & \textbf{Lead SNP} & \textbf{Location of lead SNP}\tnote{1} & \textbf{Mapped gene}     \\
        \hline
        2p14     & rs6752053   &  chr:65,439,540   & \textit{LINC02934} \\
        2p16.1   & rs243018    & chr2:60,359,572  & \textit{MIR4432HG}\\
        2p21     & rs12712928  & chr:44,964,941 & \textit{SIX3,KRTCAP2P1}\\
        2p23.3   & rs1260326, rs6547692  & chr:27,508,073, chr:27,512,105 & \textit{GCKR}\\
        2p25.3   & rs35913461  & chr2:653,575  & \textit{TMEM18, LINC01875}         \\
        2q24.2  & rs305686  & chr:162,767,422  &\textit{KCNH7} \\
        2q24.3  & rs3769873  & chr2:164,709,073   & \textit{COBLL1}\\
        2q34   & rs3828242  & chr2:211,410,212  & \textit{ERBB4}\\
        2q37.1  & rs838720  & chr2:233,394,635  & \textit{DGKD}\\
    \hline\hline
    \end{tabular}

\begin{tablenotes}
    \footnotesize\item[1] Locations are based on Genome Assembly GRCh38/hg38.
\end{tablenotes}
\end{threeparttable}}
\end{table}

\section{Concluding remarks}\label{sec:discussion}
In this paper, we proposed an empirical Bayes framework for robust and powerful inference in high-dimensional multi-study replicability analysis.
The method explicitly accommodates study heterogeneity by jointly modeling latent configurations across studies while allowing study-specific non-null densities. 
To enhance robustness against model misspecification, we incorporated a nonparametric density estimation procedure under a monotonicity constraint. 
By systematically enumerating the components of the composite null hypothesis, the proposed approach effectively borrows information across both genetic variants and studies, leading to improved power while maintaining valid false discovery control. 

We established that the procedure asymptotically controls the FDR. In addition, we derived the optimal minimax convergence rate for the estimated density function and showed that the Lfdr based rejection rule achieves optimal power among procedures operating at the same nominal FDR level. Extensive simulation studies demonstrate that our method outperforms existing alternatives. Applications to multi-cohort GWAS data further reveal biologically meaningful associations that competing methods may fail to detect.

Our theoretical development assumes conditional independence of $p$-values given the joint latent states. Although this assumption can be relaxed to allow weak dependence, as supported by our simulation results, the corresponding theoretical analysis becomes substantially more involved and is left for future investigation. Addressing stronger dependence structures, such as those arising from hidden Markov models, would require additional methodological advances.

The proposed inference framework is built on the use of $p$-values to assess replicability across studies. Despite ongoing debates regarding their limiations, $p$-values remain widely used due to their versatility and wide availability.
In many large-scale collaborative projects, privacy concernss or data-sharing restrictions limit access to individual-level data, whereas summary-level $p$-values are easily accessible.
Moreover, $p$-values can be obtained through various approaches, including asymptotic approximations, permutation procedures, and exact finite-sample methods.
Consequently, our framework is platform-independent and broadly applicable whenever summary-level $p$-values are accessible.

Although the empirical Bayes procedure extends naturally to more than two studies, the number of mixture components increases exponentially with the number of studies, posing computational challenges when $n$ is large.
The extensions described in Section \ref{sec:extend} mitigate this issue by introducing pairwise-based estimation and aggregation strategies, thereby making large-scale multi-study replicability analysis computationally feasible.

\section*{Acknowledgment}
The research of Hongyuan Cao is partially supported by NSF DMS 2311249. We thank Yingying Wei for helpful discussions.
\newpage

\bibliographystyle{abbrvnat}
\bibliography{ref}

\newpage
\appendix

\appendix

\newpage
\section{Proofs} \label{sec:proof}
\numberwithin{equation}{section}
\numberwithin{proposition}{section}
\numberwithin{lemma}{section}
\numberwithin{theorem}{section}
\setcounter{equation}{0}
\setcounter{proposition}{0}
\setcounter{lemma}{0}
\setcounter{theorem}{0}

\subsection{EM algorithm for MLE}
\label{sup:EM}
Let ${p}_1=\{p_{i1}\}_{i=1}^m$ and ${p}_2=\{p_{i2}\}_{i=1}^m$ be the observed $p$-values, and ${\theta}_1=\{\theta_{i1}\}_{i=1}^m$ and ${\theta}_2=\{\theta_{i2}\}_{i=1}^m$ be 
corresponding hidden states. The joint log-likelihood function of (${p}_1,{p}_2,{\theta}_1,{\theta}_2$) can be written as
$$
\begin{aligned}
l({p}_1,{p}_2,{\theta}_1,{\theta}_2)&=\frac{1}{m}\sum\limits_{i=1}^m\{(1-\theta_{i1})(1-\theta_{i2})\log\xi_{00}+\theta_{i1}(1-\theta_{i2})\log\xi_{10}+(1-\theta_{i1})\theta_{i2}\log\xi_{01}\\&\quad+\theta_{i1}\theta_{i2}\log\xi_{11}+\theta_{i1}\log f_1(p_{i1})+\theta_{i2}\log f_2(p_{i2})\}.
\end{aligned}
$$
We iteratively implement the following two steps.

\textbf{E-step:} Given current estimates of 
$w^{(k)}=(\xi_{00}^{(k)},\xi_{10}^{(k)},\xi_{01}^{(k)},\xi_{11}^{(k)},f_1^{(k)},f_2^{(k)})$, obtain the 
conditional expectation of the log-likelihood function as 
$$
\begin{aligned}
Q(w|w^{(k)})&=E_{{\theta}_1,{\theta}_2 \mid w^{(k)}}\{l({p}_1,{p}_2,{\theta}_1,{\theta}_2)\}\\&
=\frac{1}{m}\sum\limits_{i=1}^m\{\gamma_{i,00}^{(k)}\log\xi_{00}+\gamma_{i,10}^{(k)}\log\xi_{10}+\gamma_{i,01}^{(k)}\log\xi_{01}+\gamma_{i,11}^{(k)}\log\xi_{11}\\&\quad+(\gamma_{i,10}^{(k)}+\gamma_{i,11}^{(k)})\log f_1(p_{i1})+(\gamma_{i,01}^{(k)}+\gamma_{i,11}^{(k)})\log f_2(p_{i2})\},
\end{aligned}
$$
where 
the posterior probabilities of hidden states are calculated as follows.  
$$
\begin{aligned}
    \gamma_{i,00}^{(k)}&=\frac{\xi_{00}^{(k)}}{\xi_{00}^{(k)}+\xi_{10}^{(k)}f_1^{(k)}(p_{i1})+\xi_{01}^{(k)}f_2^{(k)}(p_{i2})+\xi_{11}^{(k)}f_1^{(k)}(p_{i1})f_2^{(k)}(p_{i2})},
    \\
    \gamma_{i,10}^{(k)}&=\frac{\xi_{10}^{(k)}f_1^{(k)}(p_{i1})}{\xi_{00}^{(k)}+\xi_{10}^{(k)}f_1^{(k)}(p_{i1})+\xi_{01}^{(k)}f_2^{(k)}(p_{i2})+\xi_{11}^{(k)}f_1^{(k)}(p_{i1})f_2^{(k)}(p_{i2})},
    \\
    \gamma_{i,01}^{(k)}&=\frac{\xi_{01}^{(k)}f_2^{(k)}(p_{i2})}{\xi_{00}^{(k)}+\xi_{10}^{(k)}f_1^{(k)}(p_{i1})+\xi_{01}^{(k)}f_2^{(k)}(p_{i2})+\xi_{11}^{(k)}f_1^{(k)}(p_{i1})f_2^{(k)}(p_{i2})},
    \\
    \gamma_{i,11}^{(k)}&=\frac{\xi_{11}^{(k)}f_1^{(k)}(p_{i1})f_2^{(k)}(p_{i2})}{\xi_{00}^{(k)}+\xi_{10}^{(k)}f_1^{(k)}(p_{i1})+\xi_{01}^{(k)}f_2^{(k)}(p_{i2})+\xi_{11}^{(k)}f_1^{(k)}(p_{i1})f_2^{(k)}(p_{i2})}.
\end{aligned}
$$

\textbf{M-step:} Update $w^{(k+1)}$ by
$$
w^{(k+1)}=\operatornamewithlimits{argmax}_{w\in \Omega} Q(w|w^{(k)}).
$$
subject to the constraint that $\xi_{00}+\xi_{01}+\xi_{10}+\xi_{11}=1$ and $f_1$ and $f_2$ are non-increasing density functions. 
We have 
$$
\xi_{00}^{(k+1)}=\frac{1}{m}\sum\limits_{i=1}^m \gamma_{i,00}^{(k)},\quad\xi_{10}^{(k+1)}=\frac{1}{m}\sum\limits_{i=1}^m \gamma_{i,10}^{(k)},\quad \xi_{01}^{(k+1)}=\frac{1}{m}\sum\limits_{i=1}^m \gamma_{i,01}^{(k)},\quad\xi_{11}^{(k+1)}=\frac{1}{m}\sum\limits_{i=1}^m \gamma_{i,11}^{(k)},
$$
and
\begin{align}
    f_1^{(k+1)}&=\operatornamewithlimits{argmax}_{f_1 \in \mathbb{H}}\sum\limits_{i=1}^m (\gamma_{i,10}^{(k)}+\gamma_{i,11}^{(k)})\log f_1(p_{i1}), \label{pava1}%
    \quad 
    \\
    f_2^{(k+1)}&=\operatornamewithlimits{argmax}_{f_2 \in \mathbb{H}}\sum\limits_{i=1}^m (\gamma_{i,01}^{(k)}+\gamma_{i,11}^{(k)})\log f_2(p_{i2}),\label{pava2}
\end{align}
where $\mathbb{H}$ is the space of a non-increasing density function. 
We repeat the above \textbf{E-step} and \textbf{M-step} until the algorithm converges. 

Next, we provide specific steps to solve (\ref{pava1}). 
Let $ 0=p_{(0),1}\le\cdots\le p_{(i),1}\le\cdots\le p_{(m),1}\le1=p_{(m+1),1}$ be the order statistics of ${p_{1}}.$ Denote $\eta_{i1}=\gamma_{i,10}^{(k)}+\gamma_{i,11}^{(k)}, i=1,\dots,m$ and let $\eta_{(i),1}$ 
corresponds to $p_{(i),1}$.
Denote $y_i =f_1(p_{(i),1})$, ${y}=(y_1,\cdots,y_m)$ and $\mathcal Q=\{(y_1,\cdots,y_m):y_1\ge y_2\ge \cdots\ge y_m\}.$ The optimization problem (\ref{pava1}) can be written as 
$$
\hat {{y}}=\operatornamewithlimits{argmax}_{{y}\in \mathcal Q} \sum\limits_{i=1}^m \eta_{(i),1}\log y_{{i}}, \text{  subject to } \sum\limits_{i=1}^m \{p_{(i),1}-p_{(i-1),1}\}y_i=1.
$$

We first ignore the monotonic constraint $\mathcal Q$. By applying the Lagrangian multiplier $\nu_1$, the objective function to maximize is
$$
L( y,\nu_1)=\sum\limits_{i=1}^m\eta_{(i),1}\log y_{i}+\nu_1 \left[\sum\limits_{i=1}^m\{p_{(i),1}-p_{(i-1),1}\}y_i-1\right].
$$

The maximizer of $\nu_1$ and $y_i$
satisfy
\begin{align*}
    \frac{\partial L( y,\nu_1)}{\partial y_i}&=\frac{\eta_{(i),1}}{y_i}+\nu_1\{p_{(i),1}-p_{(i-1),1}\}=0,\\
    \frac{\partial L( y,\nu_1)}{\partial \nu_1}&=\sum\limits_{i=1}^m\{p_{(i),1}-p_{(i-1),1}\}y_i-1=0.
\end{align*}
We have
\begin{align*}
    &\tilde\nu_1 = - \sum_{i=1}^m \eta_{(i),1} 
    \\ 
    &\tilde y_i = \frac{\eta_{(i),1}}{\{\sum_{i=1}^m \eta_{(i),1}\}\{p_{(i),1}-p_{(i-1),1}\}}, \quad i=1,\dots,m.
\end{align*}

To incorporate the monotonic constraint $\mathcal Q$, we have
$$
{\hat{y}}=\operatornamewithlimits{argmin}_{{y}\in \mathcal Q}\left\{-L({y},\tilde\nu_1)\right\}=\operatornamewithlimits{argmin}_{{y}\in \mathcal Q}\sum_{i=1}^m\eta_{(i),1}\left[-\log y_i - \frac{-\{\sum_{i=1}^m \eta_{(i),1}\}\{p_{(i),1}-p_{(i-1),1}\}}{\eta_{(i),1}}y_i\right].
$$

Let ${u} = (u_1,\dots,u_m)$ and 
$$
\hat{{u}}=\operatornamewithlimits{argmin}_{{u}\in \mathcal Q}\sum_{i=1}^m\eta_{(i),1}\left[u_i-\frac{-\{\sum_{i=1}^m \eta_{(i),1}\}\{p_{(i),1}-p_{(i-1),1}\}}{\eta_{(i),1}}\right]^2.
$$
The solution takes the max-min form
$$
\hat u_i = \max_{b\ge i}\min_{a\le i} \frac{-\{\sum_{i=1}^m \eta_{(i),1}\}\sum_{k=a}^b\{p_{(k),1}-p_{(k-1),1}\}}{\sum_{k=a}^b\eta_{(k),1}}.
$$
According to Theorem 3.1 of \cite{barlow1972isotonic}, the final estimates of (\ref{pava1}) are given by $\hat f_1(p_{(i),1}) = \hat y_i=-\frac{1}{\hat u_i}$. This can be implemented by PAVA \citep{JSSv032i05}.

The estimation of $f_2$ can be obtained by solving the optimization problem (\ref{pava2}) in the same way, and we omit the details. 

\subsection{Proof of Proposition \ref{identify}}
\label{sec:id_proof}
\begin{proof}
Recall that $w=(\xi_{00},\xi_{10},\xi_{01},\xi_{11},f_1,f_2)$ and
$$p_w(x,y)=\xi_{00}+\xi_{10}f_1(x)+\xi_{01}f_2(y)+\xi_{11}f_1(x)f_2(y).$$
Integrating out $y$ and $x$ from $p_w(x,y)$ respectively,
we get
\begin{equation}\label{marginal-x}
g_{w,1}(x)=\int_0^1 p_{w}(x,y) dy=(\xi_{00}+\xi_{01})+(\xi_{10}+\xi_{11})f_1(x),
\end{equation}and 
\begin{equation}\label{marginal-y}
g_{w,2}(y)=\int_0^1 p_{w}(x,y) dx=(\xi_{00}+\xi_{10})+(\xi_{01}+\xi_{11})f_2(y).
\end{equation}
By (\ref{marginal-x}) and (\ref{marginal-y}), we have
$$
\begin{aligned}p_w(x,y)&=\xi_{00}+\frac{\xi_{10}}{\xi_{10}+\xi_{11}}\{g_{w,1}(x)-(\xi_{00}+\xi_{01})\}+\frac{\xi_{01}}{\xi_{01}+\xi_{11}}\{g_{w,2}(y)-(\xi_{00}+\xi_{10})\}\\&\quad+\frac{\xi_{11}}{(\xi_{10}+\xi_{11})(\xi_{01}+\xi_{11})}\{g_{w,1}(x)-(\xi_{00}+\xi_{01})\}\{g_{w,2}(y)-(\xi_{00}+\xi_{10})\}\\&=\frac{\xi_{00}\xi_{11}-\xi_{10}\xi_{01}}{(\xi_{10}+\xi_{11})(\xi_{01}+\xi_{11})}\{1-g_{w,1}(x)\}\{1-g_{w,2}(y)\}+g_{w,1}(x)g_{w,2}(y).
\end{aligned}
$$

Note $p_w(x,y)=p_{w^*}(x,y)$ almost everywhere, where
$$w=(\xi_{00},\xi_{01},\xi_{10},\xi_{11},f_1,f_2),\quad\text{and}\quad w^*=(\xi_{00}^*,\xi_{01}^*,\xi_{10}^*,\xi_{11}^*, f_1^*,f_2^*).$$ The marginal distributions of $p_{w}$ and $p_{w^*}$ should be the same almost everywhere. Specifically, $$g_{w,1}(x)=g_{w^*,1}(x), \quad \mbox{and} \quad g_{w,2}(y)=g_{w^*,2}(y)\quad\text{a.e.}$$ 
Consequently $(\xi_{00},\xi_{10},\xi_{01},\xi_{11})$ and $(\xi_{00}^*,\xi_{10}^*,\xi_{01}^*,\xi_{11}^*)$ have to satisfy the following conditions. 
\begin{equation}
\xi_{00}+\xi_{10}+\xi_{01}+\xi_{11}=\xi_{00}^*+\xi_{10}^*+\xi_{01}^*+\xi_{11}^*=1, \label{id_4}
\end{equation}
and
\begin{equation}
    \frac{\xi_{00}\xi_{11}-\xi_{10}\xi_{01}}{(\xi_{10}+\xi_{11})(\xi_{01}+\xi_{11})}=\frac{\xi^*_{00}\xi^*_{11}-\xi^*_{10}\xi^*_{01}}{(\xi^*_{10}+\xi^*_{11})(\xi^*_{01}+\xi^*_{11})} \label{id_3}.
\end{equation}

Under assumption (C1) that $$\lim\limits_{x\to 1}f_1(x)=\lim\limits_{x\to 1}f_1^*(x)=\lim\limits_{x\to 1}f_2(x)=\lim\limits_{x\to 1}f_2^*(x)=0,$$
by (\ref{marginal-x}) and (\ref{marginal-y}), when $x\to 1$ we have 
    \begin{align}
        \xi_{00}+\xi_{10}&=\xi^*_{00}+\xi^*_{10}\label{id_1}=\lim\limits_{x\to 1}g_{w,1}(x),\quad\text{and}\\
        \xi_{00}+\xi_{01}&=\xi^*_{00}+ \xi^*_{01}\label{id_2}=\lim\limits_{x\to 1}g_{w,2}(x).
    \end{align}
    By (\ref{id_1}),(\ref{id_2}),(\ref{id_4}), we get
    \begin{equation}
    \xi_{01}+\xi_{11}=\xi^*_{01}+\xi^*_{11}, \quad \xi_{10}+\xi_{11}= \xi^*_{10}+\xi^*_{11}.\label{id_5}
    \end{equation}
    Plugging into (\ref{id_3}), we have
    $$ \xi_{00}\xi_{11}-\xi_{10}\xi_{01}= \xi^*_{00}\xi^*_{11}-\xi^*_{10} \xi^*_{01}.$$
    Note that 
    $$\xi_{00}\xi_{11}-\xi_{10}\xi_{01}=\xi_{00}\xi_{11}+\xi_{00}\xi_{01}-\xi_{00}\xi_{01}-\xi_{10}\xi_{01}=\xi_{00}(\xi_{01}+\xi_{11})-\xi_{01}(\xi_{00}+\xi_{10}).$$
    We have
    \begin{equation}
      \xi_{00}(\xi_{01}+\xi_{11})-\xi_{01}(\xi_{00}+\xi_{10})=\xi^*_{00}(\xi^*_{01}+\xi^*_{11})-\xi^*_{01}( \xi^*_{00}+\xi^*_{10}).\label{id_6}  
    \end{equation}
    By (\ref{id_1}) and (\ref{id_5}), (\ref{id_6}) is equivalent to 
    \begin{equation*}
      \xi_{00}(\xi_{01}+\xi_{11})-\xi_{01}(\xi_{00}+\xi_{10})=\xi^*_{00}(\xi_{01}+\xi_{11})-\xi^*_{01}( \xi_{00}+\xi_{10}).
    \end{equation*}
    So that
    \begin{equation}
        (\xi_{00}-\xi^*_{00})(\xi_{01}+\xi_{11})-(\xi_{01}-\xi^*_{01})(\xi_{00}+\xi_{10})=0.\label{id_7}
    \end{equation}
    By (\ref{id_2}), we have $\xi_{00}-\xi^*_{00}=-(\xi_{01}-\xi^*_{01})$. Plugging into (\ref{id_7}), we have 
    \begin{equation}
    \label{id_8}
    (\xi_{01}-\xi_{01}^*)(\xi_{00}+\xi_{10}+\xi_{01}+\xi_{11})=\xi_{01}-\xi^*_{01}=0.
    \end{equation}
   Plugging (\ref{id_8}) into (\ref{id_1}), (\ref{id_2}), and (\ref{id_4}), we have
    $$ \xi_{00}=\xi^*_{00},\xi_{10}=\xi^*_{10},\xi_{01}=\xi^*_{01},\xi_{11}=\xi^*_{11}.$$
    By (\ref{marginal-x}) and (\ref{marginal-y}), we obtain 
    $$f_1(x)=f_1^*(x),\quad f_2(y)=f_2^*(y) \quad a.e.$$
\end{proof}

\subsection{Proof of Proposition \ref{bracketing_2}}
\label{sec:bracket_proof_s1}

In this proof, the assumption $F\in L_{2+\epsilon'}$ for some $\epsilon'>0$ in (C2) can be replaced by a weaker condition $F\in L_{2,2/3}$, where the Lorentz norm $\|F\|_{a_1,a_2}$ is defined by
    \begin{equation}
    \label{eq:Lorentz}
    \|F(x)\|_{a_1,a_2}=\left(\int_0^{\infty}\left(\kappa \mu\left\{x:|f(x)|^{a_1} \geq \kappa\right\}\right)^{\frac{a_2}{a_1}} \frac{d \kappa}{\kappa}\right)^{\frac{1}{a_2}},
    \end{equation}
where $\mu$ is the Lebesgue measure.
The proposed functional space $L_{2,2/3}$ is a refinement of the classical $L_2$ space, providing finer control of tail behavior. Basic properties of Lorentz spaces can be found in Chapter 4.4 of \cite{bennett1988interpolation}. Specifically, as the support of $F$ is the bounded interval $(0,1)$, we have $L_{2+\epsilon'}\subset L_{2,2/3}\subset L_{2,2}=L_{2}$ for $\epsilon'>0$.

Under conditions (C1) and (C2), we rewrite $\mathcal F$ as follows.
$$\mathcal F=\{p_w=\xi_{00}+\xi_{10}f_1(x)+\xi_{01}f_2(y)+\xi_{11}f_1(x)f_2(y):w=(\xi_{00},\xi_{10},\xi_{01},\xi_{11},f_1(x),f_2(y))\in \Omega\},$$
where
$$\Omega=\Delta\times \tilde {\mathcal G}\times \tilde {\mathcal G},$$
where, for a small constant $c$, and density functional space $\mathcal G$ defined in (\ref{G1}),
$$\begin{aligned}
    \Delta&=\{(\xi_{00},\xi_{10},\xi_{01},\xi_{11})\in R^4_+:\xi_{00}+\xi_{10}+\xi_{01}+\xi_{11}=1,\xi_{00}\ge c \},\\
    \tilde {\mathcal G}&=\{f: f\in \mathcal G,f\le F,\lim\limits_{x\to1} f(x)=0\},\\
\end{aligned}$$ 
Without loss of generality, we further assume that $F(x)$ is non-increasing and large enough such that $F(x)>3,\forall x\in (0,1)$, and $c$ is small enough such that $c<1/5$.

We first show that the bracketing number of $\mathcal {\bar{F}}^{1/2}$ can be bounded by the bracketing number of $\mathcal{F}$, summarized in  
Proposition \ref{bracketing_1}. 
Then we calculate the upper bound of $N_{[]}(\epsilon,\mathcal F,\|\cdot\|_2)$, and the lower bound of $N_{[]}(\epsilon,\mathcal{\bar F}^{1/2},\|\cdot\|_2),$ respectively.
 
\begin{proposition}
\label{bracketing_1}
Under (C1) and (C2), 
we have 
$$N_{[]}(\epsilon,\mathcal{\bar F}^{1/2},\|\cdot\|_2)\le N_{[]}(2\sqrt{2c}\epsilon,\mathcal F,\|\cdot\|_2).$$
\end{proposition}

\subsubsection{Proof of Proposition \ref{bracketing_1}} 
Let $[p^L_1,p^U_1],\cdots,[p^L_K,p^U_K]$ be the smallest $\epsilon$-brackets of $\mathcal F.$ 
Without loss of generality, we further assume that $p^U_k,p^L_k\ge 0, k =1, \ldots, K,$ because every function in the functional space $\mathcal{F}$ is bounded away from $c>0$. We have
$$\left|\sqrt{\frac{p^U_k+p_{w_0}}{2}}-\sqrt{\frac{p^L_k+p_{w_0}}{2}}\right|
=\frac{1}{\sqrt 2}\frac{|p_k^U-p_k^L|}{\sqrt{p^U_k+p_{w_0}}+\sqrt{p^L_k+p_{w_0}}}
\le \frac{1}{2\sqrt{2c}} |p^U_k-p^L_k|,$$
where we used $p_{w_0} \ge c$ specified in (C2) in the last inequality. 

Then 
$$\left\|\sqrt{\frac{p^U_k+p_{w_0}}{2}}-\sqrt{\frac{p^L_k+p_{w_0}}{2}}\right \|_2\le \frac{1}{2\sqrt{2c}}\|p^U_k-p^L_k\|_2\le \frac{1}{2\sqrt{2c}}\epsilon.$$
Therefore, 
$$N_{[]}(\frac{1}{2\sqrt{2c}}\epsilon,\mathcal{\bar F}^{1/2},\|\cdot\|_2)\le N_{[]}(\epsilon,\mathcal F,\|\cdot\|_2).$$

\subsubsection{Calculating the upper bound}
\underline{\textbf{Step 1:}} Calculate the bracketing entropy of one-dimensional non-increasing density functions supported on $(0,1)$ with envelop function $F\in L^{2,2/3}$.

\begin{lemma}
\label{bracketing}
    Let $\mathcal M$ be the space of the one-dimensional non-increasing density function supported on $(0,1)$ with envelop function $F\in L_{2,2/3}$. Then there exists a constant $C_F>0$ such that
    $$\log N_{[]}(\delta, \mathcal M,\|\cdot\|_2)\le C_F\frac{1}{\delta}.$$
\end{lemma}

\begin{proof}
    We partition the space $(0,1)=\bigcup\limits_{i=1}^\infty I_i$, where $I_i=(x_{i},x_{i-1}]$ with $1= x_0\ge\cdots\ge x_n\ge\cdots>0$, and $x_i=\sup_x \{F(x)\ge 2^i\}.$ Therefore, we have 
    \begin{equation}
        \label{eq:Ii}
        I_i=\{x: 2^{i}>F(x)\ge 2^{i-1}\}
    \end{equation}

    We further define the functional space 
    $$
    \mathcal M_{w,h}=\{f(x): x\in (0,w), 0\le f(x)\le h, f \text{ is monotone nonincreasing}\}.
    $$
It has been shown by equation (2.5) on page 18 of \cite{van2000empirical} that $$\log N_{[]}(\delta,\mathcal M_{1,1},\|\cdot\|_2)\le \frac{C}{\delta}.$$
    So we have
    \begin{equation}
    \label{eq:mono_entropy}
    \log N_{[]}(\delta,\mathcal M_{w,h},\|\cdot\|_2)\le \frac{Ch w^{1/2}}{\delta}.    
    \end{equation}
    This is because if we take $f_1,f_2\in\mathcal M_{1,1}$ and define $g_1(x)=hf_1(\frac{x}{w})$ and $g_2(x)=hf_2(\frac{x}{w})$, we have $g_1(x),g_2(x)\in \mathcal M_{w,h}$ and
    $$
    \int_0^w|g_1(x)-g_2(x)|^2dx=h^2\int_0^w|f_1(\frac{x}{w})-f_2(\frac{x}{w})|^2dx=h^2w\int_0^1|f_1(x)-f_2(x)|^2dx.
    $$
    Therefore, we have $\log N_{[]}(\delta, \mathcal M_{w,h},\|\cdot\|)=\log N_{[]}(\frac{\delta}{hw^{1/2}}, \mathcal M_{1,1},\|\cdot\|)\le \frac{Chw^{1/2}}{\delta}$.

    Define $c_i=\int_{I_i}F^2(x)dx$. We find that 
    \begin{align}
        c_i&=\int_{I_i}F^2(x)dx\ge |I_i|F^2(x_{i-1})\ge |I_i|(2^{i-1})^2,\quad \text{and}  \nonumber \\
        c_i&=\int_{I_i}F^2(x)dx\le |I_i|F^2(x_{i})\le |I_i|4^i,  \label{eq:c_iI_i}
    \end{align}
     where $|I_i|$ is the length of the interval $I_i$.

    Within each interval $I_i$, we construct $\delta_i$ brackets $[f^L_{i,k}(x), f_{i,k}^U(x)], k=1,\ldots, N_i$, where $\delta_i=\frac{c_i^{1/6}}{\sqrt{\sum\limits_{i=1}^\infty c_i^{1/3}}}\delta$, and $N_i$ is $\delta_i$ bracketing number of monotone nonincreasing function supported by $|I_i|$ with envelop function $F(x)$. As the envelop function $F(x)\le 2^i$ when $x\in I_i$, based on (\ref{eq:mono_entropy}), we have $$\log N_i\le\log N_{[]}(\delta_i,\mathcal M_{|I_i|,2^i},\|\cdot\|_2)\le \frac{C2^i|I_i|^{1/2}}{\delta_i}.$$  
    Combining the local brackets constructed on each $I_i,$ we obtain the bracketing of $\mathcal M$ as follows. 
    $$
    B :=\left\{\left[f^L, f^U\right]: f^L=\sum_{i=1}^\infty f_{i, k_i}^L 1 _{I_i}, f^U=\sum_{i=1}^\infty f_{i, k_i}^U 1 _{I_i}, k_i \in\left\{1, \ldots, N_i\right\}\right\} .
    $$
    It is easy to verify that $\sum\limits_{i=1}^\infty \delta_i^2=\delta^2.$ Therefore $B$ forms a $\delta$-bracketing of $\mathcal M$ in $L_2$.
    
    So, the $\delta$ bracketing entropy of $\mathcal M$ is bounded by
    
    \begin{align}
    \log N_{[]}(\delta, \mathcal M,\|\cdot\|_2)\le & \sum\limits_{i=1}^\infty \log N_i\nonumber
    \\\le &\sum\limits_{i=1}^\infty \log N_{[]}({\delta_i,\mathcal M_{|I_i|, 2^{i}},\|\cdot\|_2})\nonumber
    \\\le &C\sum\limits_{i=1}^\infty \frac{2^{i}|I_i|^{1/2}}{\delta_i}\nonumber\\\le& C\sum\limits_{i=1}^\infty \frac{2\sqrt{c_i}}{\delta_i}\nonumber
    \\\le &2C(\sum\limits_{i=1}^\infty c_i^{1/3})^{3/2}\frac{1}{\delta},\label{eq: total_bracket}
    \end{align}
   where the second to the last inequality follows (\ref{eq:c_iI_i}), and the last inequality is obtained by substituting the value of $\delta_i$.

    Now we provide an upper bound of $\sum\limits_{i=1}^\infty c_i^{1/3}$. Based on the definition of the Lorentz norm defined in equation (\ref{eq:Lorentz}), we have
    $$
    \begin{aligned}
        \|F\|_{2,{2/3}}^{2/3}&=
    \int_0^{\infty}\left(\kappa \mu\left\{x:|F(x)|^2 \geq \kappa\right\}\right)^{\frac{1}{3}} \frac{d \kappa}{\kappa}
    \\&=\int_0^{\infty}\kappa^{-2/3} \mu\left\{x:|F(x)|^2 \geq \kappa\right\}^{\frac{1}{3}} d\kappa
    \\&=\int_0^1\kappa^{-2/3} \mu\left\{x:|F(x)|^2 \geq \kappa\right\}^{\frac{1}{3}} d\kappa+\sum\limits_{i=1}^\infty\int_{4^{i-1}}^{4^i}\kappa^{-2/3} \mu\left\{x:|F(x)|^2 \geq \kappa\right\}^{\frac{1}{3}} d\kappa.
    \end{aligned}
    $$
    Note that for the first term, as $c_1\le 4|I_1|\le 4$ in (\ref{eq:c_iI_i}), and $F(x)\ge1$ we have
    $$
    \int_0^1\kappa^{-2/3} \mu\left\{x:|F(x)|^2 \geq \kappa\right\}^{\frac{1}{3}} d\kappa=\int_0^1\kappa^{-2/3}d\kappa=3\ge \frac{3}{4^{4/3}}c_1^{1/3}.
    $$
    For the remaining terms, we have
    $$
    \begin{aligned}
    &\int_{4^{i-1}}^{4^i}\kappa^{-2/3} \mu\left\{x:|F(x)|^2 \geq \kappa\right\}^{\frac{1}{3}} d\kappa
    \\\ge &
    \int _{4^{i-1}}^{4^i}(4^i)^{-2/3}\mu\left\{x:|F(x)|^2 \geq 4^i\right\}^{\frac{1}{3}}d\kappa
    \\= & (4^i-4^{i-1})\cdot(4^i)^{-2/3}\mu\left\{x:|F(x)|^2 \geq 4^i\right\}^{\frac{1}{3}}
    \\\ge &\frac{3}{4}(4^i)^{1/3}|I_{i+1}|^{1/3}
    \\= &\frac{3}{4^{4/3}}(4^{i+1})^{1/3}|I_{i+1}|^{1/3}
    \\\ge &\frac{3}{4^{4/3}}c_{i+1}^{1/3}.\end{aligned}
    $$
 The second to the last inequality is based on (\ref{eq:Ii}) and $I_{i+1}\subset\{x: |F(x)|^2\ge 4^i\}$. The last inequality is based on (\ref{eq:c_iI_i}) that $c_i\le 4^i|I_{i}|$.
    Therefore, we have 
    \begin{align*}
        \|F\|_{2,2/3}^{2/3}&=\int_0^1\kappa^{-2/3} \mu\left\{x:|F(x)|^2 \geq \kappa\right\}^{\frac{1}{3}} d\kappa+\sum\limits_{i=1}^\infty\int_{4^{i-1}}^{4^i}\kappa^{-2/3} \mu\left\{x:|F(x)|^2 \geq \kappa\right\}^{\frac{1}{3}} d\kappa
        \\&\ge \frac{3}{4^{4/3}} c_1^{1/3} + \frac{3}{4^{4/3}}\sum\limits_{i=1}^\infty c_{i+1}^{1/3}
        \\&\ge \frac{3}{4^{4/3}}\sum\limits_{i=1}^\infty c_{i}^{1/3}.
    \end{align*}
    So we have 
    $$
    (\sum\limits_{i=1}^\infty c_{i}^{1/3})^{3/2}\le \{\frac{4^{4/3}}{3}\|F\|_{2,2/3}^{2/3}\}^{3/2}\le 4\|F\|_{2,2/3}.
    $$
    Plugging into (\ref{eq: total_bracket}), we have 
    $\log N_{[]}(\delta, \mathcal M,\|\cdot\|_2)\le 8C\|F\|_{2,2/3} \cdot \frac{1}{\delta}$.
    \end{proof}

\underline{\textbf{Step 2:}} Bound $\|p_{w^U}-p_{w^L}\|_2.$ 

Write $p_{w^U}(p_1,p_2)=\xi_{00}^U+\xi_{10}^Uf_1(p_1)+\xi_{01}^Uf_2(p_2)+\xi_{11}^Uf_1^U(p_1)f_2^U(p_2)$, and $p_{w^L}(p_1,p_2)=\xi_{00}^L+\xi_{10}^Lf_1(p_1)+\xi_{01}^Lf_2(p_2)+\xi_{11}^Lf_1^L(p_1)f_2^L(p_2)$, where $[\xi_{ij}^U,\xi_{ij}^L],i,j=0,1$ is one of the $\delta$ bracket of $\Delta$; $[f_1^U,f_1^L]$ is one of the $\delta$ bracket of $\mathcal G$; and $[f_2^U,f_2^L]$ is one of the $\delta$ bracket of $\mathcal G_2$, where $\Delta$, $\mathcal G_1$, and $\mathcal G$ are defined in (\ref{par_space}).
So 
we have $$\delta\ge\max\{|\xi_{ij}^U-\xi_{ij}^L|,\|f_1^U-f_1^L\|_2,\|f_2^U-f_2^L\|_2\},i,j=0,1.$$ Using the fact that
$$
\int A(x)B(y)dxdy=\int A(x)dx\int B(y)dy, 
$$
and the assumption $(C2)$,
 
we obtain
$$
\begin{aligned}
\|\xi_{11}^Uf_1^U(p_1)f_2^U(p_2)&-\xi_{11}^Lf_1^L(p_1)f_2^L(p_2)\|_2\\&\le\xi_{11}^U\|f_1^U(p_1)f_2^U(p_2)-f_1^L(p_1)f_2^L(p_2)\|_2+\|f_1^L(p_1)f_2^L(p_2)\|_2\cdot|\xi^U_{11}-\xi^L_{11}|\\&\le \|f_1^U(p_1)f_2^U(p_2)-f_1^L(p_1)f_2^L(p_2)\|_2+\|F\|^2_2\delta\\&\le \|F\|^2_2\delta+ \|f_1^U(p_1)(f_2^U(p_2)-f_2^L(p_2))\|_2+\|f_2^L(p_2)(f_1^U(p_1)-f_1^L(p_1))\|_2\\&\le \|F\|^2_2\delta+\|F\|_2\|f_2^U-f_2^L\|_2+\|F\|_2\|f_1^U-f_1^L\|_2\\&\le (\|F\|_2^2+2\|F\|_2)\delta.
\end{aligned}
$$
In the same way, we can further show that 
$$
\begin{aligned}
|\xi_{00}^U-\xi_{00}^L|&\le \delta,\\
\|\xi_{10}^Uf_1^U-\xi_{10}^Lf_1^L\|_2&\le (\|F\|_2+1)\delta,\\
\|\xi_{01}^Uf_2^U-\xi_{01}^Lf_2^L\|_2&\le (\|F\|_2+1)\delta.
\end{aligned}
$$
Then 
$$
\begin{aligned}
\|p_{w^U}-p_{w^L}\|_2&\le |\xi_{00}^U-\xi_{00}^L|+\|\xi_{10}^Uf_1^U-\xi_{10}^Lf_1^L\|_2+\|\xi_{01}^Uf_2^U-\xi_{01}^Lf_2^L\|_2+\|\xi_{11}^Uf_1^Uf_2^U-\xi_{11}^Lf_1^Lf_2^L\|_2\\&\le (\|F\|^2_2+4\|F\|_2+3)\cdot\delta.
\end{aligned}
$$


\underline{\textbf{Step 3}}: Construct the pairs of functions with bracketing entropy. 
We first construct the pairs in $\Delta, \mathcal G_1,\mathcal G_2$ separately and combine them together. 

Let $\delta>0$, and $K_1=\lceil\frac{1}{\delta}\rceil\le C_5\delta^{-1}$.
For $\xi_{11}$, write $\xi_{11,k}^L=\frac{k-1}{K_1},\xi_{11,k}^U=\frac{k}{K_1}$, $k=1,\cdots,K_1$. Then for any $\xi_{11}\in (0,1)$, there exists $k$ such that $\xi_{11,k}^L\le \xi_{11}\le \xi_{11,k}^U$ and $|\xi_{11,k}^L-\xi_{11,k}^U|\le \delta$. For $\xi_{00},\xi_{10},\xi_{01}$, we do the same thing. 

By Lemma \ref{bracketing}, we can choose the correponding bracketing function $(f_{1,k}^L,f_{1,k}^U)$, $k=1,\cdots,K_2$, where $K_2\le \exp(C_6/\delta)$. Then for any $f_1\in \mathcal G_1$ there exists $k$ such that $f_{1,k}^L\le f_1\le f_{1,k}^U$ with $\|f_{1,k}^L-f_{1,k}^U\|_2\le \delta$. The bracketing number of $\mathcal G_2$ is the same as the bracketing number of $\mathcal G_1$.

Consequently, we obtain the bracketing entropy of $\Delta, \mathcal G$, with
$$\log K_1\le C_5\log \delta^{-1},$$ and 
$$\log K_2\le C_6/\delta.$$

Therefore, for any $$p_{w}(x,y)=\xi_{00}+\xi_{10}f_1(x)+\xi_{01}f_2(y)+\xi_{11}f_1(x)f_2(y),$$ there exist $i_{00},i_{01},i_{10},i_{11},j_1$, and $j_2$ such that 
$$\begin{aligned}
\xi_{00,i_{00}}^L&\le \xi_{00}\le \xi_{00,i_{00}}^U,
\quad \xi_{01,i_{01}}^L\le \xi_{01}\le \xi_{01,i_{01}}^U,\\
\xi_{10,i_{10}}^L&\le \xi_{10}\le \xi_{00,i_{10}}^U,\quad
\xi_{11,i_{11}}^L\le \xi_{11}\le \xi_{11,i_{11}}^U,\\
f_{1,j_1}^L(x)&\le f_1(x)\le f_{1,j_1}^U(x),
\quad \mbox{and}\quad 
f_{2,j_2}^L(y)\le f_2(y)\le f_{2,j_2}^U(y).
\end{aligned}$$

Let $I=\{(i_{00},i_{01},i_{10},i_{01},j_1,j_2)\}$, where $i_{\cdot\cdot}=1,\cdots K_1$, and $j_{\cdot}=1,\cdots,K_2$, then $$\log N_{[]}(\delta,\mathcal F,\|\cdot\|_2)\lesssim \log |I|=\log K_1^4K_2^2\le C_7\delta^{-1}.$$

To sum up, we have
$$\log N_{[]}(\delta,\mathcal{\bar F}^{1/2},\|\cdot\|_2)\lesssim \log N_{[]}(\delta,\mathcal F,\|\cdot\|_2)\lesssim C_7\delta^{-1}.$$

\subsubsection{Calculating the lower bound}
The lower bound calculation is motivated by \cite{blei2007metric} and \cite{gao2007entropy}. These papers focus on the bounded monotone functional space, while we focus on the monotone density functional space. 
Besides, the functional space we work with is $\bar{\mathcal{F}}^{1/2}$ instead of $\mathcal F$. 


To calculate the lower bound of $N_{[]}(\delta,\bar{\mathcal{F}}^{1/2}, \|\cdot\|_2)$, we need to determine the packing number. The $\epsilon$ packing number of a functional space $\mathcal G$ is defined to be the minimal number of disjoint $g_i\in \mathcal G$ such that when $i\ne j$, $d(g_i,g_j)>\epsilon$. 
To generate $\delta$ packing of $\bar{\mathcal F}^{1/2}$, we divide the support of $f\in \bar{\mathcal F}^{1/2}$, $(0,1)$, into three parts. The first part of $f$ is the constant $3$, the second part of $f$ is a step function constructed based on \cite{blei2007metric} and \cite{gao2007entropy}, and the third part of $f$ is another constant so that $f$ is a density function. The three components are calibrated so that 
$f$ is non-increasing. 

Consider the following sequence $$A=\{ a: a=(a_1,a_2,\cdots,a_i,\cdots,a_K)\text{ where } a_i\in\{0,1\},i=1,\cdots,K\},$$
where $K=\lfloor \frac{1}{\sqrt{6912}\delta}\rfloor$, where $\lfloor x \rfloor$ is the biggest interger that is no greater than $x$. Assume that $\delta$ is small enough so that $K$ is large enough. Define
$$g_{ a}(x)=\begin{cases}3& 0\le x<\frac{1}{12}\\\frac{4}{3}\{\frac{K-k+1}{K}+\frac{a_k}{K}\}+\frac{1}{2}& \frac {k-1} {2K}+\frac{1}{12}\le x<\frac{k}{2K}+\frac{1}{12}\\t_{ a}&\frac{7}{12}\le x<1.\end{cases}$$
where $1\le k\le K$, and $t_{{a}}$ is choosen such that $g_{ a}(x)$ is a density function. In other words, $\int_0^1 g_{ a}(x)dx=1$. 
We obtain 
$$
\begin{aligned}
t_{ a}&=(1-\frac{7}{12})^{-1}\left\{1-\frac{3}{12}-\int_\frac{1}{12}^\frac{7}{12}g_{ a}(x)dx\right\}\\&=\frac{12}{5}\{\frac{3}{4}-\frac{1}{4}-\frac{1}{K}\sum\limits_{k=1}^{K}\frac{2}{3}(\frac{K-k+1}{K}+\frac{a_k}{K})\}\\&=\frac{2}{5}-\frac{4}{5K}-\frac{8}{5K^2}\sum\limits_{k=1}^{K}a_k.
\end{aligned}
$$
Note that $g_{ a}(\frac{1}{12}+)=\frac{4}{3}(1+\frac{a_1}{K})+\frac{1}{2}<3=g_{ a}(\frac{1}{12} -)$ and $g_{ a}(\frac{7}{12}+)<\frac{2}{5}<\frac{1}{2}<g_{ a}(\frac{7}{12}-)$, and $g_{ a}(x)$ is nonincreasing in $(\frac{1}{12},\frac{7}{12})$. This implies that $g_{ a}(x)$ is non-increasing in $(0,1)$.

For $ a\ne {a'}$, ${a}, {a'}\in A$, define the Hamming distance
$$d_{h}({a},{a'})=\mbox{card}\{k:a_k\ne a'_{k},k=1,\cdots,K\},$$ where $\mbox{card}$ means cardinality. For example $d_h((0,0,0),(1,1,0))=2$.
For each $ a\in A$, consider the ball $$B({a})=\{{a'}\in A,d_{h}({a},{a'})\le\lfloor K/16\rfloor\}.$$
For $0\le i\le \lfloor K/16\rfloor$, the number of points $a'$ with $d_h(a,a')=i$ is $\begin{pmatrix}
K\\i
\end{pmatrix}.$

Using the fact that 
$$\begin{pmatrix}
K\\i
\end{pmatrix}\le (\frac{Ke}{i})^i \quad\text{and} \quad (16e)^{1/16}<\sqrt{2},$$
we have
$$
\begin{aligned}
|B( a)|&=|\bigcup\limits_{i=0}^{\lfloor K/16\rfloor}\{a':d_h(a,a')=i\}|
\\&\le \sum\limits_{i=0}^{\lfloor K/16\rfloor}\begin{pmatrix}K\\i\end{pmatrix}\\&<\sum\limits_{i=0}^{\lfloor K/16\rfloor}\begin{pmatrix}
K\\\lfloor K/16\rfloor
\end{pmatrix}\\&=(1+\frac{K}{16})\begin{pmatrix}
K\\\lfloor K/16\rfloor
\end{pmatrix}\\&\le (1+\frac{K}{16})(\frac{Ke}{K/16})^{K/16}\\&=(1+\frac{K}{16})\{ (16e)^{1/16}\}^{K}.\end{aligned}$$
So that when $K$ is large enough, we have
$$
|B({a})|\le (1+\frac{K}{16})\{ (16e)^{1/16}\}^{K}=(\sqrt{2})^K[(1+\frac{K}{16})\cdot\{\frac{(16 e)^{1/16}}{\sqrt{2}}\}^K]\le (\sqrt{2})^K=2^{K/2}.
$$
Each ball contains less than $2^{K/2}$ functions. Each of the $|A|=2^K$ functions must be contained in at least one ball so that there are at least $2^{K/2}$ disjoint balls.

Let $A'\subset A$ such that $D=\{B( a), a\in A'\}$ is a set of disjoint balls. Then we have $|D|\ge 2^{K/2}.$
 Randomly choose two $B( a),B( a')\in D$ with $ a\ne  a'$, then $d_{h}( a, a')>\lfloor K/16\rfloor$. Let $f_{ a}(x,y)=g_{ a}(x),f_{ a'}(x,y)=g_{{a'}}(x)$, so that $f_a(x,y)$ has the same value for different $y$ when $x$ is fixed. Notice that for any non-increasing density function $f(x,y)$ supported in $(0,1)^2$, and for any $0<a,b<1,$ we have 
$$1=\int_0^1\int_0^1f(x,y)dxdy\ge\int_0^a\int_0^bf(x,y)dxdy\ge \int_0^a\int_0^bf(a,b)dxdy=abf(a,b),$$
then $f(a,b)\le\frac{1}{ab}$.

We have
\[
\begin{aligned}
\int_0^1&\int_0^1\left|\sqrt{\frac{f_{w_0}(x,y)+f_{ a}(x,y)}{2}}-\sqrt{\frac{f_{w_0}(x,y)+f_{ a'}(x,y)}{2}}\right|^2dxdy\\&
=\frac{1}{2}\int_0^1\int_0^1\left|\frac{|g_{ a}(x)-g_{{a}'}(x)|}{\sqrt{f_{w_0}(x,y)+f_{ a}(x,y)}+\sqrt{f_{w_0}(x,y)+f_{a'}(x,y)}}\right|^2dxdy\\&
\ge\frac{1}{2} \int_0^1\int_0^1\left|\frac{|g_{ a}(x)-g_{{a}'}(x)|}{2\sqrt{\frac{2}{xy}}}\right|^2dxdy\\
&=\frac{1}{16}\int_0^1 y\int_{0}^1 x|g_{ a}(x)-g_{{a}'}(x)|^2dxdy\\&
\ge \frac{1}{32}\int_{\frac{1}{12}}^{\frac{7}{12}}x|g_{ a}(x)-g_{{a}'}(x)|^2dx\\&
\ge \frac{1}{32}\int_{\frac{1}{12}}^{\frac{7}{12}}\frac{1}{12}|g_{ a}(x)-g_{{a}'}(x)|^2dx\\&
\ge \frac{1}{384} \sum\limits_{k=1}^K I\{a_k\ne a'_k\}\frac{1}{2K}(\frac{4}{3K})^2
\\&\ge \frac{1}{6912} K^{-2}.
\end{aligned}
\]
Recall that $K=\lfloor\frac{1}{\sqrt{6912}\delta}\rfloor\le \frac{1}{\sqrt{6912}\delta}$, we obtain
$$
\int_0^1\int_0^1\left|\sqrt{\frac{f_{0}(x,y)+f_{a}(x,y)}{2}}-\sqrt{\frac{f_{0}(x,y)+f_{a'}(x,y)}{2}}\right|^2dxdy\ge \delta^2.
$$

The number of $\delta$-packing is greater than $2^{K/2}$, or $2^{1/\sqrt{27648}\delta}$, and the bracketing entropy satisfies 
$$\log N_{[]}(\delta,\bar {\mathcal F}^{1/2},\|\cdot\|_2)\ge C_8/\delta.$$

\subsection{Proof of Theorem \ref{main_consistency}}
\label{sec:bracket_proof}
We use the following Proposition, adapted from Theorem 7.4 of \cite{van2000empirical}, to show the consistency of $p_{\hat w_n}$ based on Hellinger distance.
\begin{proposition}
\label{emp}
Let
$$J_B(\delta,\mathcal{\bar F}^{1/2},\|\cdot\|_2)=\int_{0}^\delta\sqrt{\log N_{[]}(\epsilon,\mathcal{\bar F}^{1/2}, \|\cdot\|_2)}d\epsilon,$$
where 
$$\bar{\mathcal{F}}^{1/2}=\{\sqrt{\frac{p_w+p_{w_0}}{2}}:w\in \Omega\}.$$

Take $\Psi(\delta) \geq J_B\left(\delta, \bar{\mathcal F}^{1/2}, \|\cdot\|_2\right)$ in such a way that $\Psi(\delta) / \delta^2$ is a non-increasing function of $\delta$. Then for a universal constant $\tilde c$ and for
$$
\sqrt{m} \delta_m^2 \geq \tilde c \Psi(\delta_m),
$$
we have for all $\delta \geq \delta_m$
$$
{P}\left(h\left(p_{\hat{w}}, p_{w_0}\right)>\delta\right) \leq \tilde c \exp \left(-\frac{m \delta^2}{\tilde c^2}\right),
$$
where $\hat{w}$ is defined in (\ref{mle}). 
\end{proposition}

From Proposition \ref{bracketing_2}, we have 
$$
\log N_{[]}(\epsilon,\mathcal{\bar F}^{1/2}, \|\cdot\|_2)\le C_1/\epsilon.
$$
Then, we obtain 
$$
\begin{aligned}
J_B(\delta,\mathcal{\bar F}^{1/2},\|\cdot\|_2)&=\int_{0}^\delta\sqrt{\log N_{[]}(\epsilon,\mathcal{\bar F}^{1/2}, \|\cdot\|_2)}d\epsilon\\&\le \int_0^\delta \sqrt{C_1/\epsilon} d\epsilon\\&=2\sqrt{C_1}\delta^{1/2}.
\end{aligned}
$$
Define $\Psi(\delta)=2\sqrt{C_1}\delta^{1/2}$, and $\delta_m=C_3m^{-1/3}$,where $C_3>0$ is a positive value, then we have
$$
\sqrt{m}\delta_m^2=C_3^{2}m^{-1/6}= \frac{C_3^{3/2}}{2\sqrt{C_1}}\{2\sqrt{C_1}(C_3m^{-1/3})^{1/2}\}=\tilde c\Psi(\delta_m),
$$
where $\tilde c=\frac{C_2^{3/2}}{2\sqrt{C_1}}$. By applying Proposition \ref{emp}, 
$\exists C_4,C_5>0$ such that
$$
{P}\left(h\left(p_{\hat{w}}, p_{w_0}\right)>C_3m^{-1/3} \right)\leq C_4 \exp \left[-C_5m^{1/3}\right].
$$

\subsection{Proof of Corollary \ref{parameter consistency us}}
\label{sec: consitency proof}
We adopt Lemma 5.2 in \citep{van1993hellinger} 
to show the consistency of $w$. 

\begin{proposition}
    \label{parameter consistency}
    Let $\Omega \subset \Omega^*$, where $\left(\Omega^*, \tau\right)$ is a compact metric space. Suppose $w \rightarrow p_w, w \in \Omega^*$ is $\mu$-almost everywhere continuous and that $w_0$ is identifiable. Then, if $h\left(f_{w_m}, f_{w_0}\right) \rightarrow 0$ for some sequence $\left\{w_m\right\}$, where $h$ is the Hellinger distance, we have $\tau\left(w_m, w_0\right) \rightarrow 0$.
\end{proposition} 

In proposition \ref{identify}, we have shown the identifiability of $w_0$. To use Proposition \ref{parameter consistency}, we need to find a compact space $\Omega^*$, and show that $w\to p_w$ is a continuous mapping.

Define a bigger parameter space $\Omega^*=\Delta\times \mathcal G^*\times \mathcal G^*$, where $$\mathcal G^ *=\{f: f \text{ is non-increasing on }(0,1],\|f\|_1\le 1,0\le f \le F \text{ a.e.}\}.$$ Here $f\in \mathcal G^*$ is no longer required to be a density function. 
The corresponding product metric $\tau$ is  
 \begin{align*}
 \tau(w,\hat w)&=|\xi_{00}-\hat\xi_{00}|+|\xi_{01}-\hat\xi_{01}|+|\xi_{10}-\hat\xi_{10}|+|\xi_{11}-\hat \xi_{11}|\\&\qquad+\int _0^1|f_1(x)-\hat f_1(x)|dx +\int _0^1|f_2(y)-\hat f_2(y)|dy.
\end{align*}

The following two lemmas will show that $w\to p_w$ is continuous, and $\Omega^*$ is compact.
\begin{lemma}
    \label{continuous p}
     Let $w=(\xi_{00},\xi_{01},\xi_{10},\xi_{11},f_1(x),f_2(y))$. The mapping $p:\Omega^*\to\mathcal F$, where $$p(w)(x,y)= p_w(x,y)=\xi_{00}+\xi_{10}f_1(x)+\xi_{01}f_2(y)+\xi_{11}f_1(x)f_2(y),$$ is a continuous mapping, where we use the product metric $\tau$ in $\Omega^*$ and the 
$L^1$ metric in $\mathcal F$.
\end{lemma}

\begin{lemma}
    \label{compact theta}
     $\Omega^*$ is a compact space.
\end{lemma}
Combining identifiability (Proposition \ref{identify}), continuous mapping (Lemma $\ref{continuous p}$), and compact space (Lemma $\ref{compact theta}$), 
we have $\tau(\hat w_n,w_0)\overset p\to 0$.

\subsubsection{Proof of Lemma \ref{continuous p}}
Now we show that $w\to p_w,w\in \Omega^*$ is continuous. Assume that $\tau(w^*,w)\le \delta$, using the fact that
$$
\int A(x)B(y)dxdy=\int A(x)dx\int B(y)dy, 
$$
we obtain
$$
\begin{aligned}
\|\xi_{11}f_1(p_1)f_2(p_2)&-\xi_{11}^*f_1^*(p_1)f_2^*(p_2)\|_1
\\&\le\xi_{11}\|f_1(p_1)f_2(p_2)-f_1^*(p_1)f_2^*(p_2)\|_1+\|f_1^*(p_1)f_2^*(p_2)\|_1\cdot|\xi_{11}-\xi_{11}^*|\\&\le \|f_1(p_1)f_2(p_2)-f_1^*(p_1)f_2^*(p_2)\|_1+\delta\\&\le \delta+ \|f_1(p_1)(f_2(p_2)-f_2^*(p_2))\|_1+\|f_2^*(p_2)(f_1(p_1)-f_1^*(p_1))\|_1\\&\le \delta+\|f_2-f_2^*\|_1+\|f_1-f_1^*\|_1\\&\le 3\delta,
\end{aligned}
$$
In the same way, we can show that 
\begin{align}
|\xi_{00}-\xi_{00}^*|&\le \delta\nonumber\\
\|\xi_{10}f_1-\xi_{10}^*f_1^*\|_1&\le 2\nonumber\delta\\
\|\xi_{01}f_2-\xi_{01}^*f_2^*\|_1&\le 2\delta \label{part_dist}
\end{align}

Therefore
\begin{align}
\|p_{w}-p_{w^*}\|_1&\le |\xi_{00}-\xi_{00}^*|+\|\xi_{10}f_1-\xi_{10}^*f_1^*\|_1+\|\xi_{01}f_2-\xi_{01}^*f_2^*\|_1 \nonumber\\
&\qquad+\|\xi_{11}f_1(p_1)f_2(p_2)-\xi_{11}^*f_1^*(p_1)f_2^*(p_2)\|_1\nonumber\\
&\le 8\delta,\label{total_dist}\end{align}
implying that $w\to p_w$ is continuous.

\subsubsection{Proof of Lemma \ref{compact theta}}
To show $\mathcal G^*$ is compact, it is sufficient to show that $\mathcal G^*$ is totally bounded and complete.
$\mathcal G^*$ is totally bounded when $\mathcal G^*$ can be covered by a finite number of $\epsilon$-balls. $\mathcal G^*$ is complete when every Cauchy sequence in it converges.

\underline{\textbf{Step 1}}: $\mathcal G^*$ is totally bounded.

The proof is an application of Kolmogorov–Riesz Theorem \citep{hanche2010kolmogorov}. The Theorem says that $\mathcal G^*$ is totally bounded with metric $L^1$ if and only if 
\begin{enumerate}
    \item $\mathcal G^*$ is bounded,
    \item for any $\epsilon>0$, there exists $R,$ and $\delta$, such that for any $g\in \mathcal G^*$, $|y|<\delta$, we have
$$\int_{|x|>R} g(x)< \epsilon,\quad \text{and}\quad\int_{-\infty}^\infty |g(x+y)-g(x)| dx < \epsilon.$$
    
\end{enumerate}

\begin{remark}
    $\mathcal{G}_1^*$ being bounded means that for any $g,g'\in \mathcal G^*$, $\int|g(x)-g'(x)|dx<C_9$, where $C_9$ is a constant. It is not the same as $\forall g\in \mathcal{G}^*,g$ is bounded.
\end{remark}

The remaining part is to verify the above conditions.
\begin{enumerate}
    \item $\mathcal G^*$ is bounded because for any $g\in \mathcal G^*$, we have $\|g-0\|_1\le 1$, so that for any $g,g'\in \mathcal G^*$, $\|g-g'\|_1\le \|g-0\|_1+\|g'-0\|_1\le 2 $.
    \item Taking $R=1$ will guarantee that $\int_{|x|>R} g(x)dx=0<\epsilon$.
    \item 
    Take $\delta$ such that $2\int_0^\delta F(x)\le \epsilon$. It is achievable because $F\in L^1$.
    Without loss generality, we can assume that $y>0$, so that
    $$
    \begin{aligned}
    \int_{-\infty}^\infty |g(x+y)-g(x)| dx&=\int_{-y}^0|g(x+y)-g(x)| dx+ \int_{0}^\infty|g(x+y)-g(x)| dx
    \\&= \int_{-y}^0g(x+y) dx+\int_0^\infty \{g(x)-g(x+y)\}dx
    \\&=2\int _0^y g(x)dx\\&\le 2\int _0^y F(x)dx\\&\le 2\int_0^\delta F(x)dx\\&\le \epsilon.
    \end{aligned}
    $$
\end{enumerate}

After verifying above conditions, by Kolmogorov–Riesz Theorem, $\mathcal G^*$ is totally bouned.

\underline{\textbf{Step 2}}: $\mathcal G^*$ is complete.

If $\{f_m\}$ is a Cauchy sequence in $\mathcal G^*\subset L^1$, and $f_m\in \mathcal G^*$, then there exists $f\in L^1$ such that $\|f_m-f\|_1\to 0$, because of the completeness of $L^1$ space. By Corollary 2.3 in \citep{stein2009real}, there exists a subsequence $f_{m_k}$ such that $f_{m_k}(x)\to f(x)$ almost everywhere. As $0\le f_m\le F$, $f_m$ is non-increasing and $\|f_m\|\le 1$, we have $f\le F$, $f$ is non-increasing and $0\le f\le F$ almost everywhere, which implies that $f\in \mathcal G^*$.
So that $\mathcal G^*$ is complete.

By showing that $\mathcal G^*$ is totally bounded and complete, we prove that $\mathcal G^*$ is compact. The proof of compactness of $\mathcal G^*$ is exactly the same as that of $\mathcal G^*,$ so we omit the details.

\subsection{Proof of Theorem \ref{Minimax}}
\label{minimax_proof}
Let $\mathcal F$ be the functional space defined in (\ref{eq:mixture}). Define the bounded functional space $\mathcal F_{C_l,C_u}=\{p\in \mathcal F,C_l\le p\le C_u\}$, where $C_u> C_l$.
As $\mathcal F_{C_l,C_u}\subseteq \mathcal F$, we have
$$
\mathcal R\ge \inf\limits_{\tilde p_m\in \bar{\mathcal F}_m}\sup\limits_{p_{w_0}\in \mathcal F_{C_l,C_u}} E_{p_{w_0}}[h^2(\tilde p_m,p_{w_0})],
$$
where $\mathcal R$ is the true minimax rate defined in (\ref{eq:minimax}), and $h(p,p')$ is the Hellinger distance between function $p$ and $p'$.

We use the result of Theorem 1 in \cite{yang1999information} to get the lower bound of the minimax rate. We define $H(\epsilon, {\mathcal F}_{C_l,C_u},d_{KL})$ as the $\epsilon$ bracketing entropy for the space $\mathcal {F}_{C_l,C_u}$  with respect to $d_{KL}$, where $d_{KL}$ is square root of the KL divergence, and
$$
d^2_{KL}(p_1,p_2)=\int p_1(x)\log \frac{p_1(x)}{p_2(x)} dx.
$$
We further define $M(\epsilon,{\mathcal F}_{C_l,C_u},h)$ as the $\epsilon$ packing entroy for the space $\mathcal {F}_{C_l,C_u}$  with respect to the Hellinger distance $h$.

Here is a modified version of the Theorem 1 in  \cite{yang1999information}.
\begin{theorem}
\label{thm: Fato}
    Let $\epsilon_m$ be determined by $\epsilon_m^2=H(\epsilon_m)/m$, where $H(\epsilon)$ is an upper bound of $H(\epsilon, {\mathcal F}_{C_l,C_u},d_{KL})$. Let $\epsilon_{m,d}$ be the solution of $M(\epsilon_{m,d})=4m\epsilon^2_m+2\log 2$, where $M(\epsilon)$ is a lower bound of $M(\epsilon,{\mathcal F}_{C_l,C_u},h).$
    We have
    $$
    \inf _{\tilde p_m\in \mathcal{\bar F}_m} \sup _{p_{w_0} \in \mathcal{F}} E_{p_{w_0}} h^2(\tilde p_m,p_{w_0}) \geq\frac{1}{8}\epsilon_{m, d}^2.
    $$ 
\end{theorem}

Next, we use Theorem (\ref{thm: Fato}) to derive the lower bound of our minimax rate. When $p_1,p_2\in \mathcal F_{C_l,C_u}$, we find that the upper bound of the KL divergence between two functions is 
$$
\begin{aligned}
&d^2_{KL}(p_1,p_2)\\
=&\int p_1(x)\log \frac{p_1(x)}{p_2(x)} dx\\\le&\int \{\frac{p_1(x)}{p_2(x)}-1\} p_1(x)dx
\\=&\int\frac{p^2_1(x)}{p_2(x)}-2p_1(x)+p_2(x) dx
\\=&\int \frac{\{p_1(x)-p_2(x)\}^2}{p_2(x)}dx
\\=&\int \frac{\{\sqrt{p_1(x)}+\sqrt{p_2(x)}\}^2}{p_2(x)}\{\sqrt{p_1(x)}-\sqrt{p_2(x)}\}^2dx
\\\le& \frac{8C_u}{C_l} \frac{1}{2}\int \{\sqrt{p_1(x)}-\sqrt{p_2(x)}\}^2dx.
\\\le& \frac{8C_u}{C_l} h^2(p_1,p_2).
\end{aligned}$$
Therefore, we obtain that  
$$
H(\epsilon,{\mathcal F}_{C_l,C_u},d_{KL})\le H(\frac{C_l^{1/2}}{2\sqrt{2}C_u^{1/2}}\epsilon,{\mathcal F}_{C_l,C_u},h)\le K_1 \epsilon^{-1}:=H(\epsilon),
$$
where $K_1$ is a constant.
The lower bound of the KL divergence is calculated by
$$
\begin{aligned}
&d^2_{KL}(p_1,p_2)\\
=&-2\int p_1(x)\log \sqrt\frac{p_2(x)}{p_1(x)} dx\\\ge&-2\int \{\sqrt\frac{p_2(x)}{p_1(x)}-1\} p_1(x)dx
\\=&-2\int\{\sqrt{p_1(x)p_2(x)}-p_1(x)\}dx
\\=&\int\{p_1(x)+p_2(x)- 2\sqrt{p_1(x)p_2(x)}\}dx
\\=&\int \{\sqrt{p_1(x)}-\sqrt{p_2(x)}\}^2dx
\\=&2h^2(p_1,p_2).
\\\ge &h^2(p_1,p_2).
\end{aligned}$$
Therefore
$$
M(\epsilon,{\mathcal F}_{C_l,C_u},d_{KL})\ge M(\epsilon,{\mathcal F}_{C_l,C_u},h)\ge K_2\epsilon^{-1}:=M(\epsilon).
$$

By solving equations $\epsilon_m^2=H(\epsilon_m)/m$, and $M(\epsilon_{m,d})=4m\epsilon_{m}^2+2\log 2$, we find that $\epsilon_m= K_1^{1/3}m^{-1/3}:=K_3m^{-1/3}$, and $\epsilon_{m,d}=\frac{K_2}{4K_3^2m^{1/3}+2\log 2}>K_4 m^{-1/3}$. Here $K_4=\frac{K_2}{8K_3^2}$ because when $\epsilon_{m}\asymp m^{-1/3}$, $4m\epsilon_{m}^2\ge 2\log 2$, and  $K_2\epsilon^{-1}_{m,d}= 4m\epsilon_{m}^2+ 2\log 2\le 8m\epsilon_{m}^2=K_3^2m^{1/3}.$

Then by Theorem (\ref{thm: Fato}), we find that 
\[
\inf _{\tilde p_m\in \bar {\mathcal{F}}_m} \sup _{p_{w_0} \in \mathcal{F}} E_{p_{w_0}}[h^2(\tilde p_m,p_{w_0})] \geq\inf _{\tilde p_m\in \bar{\mathcal{F}}_m} \sup _{p_{w_0} \in \mathcal{F}_{C_l,C_u}} E_{f_0}[h^2(\tilde p_m,p_{w_0})] \geq \frac{K_4^2}{8}m^{-2/3}.
\]

For the upper bound of $\mathcal R$, we take $\hat p_{m}$ to be the maximum likelihood estimation for the density $p_{w_0}$. Then by Theorem \ref{main_consistency}, $\forall p_{w_0}\in \mathcal F$ there exists $C_{10}$ such that
$$
E_{p_{w_0}}[h^2(\hat p_{m}, p_{w_0})]\le C_{10}m^{-2/3}.
$$
Therefore, we obtain that 
$$
\inf _{\tilde p_m\in \bar {\mathcal{F}}_m} \sup _{p_{w_0} \in \mathcal{F}} E_{p_{w_0}}[h^2(\tilde p_m,p_{w_0})] \leq E_{p_{w_0}}[h^2(\hat p_m, p_{w_0})]\le C_{10}m^{-2/3}.
$$

Finally, the minimax rate
$$
\inf _{\tilde p_m\in \bar {\mathcal{F}}_m} \sup _{p_{w_0} \in \mathcal{F}} E_{p_{w_0}}[h^2(\tilde p_m,p_{w_0})]\asymp m^{-2/3}.
$$

\subsection{Proof of Theorem \ref{fdr_theorem}}
\label{sec: two_lfdr_proof}
We first prove that when using the MLE, $\hat w_n$, to estimate the local false discovery rate (Lfdr), 
the average absolute difference between estimated Lfdr and true Lfdr converges to $0$
based on the following proposition. 
\begin{proposition}
\label{lfdr consistency}
    Under (C1) and (C2), if $\hat w=(\hat\xi_{00},\hat\xi_{10},\hat \xi_{01},\hat \xi_{11},\hat f_{1},\hat f_{2})$ is the MLE, 
    and $w_0=(\xi_{00},\xi_{10},\xi_{01},\xi_{11},f_{1},f_{2})$ is the true parameter, 
    we have
    $$
    \frac{1}{m}\sum\limits_{i=1}^m\left|\widehat{\operatorname{Lfdr}}(p_{i1},p_{i2})-\operatorname{Lfdr}(p_{i1},p_{i2}) \right|\overset{p}{\to} 0,
    $$
    where
    \begin{equation*}
    \widehat{\operatorname{Lfdr}}(p_{i1},p_{i2})=\frac{\hat \xi_{00} +\hat \xi_{10} \hat f_{1}\left(p_{1 i}\right)+\hat\xi_{01} \hat f_{2}\left(p_{2 i}\right)}{\hat \xi_{00} +\hat \xi_{10} \hat f_{1}\left(p_{1 i}\right)+\hat\xi_{01} \hat f_{2}\left(p_{2 i}\right)+\hat \xi_{11} \hat f_{1}\left(p_{1 i}\right) \hat f_{2}\left(p_{2 i}\right)},
    \end{equation*}
    and
    $$
      \operatorname{Lfdr}(p_{i1},p_{i2})= \frac{ \xi_{00} +\xi_{10} f_{1}\left(p_{1 i}\right)+\xi_{01} f_{2}\left(p_{2 i}\right)}{\xi_{00} +\xi_{10} f_{1}\left(p_{1 i}\right)+\xi_{01} f_{2}\left(p_{2 i}\right)+\xi_{11} f_{1}\left(p_{1 i}\right) f_{2}\left(p_{2 i}\right)}.
    $$
\end{proposition}

\subsubsection{Proof of Proposition \ref{lfdr consistency}}
Let $\mathcal H=\{\frac{\xi_{00}+\xi_{10}f_1(p_1)+\xi_{01}f_2(p_2)}{\xi_{00}+\xi_{10}f_1(p_1)+\xi_{01}f_2(p_2)+\xi_{11}f_1(p_1)f_2(p_2)}:w\in \Omega\}$.
For $w,w^*\in \Omega$ and $\tau(w,w^*)<\delta$, notice that $\xi_{00}\ge c, \xi_{00}^*\ge c, p_w\ge c,$ and $p_{w^*}\ge c$. We have
$$
\begin{aligned}\left\|\frac{p_w-\xi_{11}f_1f_2}{p_w}-\frac{p_{w^*}-\xi^*_{11}f^*_1f^*_2}{p_{w^*}}\right\|_1&=\left\|\frac{\xi_{11} f_1f_2 p_{w^*}-\xi^*_{11} f_1^*f_2^*p_{w}}{p_w p_{w^*}}\right\|_1
\\&=\left\|\frac{\xi_{11} f_1f_2 p_{w^*}-\xi_{11}f_1f_2p_{w}+\xi_{11}f_1f_2p_{w}-\xi^*_{11} f_1^*f_2^*p_{w}}{p_w p_{w^*}}\right\|_1
\\&\le \left\|\frac{\xi_{11}f_1f_2|p_{w^*}-p_{w}|+p_{w}|\xi_{11}f_1f_2-\xi_{11}^*f_1^*f_2^*|}{p_{w}p_{w^*}}\right\|_1\\&\le \left\|\frac{\xi_{11}f_1f_2}{p_w}\frac{|p_w-p_{w^*}|}{p_{w^*}}\right\|_1+\left\|\frac{1}{p_{w^*}}|\xi_{11}f_1f_2-\xi_{11}^*f_1^*f_2^*\right\|_1\\&\le \frac{\|p_{w}-p_{w^*}\|_1+\|\xi_{11}f_1f_2-\xi_{11}^*f_1^*f_2^*\|_1}{c}\\&\le \frac{11}{c}\delta.\end{aligned}
$$
The last inequality holds because of (\ref{part_dist}) and (\ref{total_dist}).
Therefore, when $\tau(\hat w_n,w_0)\overset{p}{\to} 0$, we obtain that 
$$
\|\widehat{\mbox{Lfdr}}-\mbox{Lfdr}\|_1=\left\|\frac{p_{\hat w}-\hat \xi_{11}\hat f_1\hat f_2}{p_{\hat w}}-\frac{p_{w_0}-\xi_{0,11}f_{0,1}f_{0,2}}{p_{w_{0}}}\right\|_1\overset{p}{\to} 0.
$$

In addition, by Holder's inequality, we have
$$
\begin{aligned}
\int |\widehat{\mbox{Lfdr}}(p_{i1},p_{i2})&-\mbox{Lfdr}(p_{i1},p_{i2})| p_{w_0}(p_{i1},p_{i2}) dp_{i1}dp_{i2}
\\&\le \|p_{w_0}\|_{2}\left\{\int \left|\widehat{\mbox{Lfdr}}(p_{i1},p_{i2})-\mbox{Lfdr}(p_{i1},p_{i2})\right|^{2}dp_{i1}dp_{i2}\right\}^{1/2}
\\&\le\|p_{w_0}\|_{2}\left\{\int 2\left|\widehat{\mbox{Lfdr}}(p_{i1},p_{i2})-\mbox{Lfdr}(p_{i1},p_{i2})\right|dp_{i1}dp_{i2}\right\}^{1/2}
\\&\overset{p}{\to} 0.
\end{aligned}
$$
So that $\mathbb {E}|\widehat{\mbox{Lfdr}}-\mbox{Lfdr}|\overset{p}{\to} 0$.

By the Uniformly Law of Large Numbers (Lemma 3.1 in \citep{van2000empirical}), 
we can show that
$$
P_m |\widehat {\mbox{Lfdr}}-\mbox{Lfdr}| \overset{p}{\to} E |\widehat {\mbox{Lfdr}}-\mbox{Lfdr}|\overset{p}{\to} 0.$$

Now we verify the condition of the Uniformly Law of Large Numbers. This theorem requires that $\mathcal H'$ has finite bracketing number in $L_1(P_0)$ metric, where $\mathcal H'$ is defined as
$$\mathcal H'=\{|h-\mbox{Lfdr}(w_0)|: h\in\mathcal H\}.$$

Note that
$$\frac{\xi_{00}+\xi_{10}f_1(p_1)+\xi_{01}f_2(p_2)}{\xi_{00}+\xi_{10}f_1(p_1)+\xi_{01}f_2(p_2)+\xi_{11}f_1(p_1)f_2(p_2)}=1-\frac{\xi_{11}}{\frac{\xi_{00}}{f_1(p_1)f_2(p_2)}+\frac{\xi_{10}}{f_2(p_2)}+\frac{\xi_{01}}{f_1(p_1)}+\xi_{11}}
$$
is a monotone increasing function in $p_1,p_2$, and bounded by 1. By Theorem 1.1 in \cite{gao2007entropy}, the bracketing number of $\mathcal H $ is bounded, and we have
$$N_{[]}(\mathcal{H},\epsilon,\|\cdot\|_2)<\infty.$$ We aim to show that  
$$N_{[]}(\mathcal{H}',\epsilon, \|\cdot\|_2)<\infty,$$ which can be achieved by the following lemma.

\begin{lemma}
    \label{absolute bracketing}
    Let $\mathcal F$ be a functional space, and let $a\ge1$. Define
    $$|\mathcal F|=\{|f|,f\in \mathcal F\}.$$
    If $$N_{[]}(\mathcal F,\epsilon,\|\cdot\|_a)<\infty,$$ we have
    $$N_{[]}(|\mathcal F|,\epsilon,\|\cdot\|_a)\le N_{[]}(\mathcal F,\epsilon,\|\cdot\|_a).$$
\end{lemma}
\begin{proof}
    Let $[f_i^L,f_i^U],i=1,\cdots,K$ be $\epsilon$ bracketing of $\mathcal F$. Define $g^U_i,g^L_i$ in the following way:
    \begin{itemize}
        \item When $f_i^U(x)\ge f_i^L(x)> 0$, $g^U_i(x)=f^U_i(x)$, $g^L_i(x)=f^L_i(x)$;
        \item When $f_i^U(x)\ge 0\ge f_i^L(x),g^U_i(x)=\max\{|f^U_i(x)|,|f^L_i(x)|\}$, $g^L_i(x)=0$;
        \item When $0> f_i^U(x)\ge f_i^L(x)$, $g^U_i(x)=-f^L_i(x)$, $g^L_i(x)=-f^U_i(x)$.
    \end{itemize}
    In all the three cases, we have $|g_i^U(x)-g_i^L(x)|\le |f_i^U(x)-f_i^L(x)|$. When $f_i^L\le f\le f_i^U$, we also have
    $$g_i^L\le |f|\le g_i^U.$$ This implies that $[g_i^L,g_i^U],i=1,\cdots,K$ is also $\epsilon$ bracketing of functional space $|\mathcal F|$.
    Therefore we have 
    $$N_{[]}(|\mathcal F|,\epsilon,\|\cdot\|_a)\le N_{[]}(\mathcal F,\epsilon,\|\cdot\|_a).$$
\end{proof}
By Holder's inequality, for any $g\in L^2$, we have 
$$\|g\|_{L_1(P_0)}=\|g f_{w_0}\|_{1}\le \|g\|_{2} \|f_{w_0}\|_{2}, $$
so that both the space $\mathcal H$, and $\mathcal H'$ have a finite bracketing number in $L_1(P_0)$ metric. 
This satisfies the condition of the Uniformly Law of Large Numbers (Lemma 3.1 in \citep{van2000empirical}). 
So that 
$$
P_m |\widehat {\mbox{Lfdr}}-\mbox{Lfdr}| \overset{p}{\to} E |\widehat {\mbox{Lfdr}}-\mbox{Lfdr}|\overset{p}{\to} 0.
$$
    \textbf{Note:} We find that when $\xi_{00},\xi_{01},\xi_{10}>c'$ ($c'$ is a constant), and when $\tau(w,w^*)<\delta$,
    we have 
    $$
    \begin{aligned}\left\|\frac{f_j}{p_w}-\frac{f_j^*}{p_{w^*}}\right\|_1&=\left\|\frac{f_j p_{w^*}-f_j^*p_{w}}{p_w p_{w^*}}\right\|_1
    \\&= \left\|\frac{f_j p_{w^*}-f_j p_w+f_jp_w-f_j^*p_{w}}{p_w p_{w^*}}\right\|_1
    \\&\le \left\|\frac{f_j|p_{w^*}-p_{w}|+p_{w}|f_j-f_j^*|}{p_{w}p_{w^*}}\right\|_1
    \\&\le \left\|\frac{f_j}{p_w}\frac{|p_w-p_{w^*}|}{p_{w^*}}\right\|_1+\left\|\frac{|f_j-f_{j}^*|}{p_{w^*}}\right\|_1\\&\le \frac{\frac{1}{c'}\|p_{w}-p_{w^*}\|_1+\|f_j-f_{j}^*\|_1}{c'}\\&\le \frac{8/c'+1}{c'}\delta,\end{aligned}
    $$
    where $j=0,1$ and $2$. Here we set $ f^*_0=f_0=1$.
By the same arguments used in the proof of 
proposition \ref{lfdr consistency}, we also show that
   \begin{equation}
    \label{eq:lfdr_part}
       P_m\left|\frac{\hat f_j}{p_{\hat w}}-\frac{f_{0,j}}{p_{w_0}}\right|\overset{p}{\to} 0.
    \end{equation}

\subsubsection{Proof of Theorem \ref{fdr_theorem}}
We present the sketch of the proof as follows.
Define 
$$
\begin{aligned}
R_m(\lambda)&=\sum_{i=1}^m {I}\left\{\operatorname{Lfdr}(p_{i1},p_{i2}) \leq \lambda\right\},\\
V_m(\lambda)&=\sum\limits_{i=1}^m  I\{\operatorname{Lfdr}(p_{i1},p_{i2})\le \lambda\} (1-\theta_{i1}\theta_{i2}),\\
D_{m}(\lambda)&= \sum_{i=1}^m \operatorname{Lfdr}(p_{i1},p_{i2}) {I}\left\{\operatorname{Lfdr}(p_{i1},p_{i2}) \leq \lambda\right\},\\
\end{aligned}
$$
where $R_m(\lambda)$ is the number of rejections, $V_m(\lambda)$ is the number of false discoveries and $D_m(\lambda)$ is the estimated number of false discoveries. 
If we use the estimated local false discovery rate, the corresponding values are $\hat R_m(\lambda),\hat V_m(\lambda),$ and $\hat D_m(\lambda)$, where
$$
\begin{aligned}
\hat R_m(\lambda)&=\sum_{i=1}^m {I}\left\{\widehat{\operatorname{Lfdr}}(p_{i1},p_{i2}) \leq \lambda\right\},\\
\hat V_m(\lambda)&=\sum\limits_{i=1}^m  I\{\widehat{\operatorname{Lfdr}}(p_{i1},p_{i2})\le \lambda\} (1-\theta_{i1}\theta_{i2}),\\
\hat D_{m}(\lambda)&=\sum_{i=1}^m \widehat{\operatorname{Lfdr}}(p_{i1},p_{i2}) {I}\left\{\widehat{\operatorname{Lfdr}}(p_{i1},p_{i2}) \leq \lambda\right\}.\\
\end{aligned}
$$
With a good estimate of the local false discovery rate as that in Proposition \ref{lfdr consistency}, that satisfies 
$$\frac{1}{m} \sum_{i=1}^m\left|\widehat{\operatorname{Lfdr}}(p_{i1},p_{i2})-\operatorname{Lfdr}(p_{i1},p_{i2})\right| \overset p\to  0,$$
we have the following uniform consistency result:
$$\begin{aligned}\sup\limits_{\lambda\in[\lambda_0,1]}|\hat R_m(\lambda)-R_m(\lambda)|&\overset p \to 0,\\\sup\limits_{\lambda\in[\lambda_0,1]}|\hat D_{m}(\lambda)-D_m(\lambda)|&\overset p \to 0,\\\sup\limits_{\lambda\in[\lambda_0,1]}|\hat V_{m}(\lambda)-V_m(\lambda)|&\overset p \to 0,\end{aligned}$$
where $\lambda_0$ is defined in (C5).

Using the same way for proving ``Glivenko–Cantelli Theorem" (\citep{path}), we show that $\hat R_m(\lambda)/m,$ $\hat D_m(\lambda)/m,$ and $\hat V_m(\lambda)/m$ are uniformly consistent. In other words,
$$
\begin{aligned}
\sup\limits_{\lambda\in [\lambda_0,1]}|\hat R_m(\lambda)/m-B_1(\lambda)|&\overset p \to 0,\\
\sup\limits_{\lambda\in [\lambda_0,1]}|\hat D_{m}(\lambda)/m-B_2(\lambda)|&\overset p \to 0,\\
\sup\limits_{\lambda\in [\lambda_0,1]}|\hat V_m(\lambda)/m-B_2(\lambda)|&\overset p \to 0.
\end{aligned}$$
Besides 
we have
$$
\begin{aligned}
\sup\limits_{\lambda\in [\lambda_0,1]}\left|\frac{\hat D_{m}(\lambda)}{\hat R_m(\lambda)}-\frac{B_2(\lambda)}{B_1(\lambda)}\right|&\overset p \to 0,\\
\sup\limits_{\lambda\in [\lambda_0,1]}\left|\frac{V_{m}(\lambda)}{R_m(\lambda)}-\frac{B_2(\lambda)}{B_1(\lambda)}\right|&\overset p \to 0.\\
\end{aligned}
$$
So when $m$ is sufficiently large we have $$P\left\{\left|\frac{\hat D_m(\lambda_0)}{\hat R_m(\lambda_0)}-\frac{B_2(\lambda_0)}{B_1(\lambda_0)}\right|\le \frac{\alpha-B_2(\lambda_0)/B_1(\lambda_0)}{2}\right\}\to 1.$$ So that 
$$P\left\{\frac{\hat D_m(\lambda_0)}{\hat R_m(\lambda_0)}<\alpha\right\}\to 1,\quad{\text{and}}\quad P\{\hat \lambda_m>\lambda_0\}\to 1.$$
Besides from (C5), $P(R_m(\hat \lambda_m)>0)\to 1$, so the false discovery rate can be bounded by
$$
\begin{aligned}
\mbox{FDR}(\hat \lambda_m)&=\mathbb{E}\left(\frac{V_m(\hat\lambda_m)}{R_m(\hat\lambda_m)\vee 1}\right)
\\&=\mathbb{E}\left(\frac{V_m(\hat \lambda_m)}{R_m(\hat\lambda_m)}\right)
\\&=\mathbb{E}\left\{\frac{\hat D_m(\hat\lambda_m)}{\hat R_m(\hat\lambda_m)}+\left(\frac{V_m(\hat\lambda_m)}{R_m(\hat\lambda_m)}-\frac{B_2(\hat\lambda_m)}{B_1(\hat\lambda_m)}\right)+\left(\frac{B_2(\hat\lambda_m)}{B_1(\hat\lambda_m)}-\frac{\hat D_m(\hat \lambda_m)}{\hat R_m(\hat\lambda_m)}\right)\right\}
\\&\le \alpha+ \mathbb{E}\left|\frac{V_m(\hat\lambda_m)}{R_m(\hat\lambda_m)}-\frac{B_2(\hat\lambda_m)}{B_1(\hat\lambda_m)}\right|+\mathbb{E}\left|\frac{\hat D_m(\hat\lambda_m)}{\hat R_m(\hat\lambda_m)}-\frac{B_2(\hat\lambda_m)}{B_1(\hat\lambda_m)}\right|
\\&\le \alpha+ \mathbb{E}\sup\limits_{\lambda\in [\lambda_0,1]}\left|\frac{V_m(\lambda)}{R_m(\lambda)}-\frac{B_2(\lambda)}{B_1(\lambda)}\right|+\mathbb{E}\sup\limits_{\lambda\in [\lambda_0,1]}\left|\frac{\hat D_m(\lambda)}{\hat R_m(\lambda)}-\frac{B_2(\lambda)}{B_1(\lambda)}\right|
\end{aligned}
$$
Based on the following lemma, we prove the result that
$$\limsup\limits_{m\to\infty}\mbox{FDR}(\hat \lambda_m)\le \alpha.$$
\begin{lemma}
If $|X_m|\le C_{11}$, and $X_m\overset p\to 0$, then we have $$\limsup\limits_{m\to\infty} \mathbb E[X_m]\le0.$$
\end{lemma}
\begin{proof}
We take a subsequence of $X_m$ such that $\limsup_m \mathbb E[X_m]=\lim_k \mathbb E[X_{m_k}]$. As $X_m\overset p \to 0$, we can further select a subsequence such that $X_{m_{k_l}}\overset {a.s.}\to 0$ when $l\to\infty$. (see p. 172 of \cite{path})

Then by Fatou's Lemma, because $|X_m|\le C_{11}$, we have $$\limsup_m\mathbb E[X_m]=\limsup_k\mathbb E[X_{m_k}]\le \mathbb E[\limsup_l X_{m_{k_l}}]=0.$$
\end{proof}
\subsection{Proof of Proposition \ref{thm:power}}
\label{sec: oracle_pow_proof}

For the rejection criteria $\delta(p_1,p_2)$, the power is
$$
E(\delta(p_1,p_2)\mid \theta_1\theta_2=1)=\int \delta(p_1,p_2)f_1(p_1)f_2(p_2)dp_1dp_2.
$$
The marginal false discovery rate, which is asymptotically the same as the false discovery rate when $m\to \infty$, is
{
\small
$$
\begin{aligned}
&E(1-\theta_1\theta_2\mid \delta(p_1,p_2)=1)\\=&\frac{P(\theta_1=0,\theta_2=0,\delta(p_1,p_2)=1)+P(\theta_1=1,\theta_2=0,\delta(p_1,p_2)=1)+P(\theta_1=0,\theta_2=1,\delta(p_1,p_2)=1)}{P(\delta(p_1,p_2)=1)}.
\end{aligned}
$$
}
Note that 
$$
\begin{aligned}P(\theta_1=i,\theta_2=j,\delta(p_1,p_2)=1)&=\xi_{ij}P(\delta(p_1,p_2)=1\mid \theta_1=i,\theta_2=j)
\\&=\xi_{ij}E(\delta(p_1,p_2)\mid \theta_1=i,\theta_2=j),
\end{aligned}
$$
where $\xi_{ij}=P(\theta_1=i,\theta_2=j)$.
We can obtain that
$$
E(1-\theta_1\theta_2\mid \delta(p_1,p_2)=1)=\frac{\int (\xi_{00}+\xi_{10}f_1(p_1)+\xi_{01}f_2(p_2))\delta(p_1,p_2) dp_1dp_2}{\int (\xi_{00}+\xi_{10}f_1(p_1)+\xi_{01}f_2(p_2)+\xi_{11}f_1(p_1)f_2(p_2))\delta(p_1,p_2) dp_1dp_2}.
$$
We want to find random variable $\delta(p_1,p_2)$ such that $E(\delta(p_1,p_2)\mid\theta_1\theta_2=1)$ is largest when $E(1-\theta_1\theta_2\mid \delta(p_1,p_2)=1)\le \alpha.$

We concentrate on the Lagrange function with Lagrange multiplier $\nu_2$
$$
\begin{aligned}&L(\alpha,\nu_2,\delta(p_1,p_2))\\=&-\int \delta(p_1,p_2)f_1(p_1)f_2(p_2)dp_1dp_2+\nu_2\left\{\int (\xi_{00}+\xi_{10}f_1(p_1)+\xi_{01}f_2(p_2))\delta(p_1,p_2) dp_1dp_2\right.
\\&\qquad-\left.\alpha \int (\xi_{00}+\xi_{10}f_1(p_1)+\xi_{01}f_2(p_2)+\xi_{11}f_1(p_1)f_2(p_2))\delta(p_1,p_2) dp_1dp_2\right\}
\\=&\int \{\nu_2 (1-\alpha)(\xi_{00}+\xi_{10}f_1(p_1)+\xi_{01}f_2(p_2))+(-1-\nu_2\alpha)f_1(p_1)f_2(p_2)\} \delta(p_1,p_2)dp_1dp_2.
\end{aligned}
$$
Let $\nu_2$ be fixed. $L(\alpha,\nu_2, \delta(p_1,p_2))$ is minimized when 
$$
\begin{aligned}
\delta(p_1,p_2)&=1\left\{\nu_2 (1-\alpha)(\xi_{00}+\xi_{10}f_1(p_1)+\xi_{01}f_2(p_2))+(-1-\nu_2 \alpha)f_1(p_1)f_2(p_2)\le0\right\}
\\&=1\{\mbox{Lfdr}\le\frac{1+\nu_2 \alpha}{1+\nu_2}\}.
\end{aligned}
$$
So that in the oracle setting, the rejection criteria $\delta(p_1,p_2)=1\{\mbox{Lfdr}(p_1,p_2)\le\tau\}$ will get the highest power when the FDR is controlled.

\subsection{Proof of results in section \ref{sec:extend}}
\subsubsection{Proof of Proposition \ref{thm:mfdr_bound}}
\label{sec: thm:mfdr_bound_proof}
Recall that $f(p_1,\dots,p_n\mid\theta)=\prod\limits_{j=1}^n f_j^{\theta_j}(p_j)$ is the conditional density function of the p-values given $\theta$, and $f_{jj'}(p_j,p_{j'})$ is the joint probability density function of $(p_j,p_{j'}), j\ne j', 1\le j, j'\le n$.

When $h(\theta)=1$, there exists $j$ such that $\theta_j=0$. Denote $\theta_{-j}=1$ as $\forall j'\ne j, \theta_{j'}=1$. For all $j'\ne j$, we have
$$
\begin{aligned}
&E(\delta(\lambda)\mid \theta)
\\=&E(\delta(\lambda)\mid \theta_j=0,\theta_{-j}=1)
\\=&\int 1\{\mbox{Lfdr}_{12}\le \lambda_{12},\dots,\mbox{Lfdr}_{n-1,n}\le \lambda_{n-1,n}\} f(p_1,\dots,p_n\mid \theta) dp_1,\dots,dp_n
\\\le&\int1\{\mbox{Lfdr}_{jj'}\le \lambda_{jj'} \} f(p_1,\dots,p_n\mid \theta) dp_1,\dots,dp_n
\\=&\int1\{\mbox{Lfdr}_{jj'}\le \lambda_{jj'} \} f(p_j,p_{j'}\mid \theta) dp_jdp_{j'}
\\=& \int 1\{\mbox{Lfdr}_{jj'}\le \lambda_{jj'} \}f_{j'}(p_{j'}) dp_jdp_{j'}\\=&E[\frac{1\{\mbox{Lfdr}_{jj'}\le \lambda_{jj'} \}f_{j'}(p_{j'}) }{f_{jj'}(p_{j},p_{j'})}].
\end{aligned}
$$
Consequently,
$$
E(\delta(\lambda)\mid \theta_j=0,\theta_{-j}=1)\le \min\limits_{j',j'\ne j}E[\frac{1\{\mbox{Lfdr}_{jj'}\le \lambda_{jj'} \}f_{j'}(p_{j'}) }{f_{jj'}(p_j,p_{j'})}].
$$

When $h(\theta)\ge 2$, there exist $j\ne j'$ with $\theta_j=\theta_{j'}=0$.
In this case, we have
$$
\begin{aligned}
  &E(\delta(\lambda)\mid \theta)
  \\=&\int 1\{\mbox{Lfdr}_{12}\le \lambda_{12},\dots,\mbox{Lfdr}_{n-1,n}\le \lambda_{n-1,n}\} f_{}(p_1,\dots,p_n\mid\theta) dp_1,\dots,dp_n  \\
  \le & \int 1\{\mbox{Lfdr}_{jj'}\le \lambda_{jj'} \} f_{}(p_1,\dots,p_n\mid \theta) dp_1\dots dp_n\\=&
  \int 1\{\mbox{Lfdr}_{jj'}\le \lambda_{jj'} \} f(p_j,p_{j'}\mid \theta)dp_jdp_{j'}\\=&
  \int 1\{\mbox{Lfdr}_{jj'}\le \lambda_{jj'} \} dp_jdp_{j'}\\=&E[\frac{1\{\mbox{Lfdr}_{jj'}\le \lambda_{jj'} \} }{f_{jj'}(p_j,p_{j'})}]\\\le& \max\limits_{\substack{j,j'=1,\ldots,n\\j\ne j'}} E[\frac{1\{\mbox{Lfdr}_{jj'}\le \lambda_{jj'} \} }{f_{jj'}(p_j,p_{j'})}].  
\end{aligned}
$$

We next provide an upper bound for the numerator of (\ref{eq:mfdr}). Denote 
\begin{equation}
    \pi_j=P(\theta_j=0,\theta_{-j}=1),\quad\text{ and }\quad \pi_{n+1}=P(h(\theta)\ge 2), j=1,\dots, n.
\end{equation}
Combining the upper bounds of $E(\delta(\lambda)\mid \theta)$ under $\Theta_1$ and $\Theta_2$ yields 
\begin{align}
&E(\delta(\lambda)\cdot1\{\mathcal H_0\}) \nonumber
\\\le&\sum\limits_{j=1}^n\pi_j E(\delta(\lambda)\mid \theta_j=0,\theta_{-j}=1)+\pi_{n+1}\max\limits_{\theta\in \Theta_2} E(\delta(\lambda)\mid \theta) \nonumber
\\\le & \sum\limits_{j=1}^n\pi_j\min\limits_{j',j'\ne j}E[\frac{1\{\operatorname{Lfdr}_{jj'}\le \lambda_{jj'} \}f_{j'}(p_{j'}) }{f_{jj'}(p_j,p_{j'})}]+\pi_{n+1}\max\limits_{\substack{j,j'=1,\ldots,n\\j\ne j'}} E[\frac{1\{\operatorname{Lfdr}_{jj'}\le \lambda_{jj'} \} }{f_{jj'}(p_j,p_{j'})}]. \nonumber
\end{align}


\subsubsection{Parameters estimation}
\label{sup:proportion}
Recall that $\pi_0=P(\theta_1=\cdots=\theta_n=1)$, $\pi_j=P(\theta_j=0,\theta_{-j}=1)$, and $\pi_{jj'}=P(\theta_{j}=\theta_{j'}=0,\theta_{-jj'}=1)$ where $j,j'=1,\dots, n$ and $j\ne j'$. By expanding the marginal distribution of $p$-values, we find that 
\begin{align*}
&P(\max\limits_{j=1,\ldots, n}p_{j}\le \tau)
\\=&\pi_0P(\max\limits_{j=1,\ldots, n}p_{j}\le \tau \mid\theta_1=\cdots=\theta_n=1)\\&+\sum\limits_{j=1}^n \pi_jP(\max\limits_{j'=1,\ldots, n}p_{j'}\le \tau  \mid \theta_j=0,\theta_{-j}=1)
\\&+\sum\limits_{j,j',j\ne j'} \pi_{jj'}P(\max\limits_{j''=1,\ldots, n}p_{j''}\le \tau  \mid \theta_j=\theta_{j'}=0,\theta_{-jj'}=1)
\\&+P_1
\\=&\pi_0P(\max\limits_{j=1,\ldots, n}p_{j}\le \tau \mid\theta_1=\cdots=\theta_n=1)\\&+\tau \sum\limits_{j=1}^n \pi_jP(\max\limits_{j',j'\ne j}p_{j'}\le \tau  \mid \theta_{-j}=1)
\\&+\tau ^2\sum\limits_{j,j',j\ne j'} \pi_{jj'}P(\max\limits_{j''\ne j,j'}p_{j''}\le \tau  \mid \theta_{-jj'}=1)
\\&+P_1,
\end{align*}
where 
\begin{equation}
  \label{temp:p1}
  P_1=P(\max\limits_{j=1,\ldots, n}p_j\le \tau,h(\theta)> 2) . 
\end{equation}
For the last equation, we utilize the properties that when the hidden states are given, the $p$-values are conditionally independent and that $p_j\mid\theta_j=0$ is uniformly distributed.
We also have 

\begin{align*}
&P(\max\limits_{j',j'\ne j}p_{j'}\le \tau)
\\=&(\pi_0+\pi_j)P(\max\limits_{j',j'\ne j}p_{j'}\le \tau\mid \theta_{-j}=1)
\\&+\sum\limits_{j',j'\ne j}(\pi_{j'}+\pi_{jj'})P(\max\limits_{j''\ne j,j'}p_{j''}\le \tau\mid \theta_{j'}=0 ,\theta_{-jj'}=1)
\\&+P_2(j)
\\=&(\pi_0+\pi_j)P(\max\limits_{j',j'\ne j}p_{j'}\le \tau\mid \theta_{-j}=1)
\\&+\tau\sum\limits_{j',j'\ne j}(\pi_{j'}+\pi_{jj'}) P(\max\limits_{j''\ne j,j'}p_{j''}\le \tau  \mid \theta_{-jj'}=1)
\\&+P_2(j),
\end{align*}
where 
\begin{equation}
\label{temp:p2}
P_2(j)=P(\max\limits_{j',j'\ne j}p_j\le \tau,h(\theta_{-j})\ge 2).
\end{equation}
When we omit the terms with $\tau^2$ and $P_1$ in (\ref{temp:p1}) and omit the terms with $\tau$ and $P_2(j)$ in (\ref{temp:p2}), $\pi_0$ is approximatly equal to 
$$
\tilde \pi_0=\frac{P(\max\limits_{j=1,\ldots, n}p_{j}\le \tau)-\tau\sum\limits_{j=1}^n P(\max\limits_{j',j'\ne j}p_{j'}\le \tau)}{P(\max\limits_{j=1,\ldots, n}p_{j}\le \tau\mid\theta_1=\cdots=\theta_n=1 )-\tau \sum\limits_{j=1}^n P(\max\limits_{j',j'\ne j}p_{j'}\le \tau\mid\theta_{-j}=1)}.
$$
We also approximate $\pi_j,j=1,\cdots, n$ by
$$
\tilde\pi_j=\frac{P(\max\limits_{j',j'\ne j}p_{j'}\le \tau)}{P(\max\limits_{j',j'\ne j}p_{j'}\le \tau\mid \theta_{-j}=1)}-\tilde \pi_0.
$$
Finally, we estimate $\pi_{n+1}=P(h(\theta)\ge 2)$ by $\tilde \pi_{n+1}=1-\sum\limits_{j=1}^n\tilde\pi_j-\tilde\pi_0$. 

In practice, we approximate $\tilde{\pi}_j$ using the estimated densities $\widehat{f}_j$. 
Specifically, we estimate the following probabilities: 
$P(\max\limits_{j=1,\ldots,n} p_j \le \tau ),
P(\max\limits_{j',j'\ne j} p_{j'} \le \tau ),
P(\max\limits_{j=1,\ldots,n} p_j \le \tau \mid \theta_1=\cdots=\theta_n=1),\text{ and }
P(\max\limits_{j',j'\ne j} p_{j'} \le \tau \mid \theta_{-j}=1)$
by
\begin{align}
\hat P(\max\limits_{j=1,\ldots,n}p_j\le\tau)&=\frac{1}{m}\sum\limits_{i=1}^m 1\{\max\limits_{j=1,\ldots,n} p_{ij}\le \tau\},\label{P_hat}
\\
\hat P(\max\limits_{j',j'\ne j}p_{j'}\le\tau)&=\frac{1}{m}\sum\limits_{i=1}^m 1\{\max\limits_{j',j'\ne j} p_{ij'}\le \tau\},\label{P1_hat}\\ 
\hat P(\max\limits_{j=1,\ldots,n}p_j\le\tau\mid \theta_1=\cdots=\theta_n=1)&=\prod\limits_{j=1}^n\left\{\int_0^\tau \hat f_j(p_j)dp_j\right\} \label{P_hat_c} ,\quad\text{and}\\\hat P(\max\limits_{j',j'\ne j}p_{j'}\le\tau\mid \theta_{-j}=1)&=\prod\limits_{j',j'\ne j}\left\{\int_0^\tau \hat f_{j'}(p_{j'})dp_{j'}\right\}.\label{P1_hat_c}
\end{align}

The proportions $\tilde{\pi}_0$ (defined in (\ref{eq:tilde_pi0})), $\tilde{\pi}_j$ ($j=1,\ldots,n$) (defined in (\ref{eq:tilde_pij})), and $\tilde{\pi}_{n+1}=1-\tilde\pi_0-\sum\limits_{j=1}^n\tilde \pi_j$ are then estimated by
\begin{align}
    \hat \pi_0&=\frac{\hat P(\max\limits_{j=1,\ldots,n}p_{j}\le \tau)-\tau\sum\limits_{j=1}^n \hat P(\max\limits_{j',j'\ne j}p_{j'}\le \tau)}{\hat P(\max\limits_{j=1,\ldots,n}p_{j}\le \tau\mid\theta_1=\cdots=\theta_n=1 )-\tau \sum\limits_{j=1}^n \hat P(\max\limits_{j',j'\ne j}p_{j'}\le \tau\mid\theta_{-j}=1)}, \label{pi_0_hat}\\
    \hat \pi_j&=\frac{\hat P(\max\limits_{j',j'\ne j}p_{j'}\le \tau)}{\hat P(\max\limits_{j',j'\ne j}p_{j'}\le \tau\mid \theta_{-j}=1)}-\hat \pi_0 ,\quad\text{and}\label{pi_j_hat}
    \\\hat\pi_{n+1}&=1-\hat\pi_0-\sum\limits_{j=1}^n \hat \pi_j. \label{pi_n_hat}
\end{align}

\subsubsection{An example where condition (C6) holds}
\label{sup:C6}
\begin{lemma}
    Let $X\sim N(\mu,1)$. Consider the one-sided $p$-value $P=1-\Phi(X)$, where $\Phi(\cdot)$ is the cumulative distribution function of the standard normal distribution. Denote $F$ to be the cumulative distribution function of $P$. Then we have
    $$\lim\limits_{x\to 0}\frac{F(x)}{x}=+\infty.$$
\end{lemma}
\begin{proof}
    In this case, the cumulative distribution function of the $p$-value is
    $$
    F(x)=P(1-\Phi(X)\le x)=P(X\ge \Phi^{-1}(1-x))
    =1-\Phi(\Phi^{-1}(1-x)-\mu),
    $$
    where $\Phi^{-1}(\cdot)$ is the quantile function of the standard normal distribution.
    We can calculate the corresponding probability density function as
    $$
    \begin{aligned}
    f(x)&=\frac{d F(x)}{dx}\\&=\frac{d \{1-\Phi(\Phi^{-1}(1-x)-\mu)\}}{dx}\\&=\frac{d \{1-\Phi(\Phi^{-1}(1-x)-\mu)\}}{d\Phi^{-1}(1-x)}\frac{d\Phi^{-1}(1-x)}{dx}
    \\&=-\phi(\Phi^{-1}(1-x)-\mu)\cdot \frac{d\Phi^{-1}(1-x)}{dx}.
    \end{aligned}
    $$
    Let $t=\Phi^{-1}(1-x)$, then $x=1-\Phi(t)$, and
    $$
    \begin{aligned}
        \frac{d\Phi^{-1}(1-x)}{dx}&=\frac{dt}{d\{1-\Phi(t)\}}
        \\&=\frac{1}{-\phi(t)}
        \\&=-\frac{1}{\phi(\Phi^{-1}(1-x))}.
    \end{aligned}
    $$
    Then we have
    $$
    f(x)=\frac{\phi(\Phi^{-1}(1-x)-\mu)}{\phi(\Phi^{-1}(1-x))}=\exp(\mu\Phi^{-1}(1-x)-\mu^2/2).
    $$
    Therefore, we obtain that
    $$
    \lim\limits_{x\to 0}\frac{F(x)}{x}=\lim\limits_{x\to 0}f(x)=\infty.
    $$  
    and condition (C6) holds. 
\end{proof}

\subsubsection{Space \texorpdfstring{$\Lambda$}{Lambda} is not empty}
\label{sup:lambda}
Here, $\Lambda$ is defined as 
$$
\Lambda=\{\lambda:\min\limits_{j=1,\ldots, n}E_j(\lambda)\ge E_0(\lambda)\},
$$
where $E_0(\lambda)$ and $E_j(\lambda)$ are defined in equations (\ref{eq:e0}) and (\ref{eq:ej}). Now we show that $\Lambda$ is not empty.

We first show that when there are only two studies $j$ and $j'$, for any $\lambda_{jj'}$, we have
\begin{equation}
\label{eq:tmp_2comp}
  E(1\{\mbox{Lfdr}_{jj'}\le \lambda_{jj'}\}\mid\theta_j=\theta_{j'}=0\})\le E(1\{\mbox{Lfdr}_{jj'}\le \lambda_{jj'}\}\mid\theta_j=1,\theta_{j'}=0\})  .
\end{equation}

Let
$$\gamma_{jj'}(x)=\int_0^11\{\mbox{Lfdr}_{jj'}(x,y)\le \lambda_{jj'}\}dy.$$ We find that $\gamma_{jj'}(x)$ is monotone non-increasing with respect to $x$ for a fixed $\lambda_{jj'}$. Then we have
$$
\begin{aligned}
&E(1\{\mbox{Lfdr}_{jj'}\le \lambda_{jj'}\}\mid\theta_j=\theta_{j'}=0\})\\=&\int_0^1\int_0^1 1\{\mbox{Lfdr}_{jj'}(x,y)\le \lambda_{jj'}\}dxdy\\=&\int_0^1 \gamma_{jj'}(x)dx.
\end{aligned}
$$
and
$$
\begin{aligned}
&E(1\{\mbox{Lfdr}_{jj'}\le \lambda_{jj'}\}\mid\theta_j=1,\theta_{j'}=0\})\\=&\int_0^1\int_0^1 1\{\mbox{Lfdr}_{jj'}(x,y)\le \lambda_{jj'}\}f_j(x)dxdy\\=&\int_0^1 \gamma_{jj'}(x)f_j(x)dx.
\end{aligned}
$$

Because $f_j(x)$ and $\gamma_{jj'}(x)$ are monotone non-increasing, when $x\ne x'$, $\{\gamma_{jj'}(x)-\gamma_{jj'}(x')\}\{f_j(x)-f_{j}(x')\}\ge 0$ and
$$
\begin{aligned}
0&\le \int_0^1\int_0^1 \{\gamma_{jj'}(x)-\gamma_{jj'}(x')\}\{f_j(x)-f_{j}(x')\} dxdx'\\&=2\{\int_0^1 \gamma_{jj'}(x)f_j(x)dx-\int_0^1 \gamma_{jj'}(x)dx\int_0^1 f_j(x)dx\}.
\end{aligned}
$$
Therefore $\int_0^1 \gamma_{jj'}(x)f(x)dx\ge  \int_0^1 \gamma_{jj'}(x)dx$ and the equation (\ref{eq:tmp_2comp}) holds.

Define $g_{jj'}(\lambda_{jj'})=E(1\{\mbox{Lfdr}_{jj'}\le \lambda_{jj'}\}\mid\theta_j=\theta_{j'}=0\})$. As $g_{jj'}(\lambda_{jj'})$ is continuous with respect to $\lambda_{jj'}$ and $g_{jj'}(0)=0, g_{jj'}(1)=1$. For any $\tau\in (0,1)$ there exists $\lambda_{jj'}$ such that $g_{jj'}(\lambda_{jj'})=\tau$. So we choose $\{\lambda_{jj'}\}$ such that $g_{jj'}(\lambda_{jj'})=\tau$ for any $j\ne j'$. Therefore, we have $E_0(\lambda)=\tau$. For any $j\ne j'$, $E(1\{\mbox{Lfdr}_{jj'}\le \lambda_{jj'}\}\mid\theta_j=1,\theta_{j'}=0\})  \ge \tau$, so that $\min\limits_{j}E_j(\lambda)\ge \tau=E_0(\tau)$. Thus $\Lambda$ is not empty.

We prove the continuity of $$g_{jj'}(\lambda_{jj'})=E[1\{\mbox{Lfdr}_{jj'}\le \lambda_{jj'}\}\mid \theta_{j}=\theta_{j'}=0].$$ 

Without loss of generality, we consider the two study case and define $$g(\lambda)=E(1\{\mbox{Lfdr}_{}\le \lambda_{}\mid \theta_1=\theta_2=0\})=\int1\{\mbox{Lfdr}(p_1,p_2)\le \lambda \} dp_1dp_2.$$ By the assumption (C4),
$$
B_2(\lambda)=\int1\{\mbox{Lfdr}(p_1,p_2)\le \lambda \} \mbox{Lfdr}(p_1,p_2)dp_1dp_2
$$
is continuous in $\lambda.$ Fix an arbitrary $\lambda_1 \in (0,1)$ and let $\varepsilon>0.$ By continuity of $B_2,$ there exists $ \delta<\lambda_1/2$ such that whenever $|\lambda_1-\lambda_2|<\delta$, 
$$
|B_2(\lambda_1) - B_2(\lambda_2)| <\varepsilon.
$$ 
Observe that
$$
\begin{aligned}
&|B_2(\lambda_1) - B_2(\lambda_2)|\\=&\left|\int1\{\mbox{Lfdr}(p_1,p_2)\le \lambda_1 \} \mbox{Lfdr}(p_1,p_2)dp_1dp_2-\int1\{\mbox{Lfdr}(p_1,p_2)\le \lambda_2 \} \mbox{Lfdr}(p_1,p_2)dp_1dp_2\right|
\\=&\left|\int[1\{\mbox{Lfdr}(p_1,p_2)\le \lambda_1 \}-\int1\{\mbox{Lfdr}(p_1,p_2)\le \lambda_2 \}] \mbox{Lfdr}(p_1,p_2)dp_1dp_2\right|
\\=&\int\left|1\{\mbox{Lfdr}(p_1,p_2)\le \lambda_1 \}-1\{\mbox{Lfdr}(p_1,p_2)\le \lambda_2 \}\right| \mbox{Lfdr}(p_1,p_2)dp_1dp_2.
\end{aligned}
$$
Note $\left|1\{\mbox{Lfdr}(p_1,p_2)\le \lambda_1 \}-1\{\mbox{Lfdr}(p_1,p_2)\le \lambda_2 \}\right|\ne 0$ only when $\min\{\lambda_1,\lambda_2\}\le  \mbox{Lfdr}(p_1,p_2)< \max\{\lambda_1,\lambda_2\}$. In this case, $\mbox{Lfdr}(p_1,p_2)>\min\{\lambda_1,\lambda_2\}>\frac{\lambda_1}{2}$.
Hence,
$$
\begin{aligned}
    &|B_2(\lambda_1) - B_2(\lambda_2)|
    \\=&\int\left|1\{\mbox{Lfdr}(p_1,p_2)\le \lambda_1 \}-1\{\mbox{Lfdr}(p_1,p_2)\le \lambda_2 \}\right| \mbox{Lfdr}(p_1,p_2)dp_1dp_2.
    \\>& \frac{\lambda_1}{2}\int\left|1\{\mbox{Lfdr}(p_1,p_2)\le \lambda_1 \}-1\{\mbox{Lfdr}(p_1,p_2)\le \lambda_2 \}\right| dp_1dp_2
    \\=&\frac{\lambda_1}{2}|g(\lambda_1)-g(\lambda_2)|.
\end{aligned}
$$
Therefore,
$$
|g(\lambda_1) - g(\lambda_2)| \le \frac{2}{\lambda_1}|B_2(\lambda_1) - B_2(\lambda_2)| < \frac{2\epsilon}{\lambda_1}
$$
This shows that $g(\lambda)$ is continuous at $\lambda_1$. Since $\lambda_1$ is abitrary, $g(\lambda)$ is continuous on $(0,1)$.

\subsubsection{Proof of Proposition \ref{thm:R_compare}}
\label{thm:R_compare_proof}
\begin{proof}
    Recall that $P_1$ and $P_2(j)$ are defined in (\ref{temp:p1}) and (\ref{temp:p2}), respectively. Also recall 
    $\pi_0=P(\theta_1=\cdots=\theta_n=1)$, $\pi_j=P(\theta_j=0,\theta_{-j}=1)$, and $\pi_{jj'}=P(\theta_j=\theta_{j'}=1,\theta_{-jj'}=0)$ for $j\ne j',j,j'=1,\cdots, n$, together with their estimators $\tilde\pi_{0}$ and $\tilde \pi_j$  defined in (\ref{eq:tilde_pi0}) and (\ref{eq:tilde_pij}). Set $\pi_{n+1}=1-\sum\limits_{j=1}^n\pi_j-\pi_0$, and 
    $\tilde \pi_{n+1}=1-\sum\limits_{j=1}^n\tilde \pi_j-\tilde \pi_0$. Denote $\delta(\lambda)=1\{\mbox{Lfdr}_{jj'}\le \lambda_{jj'},j,j'=1,\ldots, n,j\ne j'\}.$ Recall $E_0(\lambda)$ and $E_j(\lambda)$ in (\ref{eq:e0}) and (\ref{eq:ej}).
    
    We introduce $P_3$ and $P_4(j)$ as follows. 
    $$
    \begin{aligned}
    P_3&=\tilde \pi_0-\pi_0\\&=\frac{P_1-\tau\sum\limits_{j=1}^n P_2(j)-\tau^2\sum\limits_{j=1}^n\sum\limits_{j',j'\ne j}\pi_{j'}P(\max\limits_{j'',j''\ne j,j'}p_{j''}\le \tau\mid\theta_{-jj'}=1)}{P(\max\limits_{j=1,\ldots, n}p_{j}\le \tau\mid\theta_1=\cdots=\theta_n=1 )-\tau \sum\limits_{j=1}^n P(\max\limits_{j',j'\ne j}p_{j'}\le \tau\mid\theta_{-j}=1)},
    \end{aligned}
    $$
    and
    \begin{align}
    P_4(j)&=\tilde \pi_j-\pi_j\nonumber\\&=\frac{\tau \sum\limits_{j',j'\ne j}(\pi_{j'}+\pi_{jj'})P(\max\limits_{j'',j''\ne j,j'}p_{j''}\le \tau\mid \theta_{-jj'}=1)
    +P_2(j)}{P(\max\limits_{j',j'\ne j}p_{j'}\le \tau\mid \theta_{-j}=1)}-P_3 . \label{eq:P_4}
    \end{align}
    
    By (\ref{eq:original_t}) and (\ref{eq:R2}),
    $$
    \begin{aligned}
    \mathcal R_1(\lambda)&= \frac{\pi_{n+1} E_0(\lambda)+\sum\limits_{j=1}^n\pi_jE_j(\lambda)}{E\{\delta(\lambda)\}},\quad{\text{and}}\\
    \mathcal R_2(\lambda,\tau)&= \frac{\tilde \pi_{n+1} E_0(\lambda)+\sum\limits_{j=1}^n\tilde \pi_jE_j(\lambda)}{E\{\delta(\lambda)\}}.
    \end{aligned}
    $$
    Let $E_{\min}(\lambda)=\min\{E_1(\lambda),\cdots,E_n(\lambda)\}$ and consider
    $$
    \Lambda=\{\lambda: E_{\min}(\lambda)\ge E_0(\lambda )\}.
    $$
    For $\lambda\in \Lambda,$ we have $E_j(\lambda)-E_0(\lambda)\ge 0$. A direct calculation yields 
    $$
    \begin{aligned}
    &\{\mathcal R_2(\lambda,\tau)-\mathcal R_1(\lambda)\}\cdot E\{\delta(\lambda)\} 
    \\=& (\tilde \pi_{n+1}-\pi_{n+1})E_0(\lambda)+\sum\limits_{j=1}^n (\tilde \pi_j-\pi_j) E_j(\lambda)
    \\=& \{(1-\sum\limits_{j=1}^n\tilde\pi_j-\tilde\pi_0)-(1-\sum\limits _{j=1}^n\pi_j-\pi_0)\}E_0(\lambda)+\sum\limits_{j=1}^n (\tilde \pi_j-\pi_j) E_j(\lambda)
    \\=&[\sum\limits_{j=1}^n P_4(j) \{E_j(\lambda)-E_0(\lambda)\}]-P_3 E_0(\lambda).
    \end{aligned}
    $$
    Thus, to conclude $\mathcal R_2(\lambda,\tau)>\mathcal R_1(\lambda)$, it suffices to prove
    \begin{equation}
    \label{eq:R_comp}[\sum\limits_{j=1}^n P_4(j) \{E_j(\lambda)-E_0(\lambda)\}]-P_3 E_0(\lambda)> 0.
    \end{equation}

    Now we provide a bound of $P_3$. To obtain an upper bound for the numerator, we analyze its component $P_1$ first.
    Recall that $h(\theta)=\#\{j:\theta_j=0\}$, and $\pi_{jj'}=P(\theta_j=\theta_{j'}=0,\theta_{-jj'}=1)$ for $j,j'=1,\ldots,n,j\ne j'$.
   Consider the event $A_j=\{p_j\le \tau,\max\limits_{j',j'\ne j}p_{j'}\le \tau,\theta_j=0,h(\theta_{-j})\ge 2\}$. We find that $\{\max\limits_{j=1,\ldots, n}p_j\le \tau,h(\theta)>2  \}=\bigcup\limits_{j=1}^n A_j$. Therefore, we obtain that 
    $$
    \begin{aligned}
    P_1\le&\sum\limits_{j=1}^nP(A_j)
    \\=&\sum\limits_{j=1}^nP(\theta_j=0,h(\theta_{-j})\ge 2)P(p_j\le \tau,\max\limits_{j',j'\ne j}p_{j'}\le \tau\mid \theta_j=0,h(\theta_{-j})\ge 2)
    \\\le &\sum\limits_{j=1}^n P\{h(\theta_{-j})\ge 2\}\cdot\tau\cdot P(\max\limits_{j',j'\ne j}p_{j'}\le \tau\mid h(\theta_{-j})\ge 2)
    \\=&\tau\sum\limits_{j=1}^n P(\max\limits_{j',j'\ne j}p_{j'}\le \tau,h(\theta_{-j})\ge 2)
    \\=&\tau\sum\limits_{j=1}^n P_2(j).
    \end{aligned}
    $$

    The numerator of $P_3$ is upper-bounded by
    \begin{align}
    &P_1-\tau\sum\limits_{j=1}^n P_2(j)-\tau^2\sum\limits_{j=1}^n\sum\limits_{j',j'\ne j}\pi_{j'}P(\max\limits_{j''\ne j,j'}p_{j''}\le \tau\mid \theta_{-jj'}=1) \nonumber
    \\\le& -\tau^2\sum\limits_{j=1}^n\sum\limits_{j',j'\ne j}\pi_{j'}P(\max\limits_{j''\ne j,j'}p_{j''}\le \tau\mid \theta_{-jj'}=1)\\<&0.
    \end{align}

           The lower bound of the denominator of $P_3$ is 
    \begin{align}
         &P(\max\limits_{j=1,\ldots, n}p_{j}\le \tau\mid\theta_1=\cdots=\theta_n=1 )-\tau \sum\limits_{j=1}^nP(\max\limits_{j',j'\ne j}p_{j'}\le \tau\mid\theta_{-j}=1)\nonumber\\\ge& (F_{(1)}(\tau)-n\tau)F_{(2)}(\tau)\cdots F_{(n)}(\tau).\label{eq:P3_de}
    \end{align}
    By assumption (C6), $$\lim\limits_{\tau\to 0 } \frac{F_{\min}(\tau)}{\tau}= \infty,$$
    where $F_{min}(\tau)=\min\{F_1(\tau),\ldots,F_{n}(\tau)\}=F_{(1)}(\tau).$ Hence, there exists $\tau^*$ such that for all $\tau<\tau^*$, $\frac{F_{\min}(\tau)}{\tau}>2n$ for a fixed $n$. 
    Then we have $F_{(1)}(\tau)-n\tau>\frac{1}{2}F_{(1)}(\tau)$ and
    $$
    \begin{aligned}
&P(\max\limits_{j=1,\ldots, n}p_{j}\le \tau\mid\theta_1=\cdots=\theta_n=1 )-\tau \sum\limits_{j=1}^nP(\max\limits_{j',j'\ne j}p_{j'}\ge \tau\mid\theta_{-j}=1)
\\\ge& (F_{(1)}(\tau)-n\tau)F_{(2)}(\tau)\cdots F_{(n)}(\tau)
\\\ge& \frac{1}{2}\prod\limits_{i=1}^nF_{i}(\tau)
\\>&0.
\end{aligned}
$$
    
    Therefore $P_3< 0$. By definition (\ref{eq:P_4}), $P_4(j)>0$ for all $j=1,\ldots,n$, and therefore (\ref{eq:R_comp}) holds immediately.

\end{proof}
\subsubsection{Proof of Theorem \ref{thm:multi-fdr}}
\label{thm:multi-fdr_proof}
Before we prove Theorem \ref{thm:multi-fdr}, we need the following two lemmas.

\begin{lemma}
\label{lem: product_bound}
    If $a_1,\dots, a_n, b_1,\dots,b_n\in [0,1]$, we have $|\prod\limits_{i=1}^n a_i-\prod\limits_{i=1}^n b_i|\le\sum\limits_{i=1}^n|a_i-b_i|$.
\end{lemma}
\begin{proof}
    When $n=1$, it holds directly. We have
    $$
    \begin{aligned}
        &|\prod\limits_{i=1}^{n+1} a_i-\prod\limits_{i=1}^{n+1} b_i|\\=&|a_{n+1}(\prod\limits_{i=1}^{n} a_i-\prod\limits_{i=1}^{n} b_i)+\prod\limits_{i=1}^{n} b_i(a_{n+1}-b_{n+1})|
        \\\le& |\prod\limits_{i=1}^{n} a_i-\prod\limits_{i=1}^{n} b_i|+|a_{n+1}-b_{n+1}|.
    \end{aligned}
    $$
    By mathematical induction, we prove Lemma \ref{lem: product_bound}.
\end{proof}

\begin{lemma}
    \label{lem: min_max_bound}
    If $a_1,\cdots, a_n, b_1,\cdots,b_n\in \mathbb{R}$, we have
    $$
    \begin{aligned}
        |\min\limits_{i}a_i -\min\limits_{i}b_i|&\le \sum\limits_{i=1}^n |a_i-b_i|,\quad{\text{and}}\\
        |\max\limits_{i}a_i -\max\limits_{i}b_i|&\le \sum\limits_{i=1}^n |a_i-b_i|.
    \end{aligned}
    $$
\end{lemma}
\begin{proof}
    When $n = 1$, it holds directly.
    
    For any two values $x,y\in \mathbb{R}$, 
    $$\begin{aligned}
    \min\{x,y\}&=\frac{x+y}{2}-\frac{|x-y|}{2},\quad\text{and}\\\max\{x,y\}&=\frac{x+y}{2}+\frac{|x-y|}{2}.\end{aligned}$$

    Let $A_n=\min\limits_{i=1,\cdots n}a_i,$ and $B_n=\min\limits_{i=1,\cdots, n} b_i$. We have 
    $$
    \begin{aligned}
        &|\min\limits_{i=1,\cdots,{n+1}}a_i -\min\limits_{i=1,\cdots,{n+1}}b_i|
        \\=&\left|\frac{A_n+a_{n+1}}{2}-\frac{|A_n-a_{n+1}|}{2}-\frac{B_n+b_{n+1}}{2}+\frac{|B_n-b_{n+1}|}{2}\right|
        \\\le&\left|\frac{A_n+a_{n+1}}{2}-\frac{B_n+b_{n+1}}{2}\right|+\left|\frac{|A_n-a_{n+1}|}{2}-\frac{|B_n-b_{n+1}|}{2}\right|
        \\\le&\left|\frac{A_n+a_{n+1}}{2}-\frac{B_n+b_{n+1}}{2}\right|+\left|\frac{A_n-a_{n+1}}{2}-\frac{B_n-b_{n+1}}{2}\right|
    \\\le&\frac{|A_n-B_n|}{2}+\frac{|a_{n+1}-b_{n+1}|}{2}+\frac{|A_n-B_n|}{2}+\frac{|a_{n+1}-b_{n+1}|}{2}
    \\=&|A_n-B_n|+|a_{n+1}-b_{n+1}|.
    \end{aligned}
    $$
    By mathematical induction, $$|\min\limits_{i=1,\ldots, n}a_i -\min\limits_{i=1,\ldots,n}b_i|\le \sum\limits_{i=1}^n |a_i-b_i|.$$ In the same way, we can show that
    \[
    |\max\limits_{i=1,\ldots, n}a_i -\max\limits_{i=1,\ldots,n}b_i|\le \sum\limits_{i=1}^n |a_i-b_i|.
    \] 
    Thus, we prove the result.
\end{proof}

We need the following consistency results.
\begin{lemma}
\label{lem:pi_consistency}
For a fixed $n$ and for $j=1,\cdots, n$, if $\int_0^1 |\hat f_j(p)-f_j(p)|dp\to 0$, we have
$$
\begin{aligned}
\hat P(\max\limits_{j=1,\cdots,n}p_j\le\tau)&\overset{p}{\to} P(\max\limits_{j=1,\cdots,n}p_j\le\tau),\\
\hat P(\max\limits_{j',j'\ne j}p_{j'}\le\tau)&\overset{p}{\to} P(\max\limits_{j',j'\ne j}p_{j'}\le\tau),\\
\hat P(\max\limits_{j=1,\cdots,n}p_j\le\tau\mid \theta_1=\cdots=\theta_n=1)&\overset{p}{\to}P(\max\limits_{j=1,\cdots,n}p_j\le\tau\mid \theta_1=\cdots=\theta_n=1),\\
\hat P(\max\limits_{j',j'\ne j}p_{j'})&\overset{p}{\to}P(\max\limits_{j',j'\ne j}p_{j'}),
\end{aligned}
$$
$$\hat\pi_0\overset{p}{\to}\tilde \pi_0, \quad \hat \pi_j\overset{p}{\to} \tilde \pi_j,\quad\text{and}\quad \hat\pi_{n+1}\overset{p}{\to} \tilde \pi_{n+1},$$
where $\hat P(\max\limits_{j=1,\cdots,n}p_j\le\tau)$, $\hat P(\max\limits_{j',j'\ne j}p_{j'}\le\tau)$, $\hat P(\max\limits_{j=1,\cdots,n}p_j\le\tau\mid \theta_1=\cdots=\theta_n=1)$, $\hat P(\max\limits_{j',j'\ne j}p_{j'})$, $\hat\pi_0$, $\hat \pi_j$ and $\hat\pi_{n+1}$ are defined in equations (\ref{P_hat})-(\ref{pi_n_hat}), $\tilde \pi_0,\tilde\pi_j,$ and $\tilde\pi_{n+1}$ are defined in equations (\ref{eq:tilde_pi0}) and (\ref{eq:tilde_pij}).
\end{lemma}

\begin{proof}
By the law of large numbers, we have
$$
\hat P(\max\limits_{j=1,\cdots,n}p_j\le\tau)\overset{p}{\to} P(\max\limits_{j=1,\cdots,n}p_j\le\tau),
$$
and
$$
\hat P(\max\limits_{j',j'\ne j}p_{j'}\le\tau)\overset{p}{\to} P(\max\limits_{j',j'\ne j}p_{j'}\le\tau).
$$

Recall that 
$$
\begin{aligned}
P(\max\limits_{j=1,\cdots,n}p_j\le\tau\mid \theta_1=\cdots=\theta_n=1)&=\prod\limits_{j=1}^n\left\{\int_0^\tau f_j(p_j)dp_j\right\},
\\P(\max\limits_{j',j'\ne j}p_{j'}\le\tau\mid \theta_{-j}=1)&=\prod\limits_{j',j'\ne j}\left\{\int_0^\tau  f_{j'}(p_{j'})dp_{j'}\right\},
\\
\hat P(\max\limits_{j=1,\cdots,n}p_j\le\tau\mid \theta_1=\cdots=\theta_n=1)&=\prod\limits_{j=1}^n\left\{\int_0^\tau \hat f_j(p_j)dp_j\right\},\quad\text{and}
\\\hat P(\max\limits_{j',j'\ne j}p_{j'}\le\tau\mid \theta_{-j}=1)&=\prod\limits_{j',j'\ne j}\left\{\int_0^\tau \hat f_{j'}(p_{j'})dp_{j'}\right\}.
\end{aligned}
$$
By Lemma \ref{lem: product_bound}, we have
$$
\begin{aligned}
&|P(\max\limits_{j=1,\cdots,n}p_j\le\tau\mid \theta_1=\dots=\theta_n=1)-\hat P(\max\limits_{j=1,\cdots,n}p_j\le\tau\mid \theta_1=\dots=\theta_n=1)|
\\\le&\sum\limits_{j=1}^n\left\{\int_0^\tau |f_j(p_j)-\hat f_j(p_j)|dp_j\right\}\overset p \to 0,
\end{aligned}
$$
and
$$
\begin{aligned}
&|P(\max\limits_{j',j'\ne j}p_{j'}\le\tau\mid \theta_{-j}=1)-\hat P(\max\limits_{j',j'\ne j}p_{j'}\le\tau\mid \theta_{-j}=1)|
\\\le&\sum\limits_{j',j'\ne j}\left\{\int_0^\tau |f_{j'}(p_{j'})-\hat f_{j'}(p_{j'})|dp_{j'}\right\}\overset{p}{\to}0.
\end{aligned}
$$

By continuous mapping theorem, we have 
$$\hat\pi_0\overset p\to\tilde \pi_0, \quad \hat \pi_j\overset p\to \tilde \pi_j,j=1\ldots,n,\quad\text{and}\quad \hat\pi_{n+1}\overset p\to \tilde \pi_{n+1}.$$
Thus, we prove the result.
\end{proof}

We now prove the following uniform consistency results. 

\begin{lemma}
    Under the assumptions of Theorem~\ref{thm:multi-fdr}, for any $1\le j,j'\le n$ and $j\ne j$,  we have
    $$
    \begin{aligned}
        \sup\limits_{\lambda_{jj'}\in [\lambda_0,1]}\left|\frac{1}{m}\sum\limits_{i=1}^m \frac{1\{\widehat{\mbox{Lfdr}}_{i,jj'}\le \lambda_{jj'}\}\hat f_{i,j}}{\hat f_{i,jj'}}- E\{ \frac{1\{\mbox{Lfdr}_{jj'}\le \lambda_{jj'}\} f_{j}}{f_{jj'}}\}\right|&\overset p\to 0, \quad\text{and}\\
        \sup\limits_{\lambda_{jj'}\in [\lambda_0,1]}|\frac{1}{m}\sum\limits_{i=1}^m \frac{1\{\widehat{\mbox{Lfdr}}_{i,jj'}\le \lambda_{jj'}\}}{\hat f_{i,jj'}}- E\{ \frac{1\{\mbox{Lfdr}_{jj'}\le \lambda_{jj'}\} }{f_{jj'}}\}|&\overset p\to 0.
    \end{aligned}
    $$    
where $\lambda_0$ is defined in assumption (C5).
\end{lemma}
\begin{proof}
    We already show in the proof of Theorem \ref{fdr_theorem} that $\forall j\ne j',1\le j,j'\le n,$
\begin{equation}
\label{eq:lfdr_pair}
  \sup\limits_{\lambda_{jj'}\in [\lambda_0,1]}\frac{1}{m}\sum\limits_{i=1}^m|1\{\widehat {\mbox{Lfdr}}_{i,jj'}\le \lambda_{jj'}\}-1\{\mbox{Lfdr}_{i,jj'}\le \lambda_{jj'}\}|\overset p \to 0. 
\end{equation}

By the proof of the equation (\ref{eq:lfdr_part}) , we have
$$
\frac{1}{m}\sum\limits_{i=1}^m|\frac{\hat f_{i,j}}{\hat f_{i,jj'}}-\frac{f_{i,j}}{f_{i,jj'}}|\overset p \to 0,
$$
and
$$
\frac{1}{m}\sum\limits_{i=1}^m|\frac{1}{\hat f_{i,jj'}}-\frac{1}{f_{i,jj'}}|\overset p \to 0.
$$

For each study pair $(j,j')$ with $j \neq j'$, let 
$\xi_{00,jj'}, \xi_{10,jj'}, \xi_{01,jj'}, \xi_{11,jj'}$ 
denote the four mixture weights (corresponding to the two-study configuration
$00,10,01,11$) in that pair, which are bounded below by a small positive constant $l$.

Therefore, we obtain that
$$
\begin{aligned}
&\left|\frac{1}{m}\sum\limits_{i=1}^m \frac{1\{\widehat{\mbox{Lfdr}}_{i,jj'}\le \lambda_{jj'}\}\hat f_{i,j}}{\hat f_{i,jj'}}- \frac{1}{m}\sum\limits_{i=1}^m\frac{1\{\mbox{Lfdr}_{i,jj'}\le \lambda_{jj'}\} f_{i,j}}{f_{i,jj'}}\right|
\\\le &\frac{1}{m}\sum\limits_{i=1}^m\left| \frac{1\{\widehat{\mbox{Lfdr}}_{i,jj'}\le \lambda_{jj'}\}\hat f_{i,j}}{\hat f_{i,jj'}}-\frac{1\{\mbox{Lfdr}_{i,jj'}\le \lambda_{jj'}\} f_{i,j}}{f_{i,jj'}}\right|
\\\le& \frac{1}{m}\sum\limits_{i=1}^m\left| \frac{1\{\widehat{\mbox{Lfdr}}_{i,jj'}\le \lambda_{jj'}\}\hat f_{i,j}}{\hat f_{i,jj'}}-\frac{1\{\widehat{\mbox{Lfdr}}_{i,jj'}\le \lambda_{jj'}\}f_{i,j}}{ f_{i,jj'}}\right.\\&\qquad+\left.\frac{1\{\widehat{\mbox{Lfdr}}_{i,jj'}\le \lambda_{jj'}\}f_{i,j}}{f_{i,jj'}}-\frac{1\{\mbox{Lfdr}_{i,jj'}\le \lambda_{jj'}\} f_{i,j}}{f_{i,jj'}}\right|
\\\le& \frac{1}{m}\sum\limits_{i=1}^m\left| (\frac{\hat f_{i,j}}{\hat f_{i,jj'}}-\frac{f_{i,j}}{ f_{i,jj'}})1\{\widehat{\mbox{Lfdr}}_{i,jj'}\le \lambda_{jj'}\}\right|\\&\qquad+\frac{1}{m}\sum\limits_{i=1}^m\left|(1\{\widehat{\mbox{Lfdr}}_{i,jj'}\le \lambda_{jj'}\}-1\{\mbox{Lfdr}_{i,jj'}\le\lambda_{jj'}\})\frac{f_{i,j}}{f_{i,jj'}}\right|
\\\le& \frac{1}{m}\sum\limits_{i=1}^m\left|\frac{\hat f_{i,j}}{\hat f_{i,jj'}}-\frac{f_{i,j}}{ f_{i,jj'}}\right|+\frac{1}{m}\frac{1}{l}\sum\limits_{i=1}^m\left|1\{\widehat{\mbox{Lfdr}}_{i,jj'}\le \lambda_{jj'}\}-1\{\mbox{Lfdr}_{i,jj'}\le \lambda_{jj'}\}\right|.
\end{aligned}
$$
The last inequality holds because of the assumption $\xi_{10,jj'}>l$, and
$$
\begin{aligned}
\frac{f_{i,j}}{f_{i,jj'}}&=\frac{f_{i,j}}{\xi_{00,jj'}+\xi_{10,jj'}f_{i,j}+\xi_{01,jj'}f_{i,j'}+\xi_{11,jj'}f_{i,j}f_{i,j'}}
\\&\le \frac{f_{i,j}}{\xi_{10,jj'}f_{i,j}}
\\&=\frac{1}{\xi_{10,jj'}}\\&\le \frac{1}{l}.
\end{aligned}
$$

Thus, we have
$$
\sup\limits_{\lambda_{jj'}\in [\lambda_0,1]}\left|\frac{1}{m}\sum\limits_{i=1}^m \frac{1\{\widehat{\mbox{Lfdr}}_{i,jj'}\le \lambda_{jj'}\}\hat f_{i,j}}{\hat f_{i,jj'}}- \frac{1}{m}\sum\limits_{i=1}^m\frac{1\{\mbox{Lfdr}_{i,jj'}\le \lambda_{jj'}\} f_{i,j}}{f_{i,jj'}}\right|\overset p\to 0.
$$

By the Uniformly Law of Large Numbers (Lemma 3.1 in \citep{van2000empirical}), we have 
$$
\sup\limits_{\lambda_{jj'}\in [\lambda_0,1]}|\frac{1}{m}\sum\limits_{i=1}^m\frac{1\{\mbox{Lfdr}_{i,jj'}\le \lambda_{jj'}\} f_{i,j}}{f_{i,jj'}}- E\{ \frac{1\{\mbox{Lfdr}_{jj'}\le \lambda_{jj'}\} f_{j}}{f_{jj'}}\}|\overset p\to 0.
$$
Therefore, we have 
\begin{equation}
    \label{eq:fj}
    \sup\limits_{\lambda_{jj'}\in [\lambda_0,1]}|\frac{1}{m}\sum\limits_{i=1}^m \frac{1\{\widehat{\mbox{Lfdr}}_{i,jj'}\le \lambda_{jj'}\}\ \hat f_{i,j}}{\hat f_{i,jj'}}- E\{ \frac{1\{\mbox{Lfdr}_{jj'}\le \lambda_{jj'}\} f_{j}}{f_{jj'}}\}|\overset p\to 0.
\end{equation}

In the same way, we can show that $\forall j\ne j',1\le j,j'\le n$,
\begin{equation}
\label{eq:f0}
    \sup\limits_{\lambda_{jj'}\in [\lambda_0,1]}|\frac{1}{m}\sum\limits_{i=1}^m \frac{1\{\widehat{\mbox{Lfdr}}_{i,jj'}\le \lambda_{jj'}\}}{\hat f_{i,jj'}}- E\{ \frac{1\{\mbox{Lfdr}_{jj'}\le \lambda_{jj'}\} }{f_{jj'}}\}|\overset p\to 0.
\end{equation}
Thus, we prove the lemma.
\end{proof}

Let $D=n(n-1)/2$ be the number of pairwise studies, and 
\begin{equation}
    \label{eq:lambda-all}
    \lambda=(\lambda_{12},\ldots,\lambda_{n-1,n}).
\end{equation}

Now we show that 
$$
\sup\limits_{\lambda\in[\lambda_0,1]^D}\left|\frac{1}{m}\sum\limits_{i=1}^m\left[\prod\limits_{j,j',j\ne j'}1\{\widehat {\mbox{Lfdr}}_{i,jj'}\le \lambda_{jj'}\}\right]- E\left[\prod\limits_{j,j',j\ne j'}1\{{\mbox{Lfdr}}_{i,jj'}\le \lambda_{jj'}\}\right]\right|\overset p\to 0.
$$ 
This can be obtained 
by the following inequality
$$
\begin{aligned}
&\sup\limits_{\lambda\in[\lambda_0,1]^D}\frac{1}{m}\sum\limits_{i=1}^m\left|\prod\limits_{j,j',j\ne j'}1\{\widehat {\mbox{Lfdr}}_{i,jj'}\le \lambda_{jj'}\}-\prod\limits_{j,j',j\ne j'}1\{\mbox{Lfdr}_{i,jj'}\le \lambda_{jj'}\}\right|\\\le& \sum\limits_{j,j',j\ne j'}\sup\limits_{\lambda\in[\lambda_0,1]^D}\{\frac{1}{m}\sum\limits_{i=1}^m|1\{\widehat{\mbox{Lfdr}}_{i,jj'}\le \lambda_{jj'}\}-1\{\mbox{Lfdr}_{i,jj'}\le \lambda_{jj'}\}|\}\overset p \to 0,
\end{aligned}
$$
which is obtained by the Lemma \ref{lem: product_bound} and the equation (\ref{eq:lfdr_pair}).

Taking $\lambda$ in (\ref{eq:lambda-all}), we want to show that for a specific $\tau$
$$
\sup\limits_{\lambda\in [\lambda_0,1]^{D}}|\mathcal R_2(\lambda,\tau)-\mathcal R_3(\lambda,\tau)|\overset p\to 0,
$$
where $\mathcal R_2(\lambda,\tau)$ is defined in (\ref{eq:R2}), and $\mathcal R_3(\lambda,\tau)$ is defined in (\ref{eq:R_3})

By equation (\ref{eq:fj}) and (\ref{eq:f0}), and Lemma \ref{lem: min_max_bound}, we have
$$
\sup\limits_{\lambda\in [\lambda_0,1]^{D}}\left|\max\limits_{j,j',j\ne j'}\frac{1}{m}\sum\limits_{i=1}^m\frac{1\{\widehat{\mbox{Lfdr}}_{i,jj'}(p_{i,j},p_{i,j'})\le \lambda_{jj'}\}}{\hat{f}_{jj'}(p_{i,j},p_{i,j'})}-\max\limits_{j,j',j\ne j'}E\{ \frac{1\{\mbox{Lfdr}_{jj'}\le \lambda_{jj'}\} }{f_{jj'}}\}\right|\overset p\to 0,
$$
and
$$
\sup\limits_{\lambda\in [\lambda_0,1]^{D}}\left|\min\limits_{j,j',j\ne j'}\frac{1}{m}\sum\limits_{i=1}^m\frac{1\{\widehat{\mbox{Lfdr}}_{i,jj'}(p_{i,j},p_{i,j'})\le \lambda_{jj'}\}\hat f_{j'}(p_{i,j'})}{\hat{f}_{jj'}(p_{i,j},p_{i,j'})}-\min\limits_{j,j',j\ne j'}E\{ \frac{1\{\mbox{Lfdr}_{jj'}\le \lambda_{jj'}\} f_{j'}(p_{j'})}{f_{jj'}}\}\right|\to 0.
$$

Note that 
$$
1\{\widehat{\mbox{Lfdr}}_{i,jj'}\le \lambda_{jj'},\forall j\ne j',1\le j,j'\le n\}=\prod\limits_{j,j',j\ne j'}1\{\widehat{\mbox{Lfdr}}_{i,jj'}\le \lambda_{jj'}\},
$$
and 
$$
\delta(\lambda)=1\{\mbox{Lfdr}_{i,jj'}\le \lambda_{jj'},\forall j\ne j',1\le j,j'\le n\}=\prod\limits_{j,j',j\ne j'}1\{\mbox{Lfdr}_{i,jj'}\le \lambda_{jj'}\}.
$$
We have
$$
\begin{aligned}
&|\frac{1}{m}\sum\limits_{i=1}^m 1\{\widehat{\mbox{Lfdr}}_{i,jj'}\le \lambda_{jj'},\forall j\ne j',1\le j,j'\le n\}-\frac{1}{m}\sum\limits_{i=1}^m 1\{\mbox{Lfdr}_{i,jj'}\le \lambda_{jj'},\forall j\ne j',1\le j,j'\le n\}|
\\= &\left|\frac{1}{m}\sum\limits_{i=1}^m\prod\limits_{j,j',j\ne j'}1\{\widehat{\mbox{Lfdr}}_{i,jj'}\le\lambda_{jj'}\}-\frac{1}{m}\sum\limits_{i=1}^m\prod\limits_{j,j',j\ne j'}1\{{\mbox{Lfdr}}_{i,jj'}\le\lambda_{jj'}\}\right|
\\\le &\frac{1}{m}\sum\limits_{i=1}^m\left|\prod\limits_{j,j',j\ne j'}1\{\widehat{\mbox{Lfdr}}_{i,jj'}\le\lambda_{jj'}\}-\prod\limits_{j,j',j\ne j'}1\{{\mbox{Lfdr}}_{i,jj'}\le\lambda_{jj'}\}\right|
\\\le &\sum\limits_{j,j',j\ne j'}\{\frac{1}{m}\sum\limits_{i=1}^m\left|1\{{\mbox{Lfdr}}_{i,jj'}\le\lambda_{jj'}\}-1\{\widehat{\mbox{Lfdr}}_{i,jj'}\le\lambda_{jj'}\}\right|\}
\\\overset p \to&0.
\end{aligned}
$$
The Uniformly Law of Large Numbers (Lemma 3.1 in \citep{van2000empirical}) ensures that 
$$
\sup\limits_{\lambda\in [\lambda_0,1]^{D}}|\frac{1}{m}\sum\limits_{i=1}^m1\{{\mbox{Lfdr}}_{i,jj'}\le \lambda_{jj'},\forall j\ne j',1\le j,j'\le n\}-E(\delta(\lambda))|\to 0.
$$
Thus we have
$$
\sup\limits_{\lambda\in [\lambda_0,1]^{D}}|\frac{1}{m}\sum\limits_{i=1}^m1\{\widehat{\mbox{Lfdr}}_{i,jj'}\le \lambda_{jj'},\forall j\ne j',1\le j,j'\le n\}-E(\delta(\lambda))|\to 0.
$$

As a result, the continuous mapping theorem implies that
$$
\sup\limits_{\lambda\in [\lambda_0,1]^{D}}|\mathcal R_2(\lambda,\tau)-\mathcal R_3(\lambda,\tau)|\to 0.
$$

The remaining proof is to derive the asymptotic FDR bound.

For any $j\ne j',1\le j, j'\le n$, let 
$$
\begin{aligned}
\hat V_m(\lambda)&=\frac{1}{m}\sum\limits_{i=1}^m 1\{\widehat{\mbox{Lfdr}}_{i,jj'}\le \lambda_{jj'},\forall j\ne j',1\le j,j'\le n\}1\{i\in \mathcal H_0\}\\
V_m(\lambda)&=\frac{1}{m}\sum\limits_{i=1}^m 1\{\mbox{Lfdr}_{i,jj'}\le \lambda_{jj'},\forall j\ne j',1\le j,j'\le n\}1\{i\in \mathcal H_0\}\\
\hat R_m(\lambda)&=\frac{1}{m}\sum\limits_{i=1}^m 1\{\widehat{\mbox{Lfdr}}_{i,jj'}\le \lambda_{jj'},\forall j\ne j',1\le j,j'\le n\}\\
 R_m(\lambda)&=\frac{1}{m}\sum\limits_{i=1}^m 1\{\mbox{Lfdr}_{i,jj'}\le \lambda_{jj'},\forall j\ne j',1\le j,j'\le n\}.
\end{aligned}$$

By lemma \ref{lem: product_bound}, we have
$$
\begin{aligned}
   &|\hat V_m(\lambda)-V_m(\lambda)|\\=&|\frac{1}{m}\sum\limits_{i=1}^m 1\{i\in \mathcal H_0\}\prod\limits_{j,j',j\ne j'}(1\{\widehat{\mbox{Lfdr}}_{i,jj'}\le\lambda_{i,jj'}\})-\frac{1}{m}\sum\limits_{i=1}^m 1\{i\in \mathcal H_0\}\prod\limits_{j,j',j\ne j'}(1\{\mbox{Lfdr}_{i,jj'}\le\lambda_{i,jj'}\})| 
   \\\le& \frac{1}{m}\sum\limits_{i=1}^m\sum\limits_{j,j',j\ne j'}|1\{\widehat{\mbox{Lfdr}}_{i,jj'}\le\lambda_{i,jj'}\}-1\{\mbox{Lfdr}_{i,jj'}\le\lambda_{i,jj'}\}| .
\end{aligned}
$$
So that $\sup\limits_{\lambda\in [\lambda_0,1]^D}|\hat V_m(\lambda)-V_m(\lambda)|\to 0$. In the same way $\sup\limits_{\lambda\in [\lambda_0,1]^D}|\hat R_m(\lambda)-R_m(\lambda)|\to 0$. 

We also obtain that
$$
\begin{aligned}
&|\frac{\hat V_m(\lambda)}{\hat R_m(\lambda)}-\frac{V_m(\lambda)}{R_m(\lambda)}|
\\=&|\frac{\hat V_m(\lambda)R_m(\lambda)-V_m(\lambda)\hat R_m(\lambda)}{R_m(\lambda)\hat R_m(\lambda)}|
\\=&|\frac{\hat V_m(\lambda)R_m(\lambda)-\hat V_m(\lambda)\hat R_m(\lambda)+\hat V_m(\lambda)\hat R_m(\lambda)-V_m(\lambda)\hat R_m(\lambda)}{R_m(\lambda)\hat R_m(\lambda)}|
\\\le &|\frac{\hat V_m(\lambda)|R_m(\lambda)-\hat R_m(\lambda)|}{R_m(\lambda)\hat R_m(\lambda)}|+|\frac{\hat R_m(\lambda)|\hat V_m(\lambda)-V_m(\lambda)|}{R_m(\lambda)\hat R_m(\lambda)}|
\\\le &\frac{1}{R_m(\lambda)}\{\frac{\hat V_m(\lambda)}{\hat R_m(\lambda)}|R_m(\lambda)-\hat R_m(\lambda)|+|\hat V_m(\lambda)-V_m(\lambda)|\}.
\\\le &\frac{1}{R_m(\lambda)}\{|R_m(\lambda)-\hat R_m(\lambda)|+|\hat V_m(\lambda)-V_m(\lambda)|\}.
\end{aligned}
$$
Therefore, 

\begin{align}
&\sup\limits_{\lambda\in [\lambda_0,1]^D}|\frac{\hat V_m(\lambda)}{\hat R_m(\lambda)}-\frac{V_m(\lambda)}{R_m(\lambda)}| \nonumber
\\\le &\frac{1}{R_m\{(\lambda_0,\ldots, \lambda_0)\}}\{\sup\limits_{\lambda\in [\lambda_0,1]^D}|R_m(\lambda)-\hat R_m(\lambda)|+\sup\limits_{\lambda\in [\lambda_0,1]^D}|\hat V_m(\lambda)-V_m(\lambda)|\}\nonumber
\\\overset p \to&0. \label{eq:hatVhatR}
\end{align}

Finally, by (C4), we have $P[\hat R_m\{\hat\lambda(\hat\alpha_p)\}>0]\to 1$, and
$$
\begin{aligned}
\mbox{FDR}&=E\left(\frac{\hat V_m (\hat \lambda(\hat\alpha_p))}{\hat R_m(\hat\lambda(\hat\alpha_p))\vee1}\right)
\\&=E\{\frac{\hat V_m (\hat \lambda(\hat\alpha_p))}{\hat R_m(\hat\lambda(\hat\alpha_p))}-\frac{V_m (\hat \lambda(\hat\alpha_p))}{R_m(\hat\lambda(\hat\alpha_p))}\}+E\frac{ V_m (\hat \lambda(\hat\alpha_p))}{ R_m(\hat\lambda(\hat\alpha_p))}.
\end{aligned}
$$
Therefore, 

Based on the results
$$
\sup\limits_{\lambda\in [\lambda_0,1]^D}|\frac{V_m(\lambda)}{R_m(\lambda)}-\frac{E(\delta(\lambda)1\{\mathcal H_0\})}{E(\delta(\lambda))}|\to 0,
$$
$$
\frac{E(\delta(\lambda)1\{\mathcal H_0\})}{E(\delta(\lambda))}< \mathcal R_1(\lambda),
$$
$$
\mathcal R_1(\lambda)\le \mathcal R_2(\lambda,\tau), \lambda\in \Lambda,
$$
and
$$
\sup\limits_{\lambda\in [\lambda_0,1]^{D}}|\mathcal R_3(\lambda,\tau)-\mathcal R_2(\lambda,\tau)|\to 0,
$$
for a specific $\tau$, we have for any $\lambda\in \Lambda\cap[\lambda_0,1]^D$
$$
P(\sup\limits_{\lambda\in \Lambda\cap[\lambda_0,1]^D}\{\frac{V_m(\lambda)}{R_m(\lambda)}-R_3(\lambda,\tau)\}<0)\to 1.
$$

Finally, when $\hat\lambda (\hat \alpha_p)\in \Lambda$,  we have
$$
\begin{aligned}
\limsup\limits_{m\to\infty} FDR&=\limsup\limits_{m\to\infty}E\left(\frac{\hat V_m (\hat \lambda(\hat\alpha_p))}{\hat R_m(\hat\lambda(\hat\alpha_p))\vee1}\right)
\\&=\limsup\limits_{m\to\infty}E|\frac{\hat V_m (\hat \lambda(\hat\alpha_p))}{\hat R_m(\hat\lambda(\hat\alpha_p))}-\frac{V_m (\hat \lambda(\hat\alpha_p))}{R_m(\hat\lambda(\hat\alpha_p))}|+\limsup\limits_{m\to\infty}E\frac{ V_m (\hat \lambda(\hat\alpha_p))}{ R_m(\hat\lambda(\hat\alpha_p))}
\\&\le \lim\limits_{m\to\infty}E\sup\limits_{\lambda\in [\lambda_0,1]}|\frac{\hat V_m (\lambda)}{\hat R_m(\lambda)}-\frac{V_m (\lambda)}{R_m(\lambda)}|+\limsup\limits_{m\to\infty}\mathcal R_3(\hat\lambda(\hat\alpha_p))
\\&\le\alpha.
\end{aligned}
$$
The last inequality is based on (\ref{eq:hatVhatR}) and the definition of $\hat \lambda(\hat\alpha_p)$, which is defined in (\ref{eq:alpha_hat}).

\section{Competing methods}\label{sec:comparison_methods}
In the simulation studies, we compared the proposed method to several other methods for replicability analysis, including \textit{ad hoc} BH, MaxP \citep{benjamini2009selective}, IDR \citep{li2011measuring}, MaRR \citep{philtron2018maximum}, radjust \citep{bogomolov2018assessing}, AdaFilter \citep{wang2022detecting}, and JUMP \citep{lyu2023jump}. Let $(p_{i1}, p_{i2}), i=1,\dots,m$ denote the paired $p$-values from two studies. We review these competing methods as follows.

\subsection{The \textit{ad hoc} BH method}
BH \citep{benjamini1995controlling} is the most popular multiple testing procedure that conservatively controls the FDR for $m$ independent or positively correlated tests. In study $j, j=1,2$, the BH procedure proceeds as below:
\begin{itemize}
    \item \textit{Step 1}. Let $p_{j(1)}\le p_{j(2)}\le \dots \le p_{j(m)}$ be the ordered $p$-values in study $j$, and denote by $H^{(j)}_{(i)}$ the null hypothesis corresponding to $p_{j(i)}$;
    \item \textit{Step 2}. Find the largest $i$ such that $p_{j(i)}\le \frac{i}{m}\alpha$, i.e., $\hat k=\mbox{max}\{i\ge 1: p_{j(i)}\le \frac{i}{m}\alpha\}$, and $\hat k = 0$ if the set is empty;
    \item \textit{Step 3}. Reject all $H^{(j)}_{(i)}$ for $i=1,\dots,\hat k$.
\end{itemize}

The \textit{ad hoc} BH method for replicability analysis identifies features rejected by both studies as replicable signals.

\subsection{The MaxP method}
Define the maximum $p$-values as
\begin{equation*}
    q_i = \mbox{max}\{p_{i1}, p_{i2}\}, i=1,\dots,m.
\end{equation*}
The na{\"i}ve MaxP method directly applies BH \citep{benjamini1995controlling} to $q_i, i=1,\dots, m$ for FDR control of replicability analysis.

\subsection{The IDR procedure}
The IDR procedure \cite{li2011measuring} deals with high throughput experimental data from two studies. For feature $i,$ we have the bivariate observations $(x_{i1},x_{i2}), i=1,\dots,m$. It is assumed that $(x_{i1},x_{i2}), i=1,\dots,m$ consist of genuine signals (replicable signals across two studies) and spurious signals (non-replicable signals). Let $K_i$ denote whether the $i$th feature is a replicable signal ($K_i=1$) or not ($K_i=0$). It is assumed that $K_i, i =1, \ldots, m$ are independent and follow the Bernoulli distribution. Denote $\pi_1 = P(K_i=1)$. To induce dependence between $x_{i1}$ and $x_{i2}$, we use a copula model. Specifically, we assume that the observed data $(x_{i1}, x_{i2})$ are generated from latent variables $(z_{i1}, z_{i2})$. The latent variables 
\begin{align*}
    \begin{pmatrix}z_{i1} \\ z_{i2}\end{pmatrix}\Bigg| K_i=k &\sim N\begin{pmatrix}
        \begin{pmatrix}
            \mu_k \\ \mu_k
        \end{pmatrix},
        \sigma_k^2
        \begin{pmatrix}
            1 & \rho_k \\
            \rho_k & 1
        \end{pmatrix}
    \end{pmatrix}, \text{ } k=0,1,
\end{align*}
where $\mu_0=0, \mu_1>0, \sigma_0^2=1, \sigma_1^2>0, \rho_0=0,$ and $ 0<\rho_1\le 1$. The cumulative distribution function of $z_{ij}$ is
$$
G(x) = P(z_{ij} \le x) = \pi_1 \Phi\left(\frac{x-\mu_1}{\sigma_1}\right) + (1-\pi_1)\Phi(x). 
$$
Denote the marginal distribution function of $x_{ij}, i = 1, \ldots, m; j =1, 2,$ as $F_j.$ Generate 
$$
x_{ij} = F_j^{-1}(G(z_{ij})),  i=1,\dots,m; j=1,2.
$$
In this way, dependence across two studies is produced. 
To control the false discovery rate, we use the local irreproducible discovery rate (idr) as the test statistic, which is defined as the posterior probability of $K_i=0$ given $(x_{i1}, x_{i2}).$ Specifically,
\begin{align*}
    idr(x_{i1}, x_{i2})&:=P(K_i=0\mid x_{i1},x_{i2}) \\
    &= \frac{(1-\pi_1) h_0[G^{-1}\{ F_1(x_{i1})\}, G^{-1}\{F_2(x_{i2})\}]}{(1-\pi_1) h_0[G^{-1}\{ F_1(x_{i1})\}, G^{-1}\{F_2(x_{i2})\}] + \pi_1 h_1[G^{-1}\{ F_1(x_{i1})\}, G^{-1}\{F_2(x_{i2})\}]}.
\end{align*}
where 
$$
 h_k\sim N\begin{pmatrix}
        \begin{pmatrix}
             \mu_k \\  \mu_k
        \end{pmatrix},
        \sigma_k^2
        \begin{pmatrix}
            1 & \rho_k \\
            \rho_k & 1
        \end{pmatrix}
    \end{pmatrix}, \ k=0,1.
$$
The estimation of $(\pi_1, \mu_1, \sigma_1^2, \rho_1)$ and $(F_1, F_2)$ is through the EM algorithm \citep{dempster1977maximum}. The step-up procedure based on ordered idr can be used for FDR control \citep{sun2007oracle}. Specifically, 
let $idr_{(1)}\le \dots\le idr_{(m)}$ be the ranked $idr$ values, and denote $H_{(1)}, \dots, H_{(m)}$ as the corresponding hypotheses. Find $l=\mbox{max}\{i: i^{-1}\sum_{j=1}^i idr_j\le \alpha\}$, and reject all $H_{(i)}$ with $i=1,\dots,l$.

\subsection{The MaRR procedure}
The MaRR procedure \citep{philtron2018maximum} uses the maximum rank of each feature. The null hypothesis is that $H_{i0}: p_{i1} \text{ and } p_{i2}$ are irreproducible.
Denote $(R_{i1},R_{i2})$ as the ranks of $(p_{i1}, p_{i2}), i=1,\dots,m$ within each study. 
Define
$$
M_i = \mbox{max}\{R_{i1}, R_{i2}\}, i=1,\dots,m.
$$
Let $\pi_1$ denote the proportion of replicable signals. 
Under the assumptions: \\
(I1) 
if gene $g$ is
reproducible and gene $h$ is irreproducible
\begin{equation*}
    R_{1g} < R_{1h}, \quad R_{2g} < R_{2h};
\end{equation*}
(I2) the correlation between the ranks of the reproducible gene is non-negative;\\
(I3) the two ranks of the irreproducible gene are independent,\\
irreproducible ranks $R_{i1}$ and $R_{i2}$ are uniformly distributed between $\lfloor m\pi_1 \rfloor + 1$ and $m$.
Denote the conditional null survival function of $M_i/m$ as
\begin{align*}
    &S_{m, \pi_1}(x) = P(M_i/m > x\mid \text{ gene }i\text{ is irreproducible})\\
    =& 1 - P\left(R_{i1}/m \leq x, R_{i2}/m \leq x \mid \text{ gene }i\text{ is irreproducible}\right)\\
    =& 1 - \prod _{j=1}^2 P\left(R_{ij}/m \leq x \mid \text{ gene }i\text{ is irreproducible}\right)\\
    =& 
    \begin{cases}
        1,\qquad & x < \pi_1,\\
        1-\frac{(i_x - j_{\pi_1})^2}{(m - j_{\pi_1})^2}, \quad & \pi_1 \leq x \leq 1,
    \end{cases}
\end{align*}
where $i_x = \lfloor m x \rfloor$ and $j_{\pi_1} = \lfloor m\pi_1 \rfloor$.
The limiting conditional survival function of $M_i/m$ under the null is
\begin{align*}
    S_{m,\pi_1}(x) \to S_{\pi_1}(x) = 
    \begin{cases}
        1 & x<\pi_1 \\ 
        1-\frac{(x-\pi_1)^2}{(1-\pi_1)^2} & \pi_1\le x\le 1 \\
        0 & 1<x
    \end{cases}.
\end{align*}
The empirical survival function can be estimated by $\hat S_m(x) = \frac{1}{m}\sum_{i=1}^m I(M_i/m\ge x)$, $x\in(0,1)$. 
By strong law of large numbers and Bayesian formula, 
\begin{align*}
    \hat{S}_m (x) \rightarrow& P(M_i / m \geq x)\\
    =& (1-\pi_1) P(M_i / m \geq x\mid \text{ gene } i \text{ is irreproducible}) + \pi_1\times 0\\
    =& (1-\pi_1) S_{\pi_1} (x) \text{ for } x \in (\pi_1, 1).
\end{align*}
If we estimate $\pi_1$ by $i/m$, we can define the mean square error (MSE) as follows.
\begin{equation*}
    \mbox{MSE}(i/m) = (m-i)^{-1}\sum_{j=i}^m \left(\hat S_m(j/m) - (1-i/m)S_{i/m}(j/m)\right)^2.
\end{equation*}
$\hat{k}$ is chosen to minimize the MSE in the range between $0$ and $\lfloor 0.9 m \rfloor$.
\begin{align*}
    \hat k = \mathop{\arg\min}\limits_{i=0,1,\dots,\left \lfloor  0.9m\right \rfloor }\left\{\mbox{MSE}(i/m)
    \right\}.
\end{align*}
Thus $\hat{k}/m$ serves as a good estimation of $\pi_1$.
To control the FDR at level $\alpha$, the MaRR generates the rejection threshold as follows
    \begin{align*}
    \text{Define } \hat N = \max_{\hat k<i\le m}\left\{i: m\widehat{\mbox{FDR}}(i) = \frac{(i-\hat k)^2}{Q(i)(m-\hat k)}\le \alpha\right\},
\end{align*}
where $Q(i) = \sum_{j=1}^m I(M_j\le i)$. 
Reject all features associated with $M_i\le \hat N$. 
\cite{philtron2018maximum} relax assumption (I1) to (R1): $P(R_{1g} < R_{1h}) > 1/2$ and $P(R_{2g} < R_{2h})>1/2$, which is more plausible in practice.

\subsection{The radjust procedure}
The radjust procedure \citep{bogomolov2018assessing} 
works as follows,
\begin{itemize}
    \item \textit{Step 1}. For a pre-specified FDR level $\alpha$, compute
    $$
        R=\mbox{max}\left[r: \sum_{i\in\mathcal{S}_1\cap \mathcal{S}_2} I \left\{(p_{i1}, p_{i2})\le \left(\frac{r\alpha}{2|\mathcal S_2|}, \frac{r\alpha}{2|\mathcal S_1|}\right)\right\} = r\right],
    $$
    where $\mathcal S_j$ is the set of features pre-selected in study $j$ for $j = 1,2$. By default, it selects features with $p$-values less than or equal to $\alpha/2$.
    \item \textit{Step 2}. Declare replicable features as those 
    with indices in the set
    $$
        \mathcal{R} = \left\{i: (p_{i1}, p_{i2})\le \left(\frac{R\alpha}{2|\mathcal S_2|}, \frac{R\alpha}{2|S_1|}\right), i\in \mathcal S_1 \cap \mathcal S_2\right\}.
    $$
\end{itemize}

In this paper, we implement the adaptive version of the radjust procedure provided in \cite{bogomolov2018assessing} for comparison, which first estimates the fractions of true null hypotheses among the pre-selected features. The fractions in the two studies are estimated as follows.
\begin{equation}\label{prop_est}
    \hat\pi_0^{(1)} = \frac{1+\sum_{i\in \mathcal S_{2,\alpha}}I(p_{i1}>\alpha)}{|\mathcal S_{2,\alpha}|(1-\alpha)},\quad 
    \hat\pi_0^{(2)}= \frac{1+\sum_{i\in \mathcal S_{1,\alpha}}I(p_{i2}>\alpha)}{|\mathcal S_{1,\alpha}|(1-\alpha)},
\end{equation}
where $\mathcal S_{j,\alpha} = \mathcal S_j\cap \{1\le i\le m: p_{ij}\le \alpha\}, \ j=1,2$. The adaptive procedure with a nominal FDR level $\alpha$ 
works as follows.
\begin{itemize}
    \item \textit{Step 1}. Compute $\hat\pi_0^{(1)}$ and $\hat\pi_0^{(2)}$ using (\ref{prop_est}). Let
    $$
        R=\mbox{max}\left[r: \sum_{i\in\mathcal{S}_{1,\alpha}\cap \mathcal{S}_{2,\alpha}} I \left\{(p_{i1}, p_{i2})\le \left(\frac{r\alpha}{2|\mathcal S_{2,\alpha}|\hat\pi_0^{(1)}}, \frac{r\alpha}{2|\mathcal S_{1,\alpha}|\hat\pi_0^{(2)}}\right)\right\} = r\right],
    $$
    \item \textit{Step 2}. Reject features with indices in the set
    $$
        \mathcal{R} = \left\{i: (p_{i1}, p_{i2})\le \left(\frac{R\alpha}{2|\mathcal S_{2,\alpha}|\hat\pi_0^{(1)}}, \frac{R\alpha}{2|S_{1,\alpha}|\hat\pi_0^{(2)}}\right), i\in \mathcal S_{1,\alpha} \cap \mathcal S_{2,\alpha}\right\}.
    $$
\end{itemize}
\subsection{The AdaFilter procedure}
The AdaFilter \citep{wang2022detecting} concentrates on the partial conjunction hypothesis.
$$
H_0^{r / n} \text {: fewer than } r \text { out of } n \text { base hypotheses are non-null. }
$$
In the simulation, we concentrate on the case when $r=1$.
\begin{itemize}
    \item Step 1: Calculate the filtering and selection ``P-values”.
    $$
    \begin{aligned}
    & F_j:=(n-r+1) p_{j(r-1) }, \quad \text { and } \\
    & S_j:=(n-r+1) p_{j(r)},
    \end{aligned},
    $$
    where $$p_{j(1)}\le \cdots\le p_{j(n)}$$ is the ordered $p$-value under study $j$.
    \item Step 2: Reject $H_{0 j}^{r / n}$ if $S_j<\gamma_0^{BH}$ where
    $$
    \gamma_0^{BH}=\sup \left\{\gamma \in[0, \alpha] \left\lvert\, \frac{\gamma \sum_{j=1}^M 1_{F_j<\gamma}}{\sum_{j=1}^M 1_{S_j<\gamma} \vee 1} \leq \alpha\right.\right\} .
    $$
\end{itemize}
\subsection{The JUMP procedure}
\subsubsection{The original JUMP procedure applied to two studies}
The JUMP procedure \citep{lyu2023jump} works on the maximum of $p$-values across two studies. Define
$$
p_i^{\max} = \mbox{max}\{p_{i1}, p_{i2}\}, i=1,\dots,m.
$$
Let $s_i=(\theta_{i1}, \theta_{i2}),\ i=1,\dots,m$ denote the joint hidden states across two studies. Then $s_i\in \{(0, 0), (0, 1), (1, 0), (1, 1)\}$ with $\mathbb P(s_i=(k,l))=\xi_{kl}$ for $k,l=0,1$ and $\sum_{k,l}\xi_{kl}=1$. It can be shown that
\begin{align*}
&\mathbb P\left(p_i^{\max}\le t\mid H_{i0} \text{ is true}\right)\\
=& \frac{\xi_{00}\mathbb P(p_i^{\max}\le t\mid s_i=(0,0))}{\xi_{00}+\xi_{01}+\xi_{10}} + \frac{\xi_{01}\mathbb P(p_i^{\max}\le t\mid s_i=(0,1))}{\xi_{00}+\xi_{01}+\xi_{10}} + \frac{\xi_{10}\mathbb P(p_i^{\max}\le t\mid s_i=(1,0))}{\xi_{00}+\xi_{01}+\xi_{10}}\\
\le& \frac{\xi_{00}t^2 + (\xi_{01}+\xi_{10})t}{\xi_{00} + \xi_{01} + \xi_{10}}\le t,
\end{align*}
which means that $p_i^{\max}$ follows a super-uniform distribution under the replicability null. Denote
$$
G(t) = \frac{\xi_{00}t^2 + (\xi_{01}+\xi_{10})t}{\xi_{00} + \xi_{01} + \xi_{10}}.
$$

For a given threshold $t\in (0,1)$, a conservative estimate of the FDR 
is obtained by
$$
\mbox{FDR}^*(t)=\frac{m(\xi_{00}+\xi_{01}+\xi_{10})G(t)}{\sum_{i=1}^mI\{p_i^{\max}\le t\}\vee 1}.
$$

Following \cite{storey2002direct, storey2004strong}, the proportion of null hypotheses in study $j$ can be estimated by
$$
\hat\pi_0^{(j)}(\lambda_j)=\frac{\sum_{i=1}^m I\{p_{ij}\ge \lambda_j\}}{m(1-\lambda_j)}, \quad j=1,2.
$$
Similarly, $\xi_{00}$ is estimated by
$$
\hat\xi_{00}(\lambda_3) = \frac{\sum_{i=1}^m I\{p_{i1}\ge \lambda_3, p_{i2}\ge \lambda_3\}}{m(1-\lambda_3)^2},
$$
where $\lambda_1, \lambda_2$ and $\lambda_3$ are tuning parameters that can be selected by using the smoothing method provided in \cite{storey2003statistical}. Then we have
$$
\hat\xi_{01} = \hat\pi_0^{(1)}-\hat\xi_{00}, \quad \hat\xi_{10} = \hat\pi_0^{(2)} - \hat\xi_{00}.
$$

With these estimates, we have a plug-in estimate of FDR,
$$
\widehat{\mbox{FDR}}^*(t) = \frac{m(\hat\xi_{00}t^2+\hat\xi_{01}t+\hat\xi_{10}t)}{\sum_{i=1}^mI\{p_i^{\max}\le t\}\vee 1}.
$$
The JUMP procedure works as follows.
\begin{itemize}
    \item \textit{Step 1}. Let $p_{(1)}^{\max}\le \dots \le p_{(m)}^{\max}$ be the ordered maximum of $p$-values and denote by $H_{(i)}$ the corresponding hypothesis;
    \item \textit{Step 2}. Find the largest $k$ such that the estimated FDR is controlled, i.e.,
    $$
    \hat k=\max\{1\le k\le m: \widehat{\mbox{FDR}}^*(p_{(k)}^{\max})\le \alpha\};
    $$
    \item \textit{Step 3}. Reject $H_{(i)},$ 
    $i=1,\dots,\hat k$.
\end{itemize}

\subsubsection{Extending the JUMP procedure to three studies}\label{extend_jump}
When we have three studies with $p$-values, $(p_{i1}, p_{i2}, p_{i3}), i=1,\dots,m$, define the maximum of $p$-values as
$$
p_i^{\max} = \max\{p_{i1}, p_{i2}, p_{i3}\}, i=1,\dots,m.
$$
Let $s_i = (\theta_{i1}, \theta_{i2}, \theta_{i3}), i=1,\dots,m$ denote the joint hidden states across three studies with prior probabilities $P(s_i=(k, l, r))=\xi_{klr}$, where $k,l,r=0,1$ and $\sum_{k,l,r}\xi_{klr}=1$, $s_i \in \{(0, 0, 0), (0, 0, 1), (0, 1, 0), (0, 1, 1), (1, 0, 0), (1, 0, 1), (1, 1, 0), (1, 1, 1)\}$. The replicability null hypothesis of three studies is given by
\begin{eqnarray*}\label{null3}
    H_{i0}: s_i \in \mathbb H = \{(0, 0, 0), (0, 0, 1), (0, 1, 0), (0, 1, 1), (1, 0, 0), (1, 0, 1), (1, 1, 0)\}, i=1,\dots,m.
\end{eqnarray*}
It can be shown that
\begin{align*}
\mathbb P(p_i^{\max}\le t|H_{0i} \text{ is true}) &= \frac{\mathbb P(p_i^{\max}\le t, H_{i0} \text{ is true})}{\mathbb P(H_{i0} \text{ is true})}\\
&= \frac{\sum_{(k,l,r)\in\mathbb H}\xi_{klr}\mathbb P(p_i^{\max}\le t|s_i = (k,l,r))}{\sum_{(k,l,r)\in \mathbb H}\xi_{klr}}\\
&\le \frac{\xi_{000}t^3 + (\xi_{001} + \xi_{010} + \xi_{100})t^2 + (\xi_{011} + \xi_{101} + \xi_{110})t}{\sum_{(k,l,r)\in\mathbb H}\xi_{klr}}\le t,
\end{align*}
which means the $p_i^{\max}$ follows a super-uniform distribution under the replicability null. Denote
$$
G(t) = \xi_{000}t^3 + (\xi_{001} + \xi_{010} + \xi_{100})t^2 + (\xi_{011} + \xi_{101} + \xi_{110})t.
$$

For a given threshold $t\in(0,1)$, a conservative estimate of the FDR is obtained by
$$
\mbox{FDR}^*(t) = \frac{mG(t)}{\sum_{i=1}^mI\{p_i^{\max}\le t\}\vee 1}.
$$

Following \citep{storey2002direct, storey2004strong}, the proportion of null hypotheses in study $j$ can be estimated by
$$
\hat\pi_0^{(j)}(\lambda_j)=\frac{\sum_{i=1}^m I\{p_{ij}\ge \lambda_j\}}{m(1-\lambda_j)}, \quad j=1,2,3.
$$
Similarly,
\begin{align*}
    \hat\xi_{000}(\lambda_4) + \hat\xi_{001}(\lambda_4) = \frac{\#\{p_{i1}\ge \lambda_4, p_{i2}\ge\lambda_4, i=1,\dots,m\}}{m(1-\lambda_4)^2}\\
\hat\xi_{000}(\lambda_5) + \hat\xi_{010}(\lambda_5) = \frac{\#\{p_{i1}\ge \lambda_5, p_{i3}\ge\lambda_5, i=1,\dots,m\}}{m(1-\lambda_5)^2}\\
\hat\xi_{000}(\lambda_6) + \hat\xi_{100}(\lambda_6) = \frac{\#\{p_{i2}\ge \lambda_6, p_{i3}\ge\lambda_6, i=1,\dots,m\}}{m(1-\lambda_6)^2}\\
\hat\xi_{000}(\lambda_7) = \frac{\#\{p_{i1}\ge \lambda_7, p_{i2}\ge\lambda_7, p_{i3}\ge\lambda_7, i=1,\dots,m\}}{m(1-\lambda_7)^3},
\end{align*}
where $\lambda_1,\dots,\lambda_7$ are tuning parameters that can be selected by using the smoothing method provided in \citep{storey2003statistical}. Then we have
\begin{align*}
    \hat\xi_{011} = \hat\pi_{0}^{(1)} - \hat\xi_{000} - \hat\xi_{001} - \hat\xi_{010},\\
    \hat\xi_{101} = \hat\pi_{0}^{(2)} - \hat\xi_{000} - \hat\xi_{001} - \hat\xi_{100},\\
    \hat\xi_{110} = \hat\pi_{0}^{(3)} - \hat\xi_{000} - \hat\xi_{010} - \hat\xi_{100}.
\end{align*}

With these esitmates, we have a plug-in estimate of FDR,
$$
\widehat{\mbox{FDR}}^*(t) = \frac{m\hat G(t)}{\sum_{i=1}^mI\{p_i^{\max}\le t\}\vee 1},
$$
where $\hat G(t) = \hat\xi_{000}t^3 + (\hat\xi_{001} + \hat\xi_{010} + \hat\xi_{100})t^2 + (\hat\xi_{011} + \hat\xi_{101} + \hat\xi_{110})t$.

Then the JUMP procedure can be applied to control the FDR for the replicability analysis across three studies.

\subsection{Adjustment to our method}
In the simulations, we 
estimate $(\xi_{00},\xi_{01},\xi_{10},\xi_{11})$ with the conservative estimates proposed in \cite{lyu2023jump} due to better finite sample performance. Specifically, we estimate the proportion of null hypotheses in study $j$, $\pi^{(j)}_0$, as
\begin{align*}
    \hat\pi_{0}^{(j)}(\lambda_j) = \frac{\sum_{i=1}^mI\{p_{ij}\ge \lambda_j\}}{m(1-\lambda_j)}, j=1,2,
\end{align*}
and estimate $\xi_{00}$ as
\begin{align*}
    \hat\xi_{00}(\lambda_3) = \frac{\sum_{i=1}^mI\{p_{i1}\ge \lambda_3, p_{i2}\ge\lambda_3\}}{m(1-\lambda_3)^2},
\end{align*}
where $\lambda_1, \lambda_2$ and $\lambda_3$ are tuning parameters. We use the smoothing method provided in \cite{storey2003statistical} to select the tuning parameters. 
We have
$\hat\xi_{01} = \hat\pi_0^{(1)}-\hat\xi_{00}, 
\hat\xi_{10} = \hat\pi_0^{(2)}-\hat\xi_{00}$
and $\hat{\xi}_{11} = 1 - \hat{\xi}_{00} - \hat{\xi}_{01} - \hat{\xi}_{10}.$

\end{document}